\documentclass[cmp]{svjour}
\pdfoutput=1

\usepackage[unicode=true,bookmarksopen=true,colorlinks=true,allcolors=blue,linktocpage=true,pdfa=true]{hyperref}

\usepackage{amsmath}
\usepackage{amsfonts}

\usepackage{mathtools}
\usepackage{lmodern}
\usepackage[T1]{fontenc}
\usepackage{bbm}
\DeclareMathAlphabet{\mathbfi}{OML}{cmm}{b}{it}

\let\originalleft\left
\let\originalright\right
\renewcommand{\left}{\mathopen{}\mathclose\bgroup\originalleft}
\renewcommand{\right}{\aftergroup\egroup\originalright}

\makeatletter
\newenvironment{equations}[1][]{\subequations\ifx\relax#1\relax\else\label{#1}\fi\align\ignorespaces}{\endalign\ignorespacesafterend\endsubequations}
\def\@spliteq#1{\begin{equation}\begin{split}#1\end{split}\end{equation}}
\def\@spliteqstar#1{\begin{equation*}\begin{split}#1\end{split}\end{equation*}}
\def\splitequation{\collect@body\@spliteq}
\expandafter\def\csname splitequation*\endcsname{\collect@body\@spliteqstar}

\expandafter\def\csname endsplitequation*\endcsname{\ignorespacesafterend}
\makeatother

\makeatletter
\def\normalthmheadings{\def\@spbegintheorem##1##2##3##4{\trivlist
                 \item[\hskip\labelsep{##3##1\ ##2\@thmcounterend}]##4\addcontentsline{toc}{subsection}{##1\ ##2\@thmcounterend}}
\def\@spopargbegintheorem##1##2##3##4##5{\trivlist
      \item[\hskip\labelsep{##4##1\ ##2}]{##4(##3)\@thmcounterend\ }##5}}
\normalthmheadings
\def\reversethmheadings{\def\@spbegintheorem##1##2##3##4{\trivlist
                 \item[\hskip\labelsep{##3##2\ ##1\@thmcounterend}]##4\addcontentsline{toc}{subsection}{##2\ ##1\@thmcounterend}}
\def\@spopargbegintheorem##1##2##3##4##5{\trivlist
      \item[\hskip\labelsep{##4##2\ ##1}]{##4(##3)\@thmcounterend\ }##5}}
\makeatother
\spnewtheorem*{remark*}{Remark}{\itshape}{\rmfamily}

\spnewtheorem*{example*}{Example}{\itshape}{\rmfamily}

\let\oldendproof\endproof
\def\endproof{\hfill\squareforqed\oldendproof}

\newcommand{\mathe}{\mathrm{e}}
\newcommand{\mathi}{\mathrm{i}}
\newcommand{\total}{\mathop{}\!\mathrm{d}}

\newcommand{\abs}[1]{{\left\lvert{#1}\right\rvert}}
\newcommand{\unitmatrix}{\mathbbm{1}}
\newcommand{\eqend}[1]{\,#1}
\newcommand{\bigo}[1]{\mathcal{O}\left({#1}\right)}
\newcommand{\brst}{\mathop{}\!\mathsf{s}\hskip 0.05em\relax}
\newcommand{\st}{\mathop{}\!\hat{\mathsf{s}}\hskip 0.05em\relax}
\newcommand{\stq}{\mathop{}\!\mathsf{q}\hskip 0.05em\relax}

\newcommand{\normord}[1]{\mathopen{:}{#1}\mathclose{:}}
\DeclareMathOperator{\supp}{supp}

\usepackage{cite}

\allowdisplaybreaks

\usepackage{calc}
\usepackage{tikz}
\usetikzlibrary{shapes.misc,positioning,patterns,decorations,decorations.pathmorphing,intersections,decorations.pathreplacing}

\journalname{Communications in Mathematical Physics}

\begin{document}

\title{Anomalies in time-ordered products and applications to the BV--BRST formulation of quantum gauge theories}
\titlerunning{Anomalies in time-ordered products and quantum gauge theories}

\author{Markus B. Fr{\"o}b}
\institute{Institut f{\"u}r Theoretische Physik, Universit{\"a}t Leipzig,\\ Br{\"u}derstra{\ss}e 16, 04103 Leipzig, Germany\\ \email{mfroeb@itp.uni-leipzig.de}}
\authorrunning{M. B. Fr{\"o}b}

\date{17. July 2019}

\maketitle
\begin{abstract}
We show that every (graded) derivation on the algebra of free quantum fields and their Wick powers in curved spacetimes gives rise to a set of anomalous Ward identities for time-ordered products, with an explicit formula for their classical limit. We study these identities for the Koszul--Tate and the full BRST differential in the BV--BRST formulation of perturbatively interacting quantum gauge theories, and clarify the relation to previous results. In particular, we show that the quantum BRST differential, the quantum antibracket and the higher-order anomalies form an $L_\infty$ algebra. The defining relations of this algebra ensure that the gauge structure is well-defined on cohomology classes of the quantum BRST operator, i.e., observables. Furthermore, we show that one can determine contact terms such that also the interacting time-ordered products of multiple interacting fields are well defined on cohomology classes. An important technical improvement over previous treatments is the fact that all our relations hold off-shell and are independent of the concrete form of the Lagrangian, including the case of open gauge algebras.
\end{abstract}

\section{Introduction}
\label{sec_intro}

Studying quantum fields on curved spacetimes has lead to the discovery of many interesting effects, including the Hawking radiation of black holes~\cite{hawking1975}, the Unruh effect for accelerated observers~\cite{unruh1976,crispinohiguchimatsas2008} and, treating (quantum) gravity as an effective field theory of quantised metric perturbations, anisotropies in the cosmic microwave background which have been experimentally observed~\cite{planck2015a,planck2015b,planck2015c}. A mathematically sound framework for this study is locally covariant algebraic quantum field theory~\cite{brunettifredenhagenverch2003}. In the algebraic approach, instead of directly studying expectation values or matrix elements, one first constructs the interacting field algebra, which includes renormalisation. Imposing covariance under diffeomorphisms severely restricts the renormalisation freedom~\cite{hollandswald2001}, and in particular the renormalisation ambiguities in an arbitrary globally hyperbolic spacetime are the same as in flat spacetime (i.e., constants multiplying covariant interaction terms that are determined by power counting), with the only difference that curvature-dependent terms do appear. This was first shown for scalar fields in~\cite{hollandswald2001,hollandswald2002}, and later extended to Yang--Mills theories~\cite{hollands2008,hollands_rev} and general gauge theories in the BV--BRST framework~\cite{fredenhagenrejzner2013,rejzner2015} (based on earlier work on gauge theories in the algebraic framework~\cite{duetschboas2002,duetsch2005,brenneckeduetsch2008}), and also applied to superconformal theories~\cite{demedeiroshollands2013,taslimitehrani2017b} and perturbative quantum gravity~\cite{brunettifredenhagenrejzner2016,brunettietal2016}.

However, there are still questions that have not found a fully satisfactory answer. In particular, the (anomalous) Ward identities that express gauge invariance at the quantum level, derived in~\cite{hollands2008,hollands_rev}, have only been shown to hold for closed gauge algebras, where the commutator of two gauge transformations gives another gauge transformation. Nevertheless, in~\cite{demedeiroshollands2013,taslimitehrani2017b} it was applied to supersymmetric theories whose gauge algebra is open (i.e., the commutator of two gauge transformations contains an additional term proportional to the equations of motion), and while one would expect the same (or at least very similar) anomalous Ward identities to hold also in this case, the proof given in~\cite{hollands2008,hollands_rev} does not generalise easily. On the other hand, the treatment of~\cite{fredenhagenrejzner2013,rejzner2015} does not depend on the concrete form of the Lagrangian and thus also includes the general case of an open gauge algebra (although for some of the statements a closed gauge algebra was assumed), but did not treat explicitly the anomalies appearing in the construction of interacting observables. Moreover, Ward identities for (interacting) time-ordered products of interacting observables have not been derived at all, but these are necessary, e.g., for the observables in perturbative quantum gravity proposed in~\cite{brunettietal2016,froeb2018,froeblima2018}.

In this article, we give an answer to these open questions, and in fact relatively simple proofs starting from the realisation that anomalous Ward identities arise as a general feature in the construction of time-ordered products, for any derivation and not only the BRST differential. Our results are:
\begin{itemize}
\item The existence of a set of anomalous Ward identities for every graded derivation on the free-field algebra, Theorem~\ref{thm_anomward} (page~\pageref{thm_anomward}).
\item An explicit formula for the terms in the classical limit of these identities, Theorem~\ref{thm_classanom} (page~\pageref{thm_classanom}).
\item The particularisation of these theorems to the antibracket in the BV--BRST formalism, Theorems~\ref{thm_classanom_bv} and \ref{thm_freebrstward} (pages~\pageref{thm_classanom_bv} and~\pageref{thm_freebrstward}).
\item The (anomalous) Ward identities in perturbatively interacting gauge theories, Theorem~\ref{thm_brstward} (page~\pageref{thm_brstward}).
\item The $L_\infty$ structure of the quantum BRST operator, the quantum antibracket and higher-order anomalies, Theorem~\ref{thm_linfty} (page~\pageref{thm_linfty}).
\item The extension of classical observables to the quantum theory, and time-ordered products thereof, Theorem~\ref{thm_obs} (page~\pageref{thm_obs}).
\end{itemize}
We also give simplified proofs for some existing theorems, or (parts of) proofs that have not been given explicitly yet, and that are needed for the above results:
\begin{itemize}
\item The extra factors of $\hbar$ contained in connected time-ordered products, Theorem~\ref{thm_tc_hbar} (page~\pageref{thm_tc_hbar}).
\item The causal factorisation of interacting time-ordered products, Theorem~\ref{thm_tint_factor} (page~\pageref{thm_tint_factor}).
\item The vanishing of the anomaly in perturbatively interacting gauge theories for antifields, Theorem~\ref{thm_anom_antifields} (page~\pageref{thm_anom_antifields}).
\item That the terms violating perturbative agreement (linear in fields) are of the correct order in $\hbar$ to be removable by a field redefinition, Theorem~\ref{thm_pertagreement} (page~\pageref{thm_pertagreement}).
\end{itemize}

\section{Perturbative algebraic quantum field theory}
\label{sec_aqft}

Since extensive reviews of pAQFT exist~\cite{fredenhagenrejzner2012,hollandswald2015,fewsterverch2015,fredenhagenrejzner2016,hack2016}, and we are mainly interested in algebraic relations in the rest of the work, we give here only a summary to establish all necessary formulas. One starts with the construction of the free-field algebra on a globally hyperbolic manifold $(M,g)$ (i.e., where the Cauchy problem is globally well-posed). Assume thus given a set of dynamical fields $\{ \phi_K \}$, where the index $K$ distinguishes the type of field, and Lorentz, spinor and Lie algebra indices when necessary. Mathematically, these are sections of a vector bundle over $M$, obtained as the direct sum of the bundles corresponding to the individual fields. Repeated indices $K,L,\ldots$ are summed over, regardless of their position, and we denote by $\epsilon_K \in \{0,1\}$ the Grassmann parity of $\phi_K$. Assume further given an action $S = S_0 + S_\text{int}$, where the free action $S_0$ is at most quadratic in the fields and the interaction $S_\text{int}$ is at least cubic in the fields. We assume that $S$, $S_0$ and $S_\text{int}$ are bosonic (Grassmann even), and that $S_0$ does not contain terms linear in the fields, which is the case when any background fields (such as the spacetime metric $g$) fulfil their respective equations of motion. One can then write
\begin{equation}
\label{sec_aqft_action}
S_0 = \frac{1}{2} \int \phi_K(x) P_{KL}(x) \phi_L(x) \total x
\end{equation}
with the canonical (metric) volume form $\total x$, where $P_{KL}$ is a formally self-adjoint differential operator, $P_{KL}^* = (-1)^{\epsilon_K \epsilon_L} P_{LK}$. We further assume that $P_{KL} = 0$ if $\epsilon_K \neq \epsilon_L$ (i.e., $P_{KL}$ is itself a bosonic operator), and that it possesses unique retarded and advanced Green's functions:
\begin{equation}
\label{sec_aqft_retadv}
P_{KL}(x) G^\text{ret/adv}_{LM}(x,y) = \delta_{KM} \delta(x,y) = P^*_{LK}(y) G^\text{ret/adv}_{ML}(x,y)
\end{equation}
with $\delta(x,y)$ defined such that $\int f(x) \delta(x,y) \total x = f(y)$. The support of the retarded and advanced Green's functions fulfils
\begin{equation}
\supp \int G^\text{ret/adv}_{KL}(x,y) f_L(y) \total y \subset J^\pm\left( \supp f_K \right)
\end{equation}
for any test function $f_K$, where $J^+(S)$ [$J^-(S)$] is the causal future (past) of a set $S$. As test functions, we take here and in the following smooth and compactly supported sections of the vector bundle dual to the bundle of fields; in the scalar case, these are just smooth and compactly supported functions on $M$: $f \in \mathcal{C}^\infty_\text{c}(M)$. On a globally hyperbolic manifold, existence and uniqueness of the advanced and retarded Green's functions is guaranteed for normally hyperbolic operators (see, e.g.~\cite{baerginouxpfaeffle2007}), such as $P = \nabla^2 - m^2$ for the case of a single scalar field $\{ \phi_K \} = \{ \phi \}$ of mass $m$, or for a massless vector field in Feynman gauge $\{ \phi_K \} = \{ A_\mu \}$, where $P_{\mu\nu} = g_{\mu\nu} \nabla^2 - R_{\mu\nu}$. However, one may also be interested in other gauges, where the differential operator is not normally hyperbolic, but where nevertheless unique retarded and advanced Green's functions can be constructed~\cite{froebtaslimitehrani2018}. The retarded and advanced Green's function are related to each other by
\begin{equation}
\label{sec_aqft_retadv_rel}
G^\text{adv}_{KL}(x,y) = (-1)^{\epsilon_K \epsilon_L} G^\text{ret}_{LK}(y,x) \eqend{.}
\end{equation}

From the retarded and advanced Green's functions, one defines the Pauli--Jordan (or commutator) function
\begin{equation}
\label{sec_aqft_delta_def}
\Delta_{KL}(x,y) \equiv G^\text{ret}_{KL}(x,y) - G^\text{adv}_{KL}(x,y) = G^\text{ret}_{KL}(x,y) - (-1)^{\epsilon_K \epsilon_L} G^\text{ret}_{LK}(y,x) \eqend{.}
\end{equation}
We then define the algebra $\mathfrak{A}_0$ as the free $\ast$-algebra generated by the expressions $\phi_K(f_K)$, where $f_K$ is a test function, with the product denoted by $\star_\hbar$, the $\ast$-relation given by $\left[ \phi_K(f_K) \right]^\ast = \phi_K^\dagger(f_K^*)$, unit element $\unitmatrix$, and factored by the (anti-)commutation relation
\begin{splitequation}
\label{sec_aqft_commutator}
\left[ \phi_K(f_K), \phi_L(g_L) \right]_{\star_\hbar} &\equiv \phi_K(f_K) \star_\hbar \phi_L(g_L) - (-1)^{\epsilon_K \epsilon_L} \phi_L(g_L) \star_\hbar \phi_K(f_K) \\
&= \mathi \hbar \int f_K(x) \Delta_{KL}(x,y) g_L(y) \total x \total y \, \unitmatrix \equiv \mathi \hbar \Delta_{KL}(f_K,g_L) \, \unitmatrix \eqend{.}
\end{splitequation}
One can complete $\mathfrak{A}_0$ in the H{\"o}rmander (weak) topology~\cite{hoermander,hollands2008,brunettiduetschfredenhagen2009,fredenhagenrejzner2012b}, obtaining the \emph{free-field algebra} $\overline{\mathfrak{A}}_0$. In practice, this completion can be obtained by considering fixed two-point functions $G^+_{KL}(x,y)$ of Hadamard form, which are bisolutions
\begin{equation}
\label{sec_aqft_bisolutions}
P_{KL}(x) G^+_{LM}(x,y) = 0 = P^*_{LK}(y) G^+_{ML}(x,y)
\end{equation}
satisfying
\begin{equation}
\label{sec_aqft_bisolution_commutator}
G^+_{KL}(x,y) - (-1)^{\epsilon_K \epsilon_L} G^+_{LK}(y,x) = \Delta_{KL}(x,y)
\end{equation}
and a certain wave front set condition (microlocal spectrum condition)~\cite{radzikowski1996,brunettifredenhagenkoehler1996,brunettifredenhagen2000}. Since $P_{KL}$ was assumed to be a bosonic operator, we also have $G^+_{KL} = 0$ if $\epsilon_K \neq \epsilon_L$. One then defines a new set of generators of $\mathfrak{A}_0$, called \emph{normal-ordered products} and denoted by $\normord{\phi_{K_1} \cdots \phi_{K_n}}_G(f_{K_1} \otimes \cdots \otimes f_{K_n})$, and their integral kernels $\normord{\phi_{K_1}(x_1) \cdots \phi_{K_n}(x_n)}_G$ by
\begin{equation}
\normord{\phi_K}_G(f_K) \equiv \phi_K(f_K) \eqend{,}
\end{equation}
and then inductively such that
\begin{splitequation}
\label{sec_aqft_normord_product}
&\normord{\phi_{K_1}(x_1) \cdots \phi_{K_n}(x_n)}_G \star_\hbar \normord{\phi_{L_1}(y_1) \cdots \phi_{L_m}(y_m)}_G \\
&\quad= \normord{ \phi_{K_1}(x_1) \cdots \phi_{K_n}(x_n) \exp\left( \mathi \hbar \overleftrightarrow{G} \right) \phi_{L_1}(y_1) \cdots \phi_{L_m}(y_m) }_G \eqend{,}
\end{splitequation}
holds with
\begin{equation}
\overleftrightarrow{G} \equiv \int \overleftarrow{\frac{\delta_\text{R}}{\delta \phi_M(u)}} G^+_{MN}(u,v) \overrightarrow{\frac{\delta_\text{L}}{\delta \phi_N(v)}} \total u \total v \eqend{.}
\end{equation}
In this expression, the functional derivatives formally act to the left and right like on classical fields, which we indicate by the arrows. Moreover, since we work with fields of odd Grassmann parity, we need to distinguish left and right functional derivatives (denoted by the subscripts L and R), which satisfy graded Leibniz and commutation rules and can be exchanged according to the following:
\begin{equations}[sec_aqft_leftrightder_props]
\frac{\delta_\text{L}}{\delta \phi_K(x)} \left( F G \right) &= \frac{\delta_\text{L} F}{\delta \phi_K(x)} G + (-1)^{\epsilon_K \epsilon_F} F \frac{\delta_\text{L} G}{\delta \phi_K(x)} \eqend{,} \\
\frac{\delta_\text{R}}{\delta \phi_K(x)} \left( F G \right) &= (-1)^{\epsilon_K \epsilon_G} \frac{\delta_\text{R} F}{\delta \phi_K(x)} G + F \frac{\delta_\text{R} G}{\delta \phi_K(x)} \eqend{,} \\
\frac{\delta_\text{L} F}{\delta \phi_K} &= (-1)^{(1+\epsilon_F) \epsilon_K} \frac{\delta_\text{R} F}{\delta \phi_K} \eqend{,} \\
\frac{\delta_{\text{L}/\text{R}}^2 F}{\delta \phi_K \delta \phi_L} &= (-1)^{\epsilon_K \epsilon_L} \frac{\delta_{\text{L}/\text{R}}^2 F}{\delta \phi_L \delta \phi_K} \eqend{,} \qquad \frac{\delta_\text{L} \delta_\text{R} F}{\delta \phi_K \delta \phi_L} = \frac{\delta_\text{R} \delta_\text{L} F}{\delta \phi_L \delta \phi_K} \eqend{.}
\end{equations}
As an example, for two fields one obtains
\begin{splitequation}
\normord{ \phi_K \phi_L }_G\left( f_K \otimes g_L \right) &= \frac{1}{2} \left[ \phi_K(f_K) \star_\hbar \phi_L(g_L) - \mathi \hbar G^+_{KL}(f_K,g_L) \right] \\
&\quad+ \frac{1}{2} (-1)^{\epsilon_K \epsilon_L} \left[ \phi_L(g_L) \star_\hbar \phi_K(f_K) - \mathi \hbar G^+_{LK}(g_L,f_K) \right] \eqend{.}
\end{splitequation}
One can then consider the limit (in the weak topology) where the tensor product of test functions $f_{K_1} \otimes \cdots \otimes f_{K_n}$ tends to a distribution; due to the microlocal spectrum condition on the two-point functions $G^+_{KL}$ the product~\eqref{sec_aqft_normord_product} remains well-defined if one also imposes a certain wave front set condition on the limiting distribution. In particular, this includes the \emph{Wick monomials} for which $f_{K_1}(x_1) \otimes \cdots \otimes f_{K_n}(x_n) \to f_{K_1 \cdots K_n}(x_1) \delta(x_1,\ldots,x_n)$~\cite{brunettifredenhagenkoehler1996,hollandswald2001}.

Functional derivatives can also be defined for elements of $\mathfrak{A}_0$ by setting
\begin{equation}
\label{sec_aqft_funcder}
\frac{\delta_\text{L}}{\delta \phi_L(x)} \phi_K(f_K) = \frac{\delta_\text{R}}{\delta \phi_L(x)} \phi_K(f_K) \equiv f_K(x) \delta_{KL} \unitmatrix \eqend{,}
\end{equation}
which is compatible with the (anti-)commutation relation~\eqref{sec_aqft_commutator}, and requiring the functional derivatives to satisfy linearity and the graded Leibniz rule~\eqref{sec_aqft_leftrightder_props}, with the commuting product appearing there replaced by the star product. This functional derivative is compatible with the normal-ordered products, in the sense that
\begin{splitequation}
\label{sec_aqft_funcder_normalordered}
&\frac{\delta_\text{L}}{\delta \phi_L(x)} \normord{\phi_{K_1} \cdots \phi_{K_n}}_G(f_{K_1} \otimes \cdots \otimes f_{K_n}) = \sum_{k=1}^n f_{K_k}(x) \delta_{K_k L} (-1)^{\epsilon_L \sum_{m=1}^{k-1} \epsilon_{K_m}} \\
&\hspace{6em}\times \normord{\phi_{K_1} \cdots \hat{\phi}_{K_k} \cdots \phi_{K_n}}_G(f_{K_1} \otimes \cdots \otimes \hat{f}_{K_k} \otimes \cdots \otimes f_{K_n})
\raisetag{1.3em}
\end{splitequation}
holds (where the hat denotes omission), which can be easily proven by induction from relation~\eqref{sec_aqft_normord_product} for the respective integral kernels. The functional derivative can thus be extended straightforwardly to the completion $\overline{\mathfrak{A}}_0$. It follows by induction from the (anti-)commutator of two basic fields~\eqref{sec_aqft_commutator} and the graded Leibniz rule for the (anti-)commutator that the (anti-)commutator of two elements $F,G \in \overline{\mathfrak{A}}_0$ can be written as
\begin{equation}
\label{sec_aqft_commutator_general}
\left[ F, G \right]_{\star_\hbar} = F \star_\hbar \left[ \exp\left( \mathi \hbar \overleftrightarrow{\Delta} \right) - 1 \right] \star_\hbar G
\end{equation}
with
\begin{equation}
\overleftrightarrow{\Delta} \equiv \iint \overleftarrow{\frac{\delta_\text{R}}{\delta \phi_M(u)}} \Delta_{MN}(u,v) \overrightarrow{\frac{\delta_\text{L}}{\delta \phi_N(v)}} \total u \total v \eqend{,}
\end{equation}
and the functional derivatives act to the terms to the left or right of them only. In fact, this follows directly from the definition of the normal-ordered products~\eqref{sec_aqft_normord_product} using that the functional derivatives are compatible with the normal-ordered products~\eqref{sec_aqft_funcder_normalordered} and that the antisymmetric part of the bisolution is the commutator~\eqref{sec_aqft_bisolution_commutator}. For a basic field, we obtain in particular
\begin{equation}
\label{sec_aqft_commutator_singlefield}
\left[ F, \phi_K(f_K) \right]_{\star_\hbar} = \mathi \hbar \iint \frac{\delta_\text{R} F}{\delta \phi_L(y)} \Delta_{LK}(y,x) f_K(x) \total y \total x \eqend{.}
\end{equation}

The \emph{on-shell} free-field algebra is obtained from $\overline{\mathfrak{A}}_0$ by diving out terms corresponding to the free equations of motion. These are the normal-ordered products with one test function equal to the adjoint of the field equation acting on another test function, i.e., expressions of the form
\begin{equation}
\int \normord{\phi_{K_1}(x_1) \cdots \phi_{K_n}(x_n)}_G f_{K_1}(x_1) \cdots \left[ P^*_{LK_k}(x_k) f_{K_k}(x_k) \right] \cdots f_{K_n}(x_n) \total x_1 \cdots \total x_n
\end{equation}
for some $k \in \{1,\ldots,n\}$ and some $L$, and the completion of these expressions with a certain wave front set condition imposed on the limiting distributions as before. Since the fixed two-point function is a bisolution~\eqref{sec_aqft_bisolutions}, these expressions form a subspace closed under multiplication by arbitrary elements, which can be seen using integration by parts in the explicit formula~\eqref{sec_aqft_normord_product} for the product of two normal-ordered products, and thus constitute an ideal that we denote by $\mathfrak{I}_0$. The factor algebra $\overline{\mathfrak{A}}_0 / \mathfrak{I}_0$ is then the physical algebra which needs to be represented on a Hilbert space; we leave the details to the aforementioned reviews~\cite{fredenhagenrejzner2012,hollandswald2015,fewsterverch2015,fredenhagenrejzner2016,hack2016}.

Instead of normal ordering with respect to a fixed two-point function $G^+_{KL}$ of Hadamard form, one can also normal order with respect to the Hadamard parametrix $H^+_{KL}$, which is the common singular part of all Hadamard two-point functions. That is, for all $y$ in a normally geodesic neighbourhood of $x$, we have
\begin{equation}
G^+_{KL}(x,y) = H^+_{KL}(x,y) + W_{KL}(x,y) \eqend{,}
\end{equation}
where $W_{KL}$ is a smooth function, and $H^+_{KL}$ is locally and covariantly constructed from the geometry. One then defines the generators $\normord{\phi_{K_1} \cdots \phi_{K_n}}_H(f_{K_1} \otimes \cdots \otimes f_{K_n})$ in complete analogy to the $G$-normal ordered ones. In fact, one has the explicit relation
\begin{equation}
\label{sec_aqft_normordh_rel}
\normord{\phi_{K_1}(x_1) \cdots \phi_{K_n}(x_n)}_H = \sum_{S \subseteq \{1,\ldots,n\}} (-1)^S \normord{ \prod_{i \in S} \phi_{K_i}(x_i) }_G \prod_{j,k \not\in S} \frac{\mathi \hbar}{2} W_{K_jK_k}(x_j,x_k)
\end{equation}
for the respective integral kernels, where the factor $(-1)^S$ arises from anticommuting fermionic fields. Since the difference $W_{KL}$ is smooth, one can perform the completion in the same way as before, and in particular define the \emph{locally covariant Wick monomials}, which are the expressions~\eqref{sec_aqft_normordh_rel} smeared with $f_{K_1 \cdots K_n}(x_1) \delta(x_1,\ldots,x_n)$. Also the functional derivative extends to these generators in a straightforward way, and the analogue of relation~\eqref{sec_aqft_funcder_normalordered} holds for them.

By construction, any element $A \in \overline{\mathfrak{A}}_0$ can be written in the form
\begin{equation}
A = F_0 \unitmatrix + \sum_{k=1}^N \normord{\phi_{K_1} \cdots \phi_{K_k}}_G(F_k)
\end{equation}
for some finite $N$, where the $F_i$ are distributions satisfying a certain wave front set condition mentioned above~\cite{brunettifredenhagenkoehler1996,hollandswald2001}. Using the explicit relation~\eqref{sec_aqft_normordh_rel} between the normal-ordered and Hadamard-normal ordered products, a similar decomposition (with the same $N$) can be made using the Hadamard-normal ordered products. It follows in particular that any local element $A(f) = \int A(x) f(x) \total x \in \overline{\mathfrak{A}}_0$ can be written as a (finite) sum of locally covariant Wick monomials
\begin{equation}
\label{sec_aqft_local_sum_wick}
A(f) = \int f(x) \left[ \alpha_0(x) \unitmatrix + \sum_{k=1}^N \alpha_k(x) \normord{\phi_{K_1}(x) \cdots \phi_{K_k}(x)}_H \right] \total x \eqend{,}
\end{equation}
where the $\alpha_i(x)$ are functions determined from the $W_{KL}$ and their derivatives.

The Hadamard-normal ordered products are used in the construction of \emph{time-ordered products}. We denote by $\mathcal{F}$ the space of local smeared field polynomials, i.e. classical expressions of the form
\begin{equation}
\label{sec_aqft_calf_expr}
\int g\left[ \phi_K(x), \nabla \phi_K(x), \ldots, \nabla^s \phi_K(x) \right] f(x) \total x
\end{equation}
for some $s \in \mathbb{N}_0$, where $g$ is polynomial in its entries and $f$ is a test function. The unit in this vector space is $1(f) \equiv \int f(x) \total x$, and we will occasionally also denote by $\mathcal{F}_V$ the subspace where $\supp f \subset V$. The time-ordered products $\mathcal{T}$ are multilinear maps $\mathcal{T}_n\colon \mathcal{F}^{\otimes n} \to \overline{\mathfrak{A}}_0$, fulfilling a certain set of properties, in particular (see, e.g.~\cite{hollandswald2001,hollandswald2005,hollands2008,zahn2015} for a full detailed list):
\begin{itemize}
\item Locality and covariance: Given an isometric embedding $\psi\colon M \to M'$ that preserves the causal structure, $\psi_*$ the corresponding push-forward map, and $\alpha_\psi\colon \overline{\mathfrak{A}}_0(M,g) \to \overline{\mathfrak{A}}_0(M',g')$ the induced isomorphism between the free-field algebras (defined by $\alpha_\psi \phi_K(f,M,g) = \phi_K(\psi_* f, M',g')$ and by continuity on the completion), we have $\alpha_\psi \circ \mathcal{T}_n\left[ F^{\otimes n} \right] = \mathcal{T}'_n\left[ \left( \psi_* F \right)^{\otimes n} \right]$.
\item Graded symmetry: For elements $F,G \in \mathcal{F}$ with definite Grassmann parity, we have $\mathcal{T}_n\left[ \cdots F \otimes G \cdots \right] = (-1)^{\epsilon_F \epsilon_G} \mathcal{T}_n\left[ \cdots G \otimes F \cdots \right]$.
\item Neutral element: $\mathcal{T}_n\left[ F^{\otimes (n-1)} \otimes 1 \right] = \mathcal{T}_{n-1}\left[ F^{\otimes (n-1)} \right]$.
\item Causal factorisation: If $J^+(\supp F_i) \cap J^-(\supp F_j) = \emptyset$ for all $1 \leq i \leq \ell$, $\ell+1 \leq j \leq n$ for some $1 \leq \ell < n$, i.e. if all the $F_i$ lie to the future of all the $F_j$ or are spacelike separated from them, we have $\mathcal{T}_n\left( F_1 \otimes \cdots \otimes F_n \right) = \mathcal{T}_\ell\left( F_1 \otimes \cdots \otimes F_\ell \right) \star_\hbar \mathcal{T}_{n-\ell}\left( F_{\ell+1} \otimes \cdots \otimes F_n \right)$.
\item Field independence: We have
\begin{equation}
\label{sec_aqft_timeord_fieldindep}
\frac{\delta_\text{L}}{\delta \phi_K(x)} \mathcal{T}_n\left( F^{\otimes n} \right) = n \mathcal{T}_n\left[ F^{\otimes (n-1)} \otimes \frac{\delta_\text{L} F}{\delta \phi_K(x)} \right] \eqend{.}
\end{equation}
\end{itemize}
Note that expressions involving higher tensor powers $F^{\otimes n}$ of the same functional $F$, such as the neutral element property or the commutator~\eqref{sec_aqft_timeord_commutator}, only make sense when $F$ is bosonic; for fermionic functionals similar expressions hold with additional minus signs. To avoid unnecessary complication, and since one can always recover the time-ordered products with $n$ distinct factors of arbitrary Grassmann parity from the ones with an $n$-fold bosonic factor by polarisation, such expressions are always to be understood for bosonic $F$ in the following (except where explicitly noted), which in particular will be the case for all relations involving generating functions. For example, one obtains
\begin{equation}
\mathcal{T}_2\left( F \otimes G \right) = \frac{1}{2} \left[ \mathcal{T}_2\left[ \left( F + G \right) \otimes \left( F + G \right) \right] - \mathcal{T}_2\left( F \otimes F \right) - \mathcal{T}_2\left( G \otimes G \right) \right]
\end{equation}
if at least one of $F$ and $G$ is bosonic, and
\begin{equation}
\label{sec_aqft_timeordprod_polarisation2}
\alpha \beta \, \mathcal{T}_2\left( F \otimes G \right) = - \frac{1}{2} \mathcal{T}_2\left[ \left( \alpha F + \beta G \right) \otimes \left( \alpha F + \beta G \right) \right]
\end{equation}
with auxiliary Grassmann variables $\alpha$ and $\beta$ if $F$ and $G$ are both fermionic, using the multilinearity and graded symmetry property of the time-ordered products. In the field independence property, the functional derivative on the left-hand side is the derivative in the algebra $\overline{\mathfrak{A}}_0$ given by~\eqref{sec_aqft_funcder} and~\eqref{sec_aqft_funcder_normalordered}, while the functional derivative on the right-hand side acts on elements of $\mathcal{F}$. Using the formula for the commutator of an element of $\overline{\mathfrak{A}}_0$ and a basic field~\eqref{sec_aqft_commutator_singlefield}, we obtain from the field independence property the commutation relation
\begin{equation}
\label{sec_aqft_timeord_commutator}
\left[ \mathcal{T}_n\left( F^{\otimes n} \right), \phi_K(x) \right]_{\star_\hbar} = \mathi \hbar \, n \, \mathcal{T}_n\left[ F^{\otimes (n-1)} \otimes \int \frac{\delta_\text{R} F}{\delta \phi_L(y)} \Delta_{LK}(y,x) \total y \right] \eqend{.}
\end{equation}
If one only wants to construct the time-ordered products in the on-shell free-field algebra, it is sufficient to impose the commutator condition instead of the field independence property~\cite{brunettifredenhagen2000,hollandswald2001,hollandswald2002}, but for an off-shell construction field independence is necessary~\cite{brunettiduetschfredenhagen2009}.

The time-ordered products can be constructed inductively~\cite{brunettifredenhagen2000,hollandswald2001,hollandswald2002,brunettiduetschfredenhagen2009}: one starts with $\mathcal{T}_1\left( F \right) = \normord{F}_H$. By causal factorisation, time-ordered products with more than one factor are already defined (as algebra-valued distributions) in terms of time-ordered products with less factors except on the total diagonal, i.e., the integral kernel $\mathcal{T}_n\left[ F_1(x_1) \otimes \cdots \otimes F_n(x_n) \right]$ gives a finite result when smeared with test functions whose support does not include the points where $x_1 = \cdots = x_n$. In a neighbourhood of the total diagonal, one then performs an expansion of $\mathcal{T}_n$ in terms of the generators~\eqref{sec_aqft_normordh_rel} (normal-ordered with respect to the Hadamard parametrix), with c-number distributions as expansion coefficients. These coefficients are then extended to the total diagonal in such a way that all the required properties are preserved in the extension~\cite{hollandswald2001,hollandswald2002,hollands2008}. In particular, the locality and covariance property requires that the normal ordering is done with respect to the Hadamard parametrix (which is locally and covariantly constructed from the geometry), and not with respect to a two-point function. Since only the coefficients are extended, which are scalar c-number distributions, the construction also shows that time-ordered products preserve quantum numbers, in the following sense: For any grading $G = \mathbb{Z}/k \mathbb{Z}$ on the free-field algebra $\overline{\mathfrak{A}}_0 = \bigoplus_{g \in G} \overline{\mathfrak{A}}_0^g$, and the corresponding classical grading on $\mathcal{F} = \bigoplus_{g \in G} \mathcal{F}^g$, we have
\begin{equation}
g\left[ \mathcal{T}_n\left( F_1 \otimes \cdots \otimes F_n \right) \right] = \sum_{i=1}^n g(F_i)
\end{equation}
whenever the $F_i$ are homogeneous elements (with a definite grading).

The above construction is not unique, and it has been shown~\cite{hollandswald2001,hollandswald2003,duetschfredenhagen2004,brunettiduetschfredenhagen2009,khavkinemelatimoretti2019} that two different sets of time-ordered products $\mathcal{T}$ and $\hat{\mathcal{T}}$ are related by the expected renormalisation freedom. This can be most easily written by combining the time-ordered products into a generating function
\begin{equation}
\mathcal{T}\left[ \exp_\otimes\left( \frac{\mathi}{\hbar} F \right) \right] \equiv \sum_{n=0}^\infty \frac{1}{n!} \left( \frac{\mathi}{\hbar} \right)^n \mathcal{T}_n\left( F^{\otimes n} \right)
\end{equation}
with $\mathcal{T}_0 \equiv \unitmatrix$. We then have the relation
\begin{equation}
\label{sec_aqft_renormfreedom}
\hat{\mathcal{T}}\left[ \exp_\otimes\left( \frac{\mathi}{\hbar} F \right) \right] = \mathcal{T}\left[ \exp_\otimes\left( \frac{\mathi}{\hbar} F + \frac{\mathi}{\hbar} \mathcal{Z}\left( \mathe_\otimes^F \right) \right) \right] \eqend{,}
\end{equation}
where the generating functional
\begin{equation}
\mathcal{Z}\left( \mathe_\otimes^F \right) \equiv \sum_{n=0}^\infty \frac{1}{n!} \mathcal{Z}_n\left( F^{\otimes n} \right) \eqend{,} \qquad \mathcal{Z}_0 \equiv 0
\end{equation}
contains multilinear maps $\mathcal{Z}_n\colon \mathcal{F}^n \to \mathcal{F}$ fulfilling a certain set of properties, in particular (see, e.g.~\cite{hollandswald2005,hollands2008} for a full detailed list):
\begin{itemize}
\item Locality and covariance: Given an isometric embedding $\psi\colon M \to M'$ that preserves the causal structure and $\psi_*$ the corresponding push-forward map, we have $\psi^* \circ \mathcal{Z}'_n\left[ F^{\otimes n} \right] = \mathcal{Z}_n\left[ \left( \psi^* F \right)^{\otimes n} \right]$.
\item Graded symmetry: For elements $F,G \in \mathcal{F}$ with definite Grassmann parity, we have $\mathcal{Z}\left[ \cdots \otimes F \otimes G \otimes \cdots \right] = (-1)^{\epsilon_F \epsilon_G} \mathcal{Z}\left[ \cdots \otimes G \otimes F \otimes \cdots \right]$.
\item Neutral element: $\mathcal{Z}_n\left[ F^{\otimes (n-1)} \otimes 1 \right] = \mathcal{Z}_{n-1}\left[ F^{\otimes (n-1)} \right]$.
\item Support on the total diagonal: If there exist $i$, $j$ such that $\supp F_i \cap \supp F_j = \emptyset$, we have $\mathcal{Z}_n\left( F_1 \otimes \cdots \otimes F_n \right) = 0$.
\item Order in $\hbar$: If $F = \bigo{\hbar^0}$, then $\mathcal{Z}_n\left( F^{\otimes n} \right) = \bigo{\hbar}$.
\item Field independence:
\begin{equation}
\label{sec_aqft_renorm_fieldindependence}
\frac{\delta_\text{L}}{\delta \phi_K(x)} \mathcal{Z}_n\left( F^{\otimes n} \right) = n \mathcal{Z}_n\left[ F^{\otimes (n-1)} \otimes \frac{\delta_\text{L} F}{\delta \phi_K(x)} \right] \eqend{.}
\end{equation}
\end{itemize}
By the field independence property, one sees that also the $\mathcal{Z}_n$ preserve quantum numbers: $g\left[ \mathcal{Z}_n\left( F_1 \otimes \cdots \otimes F_n \right) \right] = \sum_{i=1}^n g(F_i)$ whenever the $F_i$ are homogeneous elements.

We also define the \emph{connected time-ordered products} $\mathcal{T}^\text{c}$ by
\begin{equation}
\label{sec_aqft_tc_def}
\mathcal{T}\left[ \exp_\otimes\left( \frac{\mathi}{\hbar} F \right) \right] = \exp_{\star_\hbar}\left[ \frac{\mathi}{\hbar} \mathcal{T}^\text{c}\left[ \mathe_\otimes^F \right] \right] \eqend{,}
\end{equation}
understood as an equality between generating functionals. Expanding in powers of $F$, this relation reads
\begin{splitequation}
\label{sec_aqft_tc_expanded}
&\mathcal{T}^\text{c}_n\left[ F_1 \otimes \cdots \otimes F_n \right] = \left( \frac{\mathi}{\hbar} \right)^{n-1} \sum_{k=1}^n \frac{(-1)^{k+1}}{k} \\
&\quad\times \sum_{K_1 \cup \cdots \cup K_k = \{ 1, \ldots, n\}, K_i \neq \emptyset} \mathcal{T}_{\abs{K_1}}\left[ \bigotimes_{\ell \in K_1} F_\ell \right] \star_\hbar \cdots \star_\hbar \mathcal{T}_{\abs{K_k}}\left[ \bigotimes_{\ell \in K_k} F_\ell \right]
\end{splitequation}
for bosonic $F_i$, while the appropriate sign factors for fermionic $F_i$ can be determined by multilinearity [similar to equation~\eqref{sec_aqft_timeordprod_polarisation2}]. Contrary to appearance, the connected time-ordered products are formal power series in $\hbar$~\cite{duetschfredenhagen2001}, and do not contain negative powers of $\hbar$. A simple proof can be given by noting that in the inductive construction of the time-ordered products, one only has to extend the connected part to the total diagonal:
\begin{theorem}
\label{thm_tc_hbar}
The connected time-ordered products are formal power series in $\hbar$.
\end{theorem}
\begin{proof}
For $n = 1$ we have $\mathcal{T}^\text{c}_1\left( F \right) = \mathcal{T}_1\left( F \right) = \normord{F}_H$, which is a formal power series in $\hbar$. Assume thus that is has been shown that $\mathcal{T}^\text{c}_k\left( F_1 \otimes \cdots \otimes F_k \right) = \bigo{\hbar^0}$ for all $k < n$, and consider functionals $F_+ = F_1 + \cdots + F_\ell$ and $F_- = F_{\ell+1} + \cdots + F_n$ such that $J^+(\supp F_+) \cap J^-(\supp F_-) = \emptyset$. From the definition of the connected time-ordered products~\eqref{sec_aqft_tc_def} and using the causal factorisation of the time-ordered products, we obtain
\begin{splitequation}
\exp_{\star_\hbar}\left[ \frac{\mathi}{\hbar} \mathcal{T}^\text{c}\left[ \mathe_\otimes^{F_+ + F_-} \right] \right] &= \mathcal{T}\left[ \exp_\otimes\left( \frac{\mathi}{\hbar} F_+ + \frac{\mathi}{\hbar} F_- \right) \right] \\
&= \sum_{n=0}^\infty \frac{1}{n!} \left( \frac{\mathi}{\hbar} \right)^n \mathcal{T}_n\left[ \sum_{k=0}^n \frac{n!}{k! (n-k)!} F_+^{\otimes k} \otimes F_-^{\otimes (n-k)} \right] \\
&= \left[ \sum_{k=0}^\infty \frac{1}{k!} \left( \frac{\mathi}{\hbar} \right)^k \mathcal{T}_k\left[ F_+^{\otimes k} \right] \right] \star_\hbar \left[ \sum_{n=0}^\infty \frac{1}{n!} \left( \frac{\mathi}{\hbar} \right)^n \mathcal{T}_n\left[ F_-^{\otimes n} \right] \right] \\
&= \mathcal{T}\left[ \exp_\otimes \left( \frac{\mathi}{\hbar} F_+ \right) \right] \star_\hbar \mathcal{T}\left[ \exp_\otimes \left( \frac{\mathi}{\hbar} F_- \right) \right] \\
&= \exp_{\star_\hbar}\left[ \frac{\mathi}{\hbar} \mathcal{T}^\text{c}\left[ \mathe_\otimes^{F_+} \right] \right] \star_\hbar \exp_{\star_\hbar}\left[ \frac{\mathi}{\hbar} \mathcal{T}^\text{c}\left[ \mathe_\otimes^{F_-} \right] \right] \eqend{.}
\raisetag{2em}
\end{splitequation}
The assertion now follows from the exponential (Baker--Campbell--Hausdorff) theorem, with an explicit formula due to Dynkin; see~\cite{achillesbonfliglioli2012} for a modern account of its history. It reads
\begin{splitequation}
&\exp_{\star_\hbar}\left[ \frac{\mathi}{\hbar} \mathcal{T}^\text{c}\left[ \mathe_\otimes^{F_+} \right] \right] \star_\hbar \exp_{\star_\hbar}\left[ \frac{\mathi}{\hbar} \mathcal{T}^\text{c}\left[ \mathe_\otimes^{F_-} \right] \right] \\
&\quad= \exp_{\star_\hbar}\left[ \frac{\mathi}{\hbar} \mathcal{T}^\text{c}\left[ \mathe_\otimes^{F_+} \right] + \frac{\mathi}{\hbar} \mathcal{T}^\text{c}\left[ \mathe_\otimes^{F_-} \right] + \sum_{k=2}^\infty \left( \frac{\mathi}{\hbar} \right)^k Z_k\left[ \mathcal{T}^\text{c}\left[ \mathe_\otimes^{F_+} \right], \mathcal{T}^\text{c}\left[ \mathe_\otimes^{F_-} \right] \right] \right] \eqend{,}
\end{splitequation}
where $Z_k$ is a $(k-1)$-fold commutator, homogeneous of degree $k$ in its arguments. We thus have
\begin{splitequation}
\mathcal{T}^\text{c}\left[ \mathe_\otimes^{F_+ + F_-} \right] = \mathcal{T}^\text{c}\left[ \mathe_\otimes^{F_+} \right] + \mathcal{T}^\text{c}\left[ \mathe_\otimes^{F_-} \right] + \sum_{k=2}^\infty \left( \frac{\mathi}{\hbar} \right)^{k-1} Z_k\left[ \mathcal{T}^\text{c}\left[ \mathe_\otimes^{F_+} \right], \mathcal{T}^\text{c}\left[ \mathe_\otimes^{F_-} \right] \right] \eqend{,}
\end{splitequation}
and since the commutator of two elements in $\overline{\mathfrak{A}}_0$ is at least of order $\hbar$, all terms in this formula are at least of order $\hbar^0$. Therefore the connected time-ordered products are of order $\hbar^0$ outside of the total diagonal, and since the extension to the total diagonal is independent of $\hbar$, this property is maintained in the extension.
\end{proof}
Note that these connected time-ordered products do not involve the choice of a state $\omega$ as in other approaches~\cite{hollands2008} and thus do not directly correspond to the connected products usually considered in quantum field theory (but they agree with the definition of~\cite{duetschfredenhagen2001} up to an overall factor); one might call them \emph{algebraic connected time-ordered products}.

One now has all the ingredients at hand to perturbatively construct the interacting theory in perturbation theory, and one defines the \emph{interacting time-ordered products} $\mathcal{T}_F$ with interaction $F$ by the generating functional
\begin{equation}
\label{sec_aqft_inttimeord_def}
\mathcal{T}_F\left[ \exp_\otimes\left( \frac{\mathi}{\hbar} G \right) \right] \equiv \mathcal{T}\left[ \exp_\otimes\left( \frac{\mathi}{\hbar} F \right) \right]^{\star_\hbar (-1)} \star_\hbar \mathcal{T}\left[ \exp_\otimes\left( \frac{\mathi}{\hbar} (F+G) \right) \right] \eqend{.}
\end{equation}
Note that since the time-ordered products are maps from $\mathcal{F}^{\otimes n}$, in particular their entries are smeared with a test function and we cannot simply set $F = S_\text{int}$. Instead, for $S_\text{int} = \int \mathcal{L} \total x$ with the Lagrange density $\mathcal{L}$, we set $L = \int g(x) \mathcal{L} \total x$ with a test function $g$, and consider $L$ as the interaction. The limit $g \to 1$ can then be realised at the algebraic level as an inductive limit~\cite{brunettifredenhagen2000}, called the \emph{algebraic adiabatic limit}, but we will not need it in the following. Again, we can show easily that the interacting time-ordered products are formal power series in $\hbar$, contrary to appearance. Rewriting the definition in terms of the connected time-ordered products and using the exponential theorem, we obtain
\begin{splitequation}
\label{sec_aqft_tfg_in_tcfg}
\mathcal{T}_F\left[ \exp_\otimes\left( \frac{\mathi}{\hbar} G \right) \right] &= \exp_{\star_\hbar}\left[ - \frac{\mathi}{\hbar} \mathcal{T}^\text{c}\left( \mathe_\otimes^F \right) \right] \star_\hbar \exp_{\star_\hbar}\left[ \frac{\mathi}{\hbar} \mathcal{T}^\text{c}\left( \mathe_\otimes^{F+G} \right) \right] \\
&= \exp_{\star_\hbar}\left[ \frac{\mathi}{\hbar} \mathcal{T}^\text{c}_F\left( \mathe_\otimes^G \right) \right]
\end{splitequation}
with
\begin{equation}
\mathcal{T}^\text{c}_F\left[ \mathe_\otimes^G \right] = - \mathcal{T}^\text{c}\left[ \mathe_\otimes^F \right] + \mathcal{T}^\text{c}\left[ \mathe_\otimes^{F+G} \right] + \sum_{k=2}^\infty \left( \frac{\mathi}{\hbar} \right)^{k-1} Z_k\left[ - \mathcal{T}^\text{c}\left[ \mathe_\otimes^F \right], \mathcal{T}^\text{c}\left[ \mathe_\otimes^{F+G} \right] \right] \eqend{.}
\end{equation}
Since $Z_k$ is a $(k-1)$-fold commutator and thus at least of order $\hbar^{k-1}$, and the connected time-ordered products are formal power series in $\hbar$ as shown previously, the right-hand side and thus $\mathcal{T}^\text{c}_F\left[ \mathe_\otimes^G \right]$ is again a formal power series in $\hbar$. Writing $G = G_1 + \cdots + G_n$ and expanding the relation~\eqref{sec_aqft_tfg_in_tcfg} in $G$, we obtain
\begin{splitequation}
&\mathcal{T}_{F,n}\left[ G_1 \otimes \cdots \otimes G_n \right] = \sum_{k=1}^n \frac{1}{k!} \left( \frac{\hbar}{\mathi} \right)^{n-k} \\
&\qquad\times \sum_{K_1 \cup \cdots \cup K_k = \{ 1, \ldots, n\}, K_i \neq \emptyset} \mathcal{T}^\text{c}_{F,\abs{K_1}}\left[ \bigotimes_{\ell \in K_1} G_\ell \right] \star_\hbar \cdots \star_\hbar \mathcal{T}^\text{c}_{F,\abs{K_k}}\left[ \bigotimes_{\ell \in K_k} G_\ell \right] \eqend{,}
\end{splitequation}
such that also the interacting time-ordered products do not contain negative powers of $\hbar$.

We furthermore define the \emph{retarded products} $\mathcal{R}$ as a special case of the interacting time-ordered products $\mathcal{T}_F$~\eqref{sec_aqft_inttimeord_def}, namely the linear term: $\mathcal{R}\left( \mathe_\otimes^F; G \right) = \mathcal{T}_{F,1}\left( G \right)$. Explicitly, we have
\begin{equation}
\label{sec_aqft_retprod_def}
\mathcal{R}\left( \mathe_\otimes^F; G \right) = \mathcal{T}\left[ \exp_\otimes\left( \frac{\mathi}{\hbar} F \right) \right]^{\star_\hbar (-1)} \star_\hbar \mathcal{T}\left[ \exp_\otimes\left( \frac{\mathi}{\hbar} F \right) \otimes G \right] \eqend{,}
\end{equation}
again understood as an equality between generating functionals, where
\begin{equation}
\mathcal{R}\left( \mathe_\otimes^F; G \right) = \sum_{n=0}^\infty \frac{1}{n!} \mathcal{R}_n\left( F^{\otimes n}; G \right) \eqend{.}
\end{equation}
Since the interacting time-ordered products are formal power series in $\hbar$, the retarded products are as well. It follows that the retarded products have a classical limit as $\hbar \to 0$, and it has been shown in~\cite{duetschfredenhagen2001,duetschfredenhagen2003} that in this limit they are equal to the classical retarded products,
\begin{equation}
\label{sec_aqft_rclass_limit}
\lim_{\hbar \to 0} \mathcal{R}\left( \mathe_\otimes^F; G \right) = R^\text{cl}\left( \mathe_\otimes^F; G \right) \eqend{,} \qquad \lim_{\hbar \to 0} \mathcal{R}_n\left( F^{\otimes n}; G \right) = R^\text{cl}_n\left( F^{\otimes n}; G \right) \eqend{,}
\end{equation}
identifying elements of $\overline{\mathfrak{A}}_0$ in this limit with classical local smeared field polynomials, i.e. elements of $\mathcal{F}$. An explicit formula for the classical retarded products is given by $R^\text{cl}_0\left( -; G \right) = G$ and the recursion relation~\cite[Eq. (42)]{duetschfredenhagen2003},
\begin{splitequation}
\label{sec_aqft_rcl_recursion}
&R^\text{cl}_{n+1}\left( F^{\otimes (n+1)}; G \right) = - \sum_{k=0}^n \frac{n!}{k! (n-k)!} \\
&\hspace{6em}\times R^\text{cl}_{n-k}\left( F^{\otimes (n-k)}; \iint \frac{\delta_\text{R} F}{\delta \phi_K(x)} \Delta^{\text{adv} (k)}_{KL}(x,y) \frac{\delta_\text{L} G}{\delta \phi_L(y)} \total x \total y \right) \eqend{,}
\end{splitequation}
with $\Delta^{\text{adv} (k)}_{KL}(x,y)$ defined by $\Delta^{\text{adv} (0)}_{KL}(x,y) = G^\text{adv}_{KL}(x,y)$ and the recursion
\begin{equation}
\Delta^{\text{adv} (k)}_{KL}(x,y) = - k \iint \Delta^{\text{adv} (k-1)}_{KM}(x,u) \frac{\delta_\text{L} \delta_\text{R} F}{\delta \phi_M(u) \delta \phi_N(v)} G^\text{adv}_{NL}(v,y) \total u \total v \eqend{.}
\end{equation}
In particular, at first order we have the expression
\begin{splitequation}
\label{sec_aqft_rcl_firstorder}
R^\text{cl}_1\left( F; G \right) &= - \iint \frac{\delta_\text{R} F}{\delta \phi_M(x)} G^\text{adv}_{MN}(x,y) \frac{\delta_\text{L} G}{\delta \phi_N(y)} \total x \total y \\
&= - \iint \frac{\delta_\text{R} G}{\delta \phi_M(x)} G^\text{ret}_{MN}(x,y) \frac{\delta_\text{L} F}{\delta \phi_N(y)} \total x \total y \eqend{,}
\end{splitequation}
and this relation is valid without any additional signs for $G$ of arbitrary Grassmann parity (while $F$ is assumed to be bosonic as stated before).

For later use, we now derive a relation that holds whenever $G$ is linear in fields, i.e. when $\delta^2 G/\delta \phi_K(x) \delta \phi_L(y) = 0$ for all $K$ and $L$. In this case, we have
\begin{splitequation}
&\iint \frac{\delta_\text{L} \delta_\text{R} F}{\delta \phi_M(u) \delta \phi_N(v)} G^\text{adv}_{NL}(v,y) \frac{\delta_\text{L} G}{\delta \phi_L(y)} \total v \total y \\
&\quad= \iint \frac{\delta_\text{L}}{\delta \phi_M(u)} \left[ \frac{\delta_\text{R} F}{\delta \phi_N(v)} G^\text{adv}_{NL}(v,y) \frac{\delta_\text{L} G}{\delta \phi_L(y)} \right] \total v \total y = - \frac{\delta_\text{L} R^\text{cl}_1\left( F; G \right)}{\delta \phi_M(u)} \eqend{,}
\end{splitequation}
and thus for $k > 0$
\begin{splitequation}
&\iint \frac{\delta_\text{R} F}{\delta \phi_K(x)} \Delta^{\text{adv} (k)}_{KL}(x,y) \frac{\delta_\text{L} G}{\delta \phi_L(y)} \total x \total y \\
&\quad= k \iint \frac{\delta_\text{R} F}{\delta \phi_K(x)} \Delta^{\text{adv} (k-1)}_{KL}(x,y) \frac{\delta_\text{L} R^\text{cl}_1\left( F; G \right)}{\delta \phi_L(y)} \total x \total y \eqend{.}
\end{splitequation}
After some rearrangements, it follows that we have the recursion relation
\begin{equation}
\label{sec_aqft_rcl_glinear_rel}
R^\text{cl}_{n+1}\left( F^{\otimes (n+1)}; G \right) = (n+1) R^\text{cl}_n\left[ F^{\otimes n}; R^\text{cl}_1\left( F; G \right) \right]
\end{equation}
whenever $G$ is linear in fields.

From the recursion relation~\eqref{sec_aqft_rcl_recursion}, it is seen that the classical retarded products have the further properties~\cite{duetschfredenhagen2003}:
\begin{itemize}
\item Linearity in the second argument:
\begin{equation}
\label{sec_aqft_rcl_linearity}
R^\text{cl}_n\left( F^{\otimes n}; \alpha G + \beta H \right) = \alpha R^\text{cl}_n\left( F^{\otimes n}; G \right) + \beta R^\text{cl}_n\left( F^{\otimes n}; H \right) \eqend{,}
\end{equation}
where $G,H \in \mathcal{F}$ and $\alpha,\beta$ are (possibly Grassmann-valued) c-numbers.
\item Field independence:
\begin{equation}
\label{sec_aqft_rcl_funcder}
\frac{\delta_\text{L}}{\delta \phi_K(x)} R^\text{cl}_n\left( F^{\otimes n}; G \right) = n R^\text{cl}_n\left[ F^{\otimes (n-1)} \otimes \frac{\delta_\text{L} F}{\delta \phi_K(x)}; G \right] + R^\text{cl}_n\left[ F^{\otimes n}; \frac{\delta_\text{L} G}{\delta \phi_K(x)} \right] \eqend{.}
\end{equation}
\item Factorisation:
\begin{equation}
\label{sec_aqft_rcl_factorisation}
R^\text{cl}_n\left( F^{\otimes n}; G H \right) = \sum_{k=0}^n \frac{n!}{k! (n-k)!} R^\text{cl}_k\left( F^{\otimes k}; G \right) R^\text{cl}_{n-k}\left( F^{\otimes (n-k)}; H \right) \eqend{.}
\end{equation}
This can be proven by a straightforward induction argument: it holds for $n = 0$ and on the right-hand side of the recursion relation~\eqref{sec_aqft_rcl_recursion} only retarded products with a lower number of terms in the first argument appear.
\item The GLZ (Glaser--Lehmann--Zimmermann) relation:
\begin{splitequation}
\label{sec_aqft_rcl_glz}
&R^\text{cl}_{n+1}\left( F^{\otimes n} \otimes G; H \right) = R^\text{cl}_{n+1}\left( F^{\otimes n} \otimes H; G \right) \\
&\qquad+ \sum_{k=0}^n \frac{n!}{(n-k)! k!} \left\{ R^\text{cl}_k\left( F^{\otimes k}; G \right), R^\text{cl}_{n-k}\left( F^{\otimes (n-k)}; H \right) \right\} \eqend{,}
\end{splitequation}
where $\{ \cdot, \cdot \}$ is the classical Poisson bracket defined by
\begin{equation}
\label{sec_aqft_poissonbracket_def}
\left\{ F, G \right\} \equiv \iint \frac{\delta_\text{R} F}{\delta \phi_M(x)} \Delta_{MN}(x,y) \frac{\delta_\text{L} G}{\delta \phi_N(y)} \total x \total y \eqend{,}
\end{equation}
or alternatively obtained from the quantum commutator according to
\begin{equation}
\label{sec_aqft_poissonbracket_limit}
\left\{ F, G \right\} = \lim_{\hbar \to 0} \frac{1}{\mathi \hbar} \left[ F, G \right]_{\star_\hbar} \eqend{,}
\end{equation}
with the same identification of elements of $\overline{\mathfrak{A}}_0$ with elements of $\mathcal{F}$ in this limit as in the relation~\eqref{sec_aqft_rclass_limit}. The GLZ relation is most easily proven by writing it in generating functional form
\begin{equation}
R^\text{cl}\left( \mathe_\otimes^F \otimes G; H \right) = R^\text{cl}\left( \mathe_\otimes^F \otimes H; G \right) + \left\{ R^\text{cl}\left( \mathe_\otimes^F; G \right), R^\text{cl}\left( \mathe_\otimes^F; H \right) \right\} \eqend{,}
\end{equation}
and obtaining the latter as the classical limit of
\begin{equation}
\mathcal{R}\left( \mathe_\otimes^F \otimes G; H \right) = \mathcal{R}\left( \mathe_\otimes^F \otimes H; G \right) - \frac{\mathi}{\hbar} \left[ \mathcal{R}\left( \mathe_\otimes^F; G \right), \mathcal{R}\left( \mathe_\otimes^F; H \right) \right]_{\star_\hbar} \eqend{,}
\end{equation}
which in turn follows from the definition of the retarded products~\eqref{sec_aqft_retprod_def} by polarisation, writing
\begin{splitequation}
\mathcal{R}\left( \mathe_\otimes^F \otimes G; H \right) &= \left. \frac{\partial}{\partial \alpha} \mathcal{R}\left( \mathe_\otimes^{F+\alpha G}; H \right) \right\rvert_{\alpha = 0} = - \frac{\mathi}{\hbar} \mathcal{R}\left( \mathe_\otimes^F; G \right) \star_\hbar \mathcal{R}\left( \mathe_\otimes^F; H \right) \\
&\quad+ \frac{\mathi}{\hbar} \mathcal{T}\left[ \exp_\otimes\left( \frac{\mathi}{\hbar} F \right) \right]^{\star_\hbar (-1)} \star_\hbar \mathcal{T}\left[ \exp_\otimes\left( \frac{\mathi}{\hbar} F \right) \otimes G \otimes H \right] \eqend{.}
\end{splitequation}
\end{itemize}

Lastly, we need the analogue of the causal factorisation condition for interacting time-ordered products~\cite{epsteinglaser1973,duetschfredenhagen2001}:
\begin{theorem}
\label{thm_tint_factor}
The interacting time-ordered products factorise according to
\begin{equation}
\label{sec_aqft_inttimeord_fact}
\mathcal{T}_L\left[ \exp_\otimes\left( \frac{\mathi}{\hbar} (F+G) \right) \right] = \mathcal{T}_L\left[ \exp_\otimes\left( \frac{\mathi}{\hbar} F \right) \right] \star_\hbar \mathcal{T}_L\left[ \exp_\otimes\left( \frac{\mathi}{\hbar} G \right) \right]
\end{equation}
whenever $F$ is supported in the future of $G$: $J^+(\supp F) \cap J^-(\supp G) = \emptyset$.
\end{theorem}
\begin{proof}
Multiplying the relation~\eqref{sec_aqft_inttimeord_fact} from the left with $\mathcal{T}\left[ \exp_\otimes\left( \frac{\mathi}{\hbar} L \right) \right]$ and expanding in powers of $L$, at order $L^n$ we have to show that
\begin{splitequation}
&N_n\left( L^{\otimes n}, F, G \right) \equiv \mathcal{T}\left[ \exp_\otimes\left( \frac{\mathi}{\hbar} (F+G) \right) \otimes L^{\otimes n} \right] \\
&- \frac{\partial^n}{\partial \alpha^n} \left[ \mathcal{T}\left[ \exp_\otimes\left( \frac{\mathi}{\hbar} F \right) \otimes \mathe_\otimes^{\alpha L} \right] \star_\hbar \mathcal{T}\bigg[ \mathe_\otimes^{\alpha L} \bigg]^{\star_\hbar (-1)} \star_\hbar \mathcal{T}\left[ \exp_\otimes\left( \frac{\mathi}{\hbar} G \right) \otimes \mathe_\otimes^{\alpha L} \right] \right]_{\alpha = 0}
\end{splitequation}
vanishes for the given condition on the support of $F$ and $G$. Under this condition, there exists a spacelike Cauchy surface $\Sigma$ such that $\supp F$ lies to its future and $\supp G$ to its past, $J^+(\supp F) \cap J^-(\Sigma) = \emptyset = J^-(\supp G) \cap J^+(\Sigma)$, and then a smooth time function $t$ with $t^{-1}(0) = \Sigma$ and consequently a smooth foliation of spacetime $\mathbb{R} \times \Sigma$~\cite{bernalsanchez2006}. Since $J^\pm$ are closed sets, it follows that there exists some $\tau > 0$ such that $J^+(\supp F) \cap J^-(\Sigma_{[-\tau,\tau]}) = \emptyset = J^-(\supp G) \cap J^+(\Sigma_{[-\tau,\tau]})$ with $\Sigma_S \equiv \bigcup_{s \in S} t^{-1}(s)$, i.e., we can fatten the Cauchy surface and still have it separating $\supp F$ and $\supp G$. We now partition the support of $L$ into $2n+3$ slices, obtaining $L = \sum_{k=-n-1}^{n+1} L_k$, in such a way that (see figure~\ref{fig_causality})
\begin{equation}
\supp L_k \subseteq \begin{cases} J^+\left( \Sigma_{\left[ \frac{n}{n+1} \tau, \tau \right]} \right) & k = n+1 \\ \Sigma_{\left[ \frac{k-1}{n+1} \tau, \frac{k+1}{n+1} \tau \right]} & -n \leq k \leq n \\ J^-\left( \Sigma_{\left[ -\tau, -\frac{n}{n+1} \tau \right]} \right) & k = -n-1 \eqend{,} \end{cases}
\end{equation}
which could be done by choosing a suitable partition of unity to decompose the test function in $L$. Using the multilinearity of time-ordered products, we thus only have to show that $N_n\left( L_{i_1} \otimes \cdots \otimes L_{i_n}, F, G \right)$ vanishes for arbitrary $i_k \in \{ -n-1,\ldots,n+1 \}$. Because there are $2n+3$ slices but only $n$ of them enter, and due to the chosen support properties of the $L_k$, for any choice of the $i_k$ there exists an $\ell \in \{ -n-1,\ldots,n \}$ such that $\supp L_{i_k} \cap \Sigma_{\left( \frac{\ell}{n+1} \tau, \frac{\ell+1}{n+1} \tau \right)} = \emptyset$ for all $i_k$.
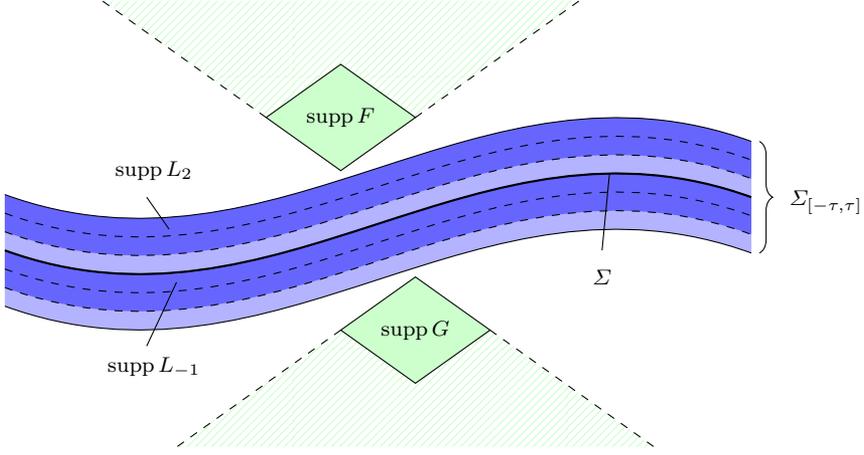
\begin{figure}[ht]
\centering
\begin{tikzpicture}
\def\figaspect{1.4}
\def\figscale{0.7}
\colorlet{mdark}{blue!60} 
\colorlet{mmed}{blue!30} 
\colorlet{mlight}{green!20} 
\colorlet{mhatch}{gray!90} 
\colorlet{mcaus}{red!40} 
\fill[pattern=north east lines, pattern color=mlight] (0.5*\figscale*\figaspect,-1.0*\figscale) -- +(3.2*\figscale*\figaspect,-3.2*\figscale) -- +(-3.2*\figscale*\figaspect,-3.2*\figscale) -- cycle;
\fill[pattern=north east lines, pattern color=mlight] (-0.5*\figscale*\figaspect,1.0*\figscale) -- +(-3.2*\figscale*\figaspect,3.2*\figscale) -- +(3.2*\figscale*\figaspect,3.2*\figscale) -- cycle;
\fill[color=mlight] (0.5*\figscale*\figaspect,2.0*\figscale) -- (-0.5*\figscale*\figaspect,3.0*\figscale) -- (-1.5*\figscale*\figaspect,2.0*\figscale) -- (-0.5*\figscale*\figaspect,1.0*\figscale) -- cycle;
\draw (0.5*\figscale*\figaspect,2.0*\figscale) -- (-0.5*\figscale*\figaspect,3.0*\figscale) -- (-1.5*\figscale*\figaspect,2.0*\figscale) -- (-0.5*\figscale*\figaspect,1.0*\figscale) -- cycle;
\fill[color=mlight] (1.5*\figscale*\figaspect,-2.0*\figscale) -- (0.5*\figscale*\figaspect,-1.0*\figscale) -- (-0.5*\figscale*\figaspect,-2.0*\figscale) -- (0.5*\figscale*\figaspect,-3.0*\figscale) -- cycle;
\draw (1.5*\figscale*\figaspect,-2.0*\figscale) -- (0.5*\figscale*\figaspect,-1.0*\figscale) -- (-0.5*\figscale*\figaspect,-2.0*\figscale) -- (0.5*\figscale*\figaspect,-3.0*\figscale) -- cycle;
\draw[dashed] (0.5*\figscale*\figaspect,2.0*\figscale) -- +(2.2*\figscale*\figaspect,2.2*\figscale);
\draw[dashed] (-1.5*\figscale*\figaspect,2.0*\figscale) -- +(-2.2*\figscale*\figaspect,2.2*\figscale);
\draw[dashed] (1.5*\figscale*\figaspect,-2.0*\figscale) -- +(2.2*\figscale*\figaspect,-2.2*\figscale);
\draw[dashed] (-0.5*\figscale*\figaspect,-2.0*\figscale) -- +(-2.2*\figscale*\figaspect,-2.2*\figscale);
\fill[color=mmed](-5.0*\figscale*\figaspect,-1.55*\figscale) to[out=-20,in=160] (5.0*\figscale*\figaspect,-0.55*\figscale) -- (5.0*\figscale*\figaspect,1.55*\figscale) to[out=160,in=-20] (-5.0*\figscale*\figaspect,0.55*\figscale) -- cycle;
\fill[color=mdark](-5.0*\figscale*\figaspect,-1.2*\figscale) to[out=-20,in=160] (5.0*\figscale*\figaspect,-0.2*\figscale) -- (5.0*\figscale*\figaspect,0.5*\figscale) to[out=160,in=-20] (-5.0*\figscale*\figaspect,-0.5*\figscale) -- cycle;
\fill[color=mdark](-5.0*\figscale*\figaspect,-0.15*\figscale) to[out=-20,in=160] (5.0*\figscale*\figaspect,0.85*\figscale) -- (5.0*\figscale*\figaspect,1.55*\figscale) to[out=160,in=-20] (-5.0*\figscale*\figaspect,0.55*\figscale) -- cycle;
\node at (-0.5*\figscale*\figaspect,2.0*\figscale) {supp\,$F$};
\node at (0.5*\figscale*\figaspect,-2.0*\figscale) {supp\,$G$};
\draw [decorate,decoration={brace,amplitude=5pt,mirror,raise=3pt},yshift=0pt]
(5.0*\figscale*\figaspect,-0.55*\figscale) -- (5.0*\figscale*\figaspect,1.55*\figscale);
\node at (6.0*\figscale*\figaspect,0.4*\figscale) {$\Sigma_{[-\tau,\tau]}$};
\node at (-3.0*\figscale*\figaspect,1.0*\figscale) {supp\,$L_2$};
\node at (-3.0*\figscale*\figaspect,-2.7*\figscale) {supp\,$L_{-1}$};
\node at (3.0*\figscale*\figaspect,-1.0*\figscale) {$\Sigma$};
\draw (-3.1*\figscale*\figaspect,0.5*\figscale) -- (-2.8*\figscale*\figaspect,-0.1*\figscale);
\draw (-3.1*\figscale*\figaspect,-2.3*\figscale) -- (-2.7*\figscale*\figaspect,-1.1*\figscale);
\draw (3.0*\figscale*\figaspect,-0.5*\figscale) -- (3.1*\figscale*\figaspect,0.95*\figscale);
\draw[line width=0.8pt] (-5.0*\figscale*\figaspect,-0.5*\figscale) to[out=-20,in=160] (5.0*\figscale*\figaspect,0.5*\figscale);
\draw (-5.0*\figscale*\figaspect,-1.55*\figscale) to[out=-20,in=160] (5.0*\figscale*\figaspect,-0.55*\figscale);
\draw (-5.0*\figscale*\figaspect,0.55*\figscale) to[out=-20,in=160] (5.0*\figscale*\figaspect,1.55*\figscale);
\draw[dashed] (-5.0*\figscale*\figaspect,-1.2*\figscale) to[out=-20,in=160] (5.0*\figscale*\figaspect,-0.2*\figscale);
\draw[dashed] (-5.0*\figscale*\figaspect,-0.85*\figscale) to[out=-20,in=160] (5.0*\figscale*\figaspect,0.15*\figscale);
\draw[dashed] (-5.0*\figscale*\figaspect,-0.5*\figscale) to[out=-20,in=160] (5.0*\figscale*\figaspect,0.5*\figscale);
\draw[dashed] (-5.0*\figscale*\figaspect,-0.15*\figscale) to[out=-20,in=160] (5.0*\figscale*\figaspect,0.85*\figscale);
\draw[dashed] (-5.0*\figscale*\figaspect,0.2*\figscale) to[out=-20,in=160] (5.0*\figscale*\figaspect,1.2*\figscale);
\end{tikzpicture}
\caption{Supports of the functionals involved in the proof of the causal factorisation condition for interacting time-ordered products for $n = 2$, $i_1 = 2$ and $i_2 = -1$. Hatched areas indicate the causal past and future. In this example, the separating Cauchy surface $\Sigma_{\frac{\ell+1/2}{n+1} \tau}$ could be taken to lie between $L_2$ and $L_{-1}$, taking $\ell = 0$, or between $L_{-1}$ and $G$, taking $\ell = -3$.}
\label{fig_causality}
\end{figure}
We thus obtain a separating Cauchy surface $\Sigma_{\frac{\ell+1/2}{n+1} \tau}$ and set
\begin{equation}
L_\pm \equiv \sum_{k\colon \supp L_{i_k} \subseteq J^\pm\left( \Sigma_{\frac{\ell+1/2}{n+1} \tau} \right)} \alpha_k L_{i_k}
\end{equation}
with constants $\alpha_k$, and $m \equiv \abs{\{k\colon \supp L_{i_k} \subseteq J^+\left( \Sigma_{\frac{\ell+1/2}{n+1} \tau} \right) \}}$. We can then recover $N_n\left( L_{i_1} \otimes \cdots \otimes L_{i_n}, F, G \right)$ from
\begin{splitequation}
&N_n\left( L_+^{\otimes m} \otimes L_-^{\otimes (n-m)}, F, G \right) = \mathcal{T}\left[ \exp_\otimes\left( \frac{\mathi}{\hbar} (F+G) \right) \otimes L_+^{\otimes m} \otimes L_-^{\otimes (n-m)} \right] \\
&\quad- \frac{\partial^m}{\partial \alpha^m} \frac{\partial^{n-m}}{\partial \beta^{n-m}} \Bigg[ \mathcal{T}\left[ \exp_\otimes\left( \frac{\mathi}{\hbar} F \right) \otimes \mathe_\otimes^{\alpha L_+ + \beta L_-} \right] \star_\hbar \mathcal{T}\bigg[ \mathe_\otimes^{\alpha L_+ + \beta L_-} \bigg]^{\star_\hbar (-1)} \\
&\hspace{8em}\star_\hbar \mathcal{T}\left[ \exp_\otimes\left( \frac{\mathi}{\hbar} G \right) \otimes \mathe_\otimes^{\alpha L_+ + \beta L_-} \right] \Bigg]_{\alpha = \beta = 0}
\end{splitequation}
by taking derivatives with respect to the $\alpha_k$. Since by construction we have $J^+(\supp L_+) \cap J^-(\supp G) = J^+(\supp F) \cap J^-(\supp L_-) = J^+(\supp L_+) \cap J^-(\supp L_-) = \emptyset$, we can use causal factorisation to obtain:
\begin{equations}
\mathcal{T}\left[ \exp_\otimes\left( \frac{\mathi}{\hbar} F \right) \otimes \mathe_\otimes^{\alpha L_+ + \beta L_-} \right] &= \mathcal{T}\left[ \exp_\otimes\left( \frac{\mathi}{\hbar} F \right) \otimes \mathe_\otimes^{\alpha L_+} \right] \star_\hbar \mathcal{T}\left[ \mathe_\otimes^{\beta L_-} \right] \eqend{,} \\
\mathcal{T}\left[ \mathe_\otimes^{\alpha L_+ + \beta L_-} \right]^{\star_\hbar (-1)} &= \mathcal{T}\left[ \mathe_\otimes^{\beta L_-} \right]^{\star_\hbar (-1)} \star_\hbar \mathcal{T}\left[ \mathe_\otimes^{\alpha L_+} \right]^{\star_\hbar (-1)} \eqend{,} \\
\mathcal{T}\left[ \exp_\otimes\left( \frac{\mathi}{\hbar} G \right) \otimes \mathe_\otimes^{\alpha L_+ + \beta L_-} \right] &= \mathcal{T}\left[ \mathe_\otimes^{\alpha L_+} \right] \star_\hbar \mathcal{T}\left[ \exp_\otimes\left( \frac{\mathi}{\hbar} G \right) \otimes \mathe_\otimes^{\beta L_-} \right] \eqend{.}
\end{equations}
It follows that
\begin{splitequation}
&N_n\left( L_+^{\otimes m} \otimes L_-^{\otimes (n-m)}, F, G \right) = \mathcal{T}\left[ \exp_\otimes\left( \frac{\mathi}{\hbar} (F+G) \right) \otimes L_+^{\otimes m} \otimes L_-^{\otimes (n-m)} \right] \\
&\quad- \frac{\partial^m}{\partial \alpha^m} \frac{\partial^{n-m}}{\partial \beta^{n-m}} \Bigg[ \mathcal{T}\left[ \exp_\otimes\left( \frac{\mathi}{\hbar} F \right) \otimes \mathe_\otimes^{\alpha L_+} \right] \star_\hbar \mathcal{T}\left[ \exp_\otimes\left( \frac{\mathi}{\hbar} G \right) \otimes \mathe_\otimes^{\beta L_-} \right] \Bigg]_{\alpha = \beta = 0} \eqend{,}
\end{splitequation}
and since by causal factorisation the term in brackets in the last line is equal to $\mathcal{T}\left[ \exp_\otimes\left( \frac{\mathi}{\hbar} (F+G) \right) \otimes \mathe_\otimes^{\alpha L_+ + \beta L_-} \right]$, we obtain $N_n\left( L_+^{\otimes m} \otimes L_-^{\otimes (n-m)}, F, G \right) = 0$ and thus $N_n\left( L^{\otimes n}, F, G \right) = 0$.
\end{proof}
Multiplying equation~\eqref{sec_aqft_inttimeord_fact} with $\mathcal{T}\left[ \exp_\otimes\left( \frac{\mathi}{\hbar} (L+F) \right) \right]^{\star_\hbar (-1)} \star_\hbar \mathcal{T}\left[ \exp_\otimes\left( \frac{\mathi}{\hbar} L \right) \right]$ from the left, we further obtain
\begin{equation}
\mathcal{T}_{L+F}\left[ \exp_\otimes\left( \frac{\mathi}{\hbar} G \right) \right] = \mathcal{T}_L\left[ \exp_\otimes\left( \frac{\mathi}{\hbar} G \right) \right] \eqend{,}
\end{equation}
which at first order in $G$ reads
\begin{equation}
\mathcal{R}\left( \mathe_\otimes^{L+F}; G \right) = \mathcal{R}\left( \mathe_\otimes^L; G \right) \eqend{.}
\end{equation}
That is, retarded products are not influenced by interactions in the future, which justifies the name ``retarded''.

\section{Anomalous Ward identities}
\label{sec_anomward}

By Noether's theorem, classical symmetries of the Lagrangian give rise to a conserved charge, and symmetry transformations on the classical phase space are implemented by the Poisson bracket with the corresponding charge. Since classical (interacting) observables factorise~\eqref{sec_aqft_rcl_factorisation}, a product of two observables invariant under a symmetry transformation is again invariant. In the quantum theory, the symmetry transformations are obtained by the graded commutator with the operator corresponding to the classical charge, but factorisation generally no longer holds. One says that the symmetry is preserved in the quantum theory if it is possible to determine a derivation on the quantum algebra $\overline{\mathfrak{A}}_0$ that reduces in the classical limit to the Poisson bracket with the classical charge. This could be done, for example, by determining quantum corrections to the operator corresponding to the classical charge. The relations between time-ordered products that hold in this case are known as Ward--Takahashi identities~\cite{ward1950,rohrlich1950,takahashi1957}; in general however there might be anomalous terms in these identities due to the lack of factorisation, or due to the impossibility to find a suitable derivation. We first consider the case of a general derivation, and afterwards consider derivations obtained as graded commutators.
\begin{theorem}
\label{thm_anomward}
Given a derivation $D$ acting on the free-field algebra $\overline{\mathfrak{A}}_0$ in a locally and covariant way and preserving spacetime locality, there exist multilinear maps
\begin{equation}
\mathcal{D}_n\colon \mathcal{F}^{\otimes n} \to \mathcal{F} \eqend{,} \qquad \mathcal{A}_n\colon \mathcal{F}^{\otimes n} \to \mathcal{F} \eqend{,} \qquad \mathcal{D}_0 = \mathcal{A}_0 = 0
\end{equation}
such that
\begin{equation}
\label{thm_anomward_identity}
D \mathcal{T}\left[ \exp_\otimes\left( \frac{\mathi}{\hbar} F \right) \right] = \frac{\mathi}{\hbar} \mathcal{T}\left[ \left[ \mathcal{D}\left( \mathe_\otimes^F \right) + \mathcal{A}\left( \mathe_\otimes^F \right) \right] \otimes \exp_\otimes\left( \frac{\mathi}{\hbar} F \right) \right]
\end{equation}
holds as an equality between generating functionals, with
\begin{equation}
\mathcal{D}\left( \mathe_\otimes^F \right) = \sum_{n=0}^\infty \frac{1}{n!} \mathcal{D}_n\left( F^{\otimes n} \right) \eqend{,} \qquad \mathcal{A}\left( \mathe_\otimes^F \right) = \sum_{n=0}^\infty \frac{1}{n!} \mathcal{A}_n\left( F^{\otimes n} \right) \eqend{,}
\end{equation}
and the following properties (stated for $\mathcal{D}_n$, but fulfilled also by $\mathcal{A}_n$):
\begin{enumerate}
\item \textbf{Locality and covariance:} Given an isometric embedding $\psi\colon M \to M'$ that preserves the causal structure, $\psi_*$ the corresponding push-forward map, $\alpha_\psi\colon \overline{\mathfrak{A}}_0(M,g) \to \overline{\mathfrak{A}}_0(M',g')$ the induced isomorphism between the free-field algebras, $D'$ the derivation acting on $\overline{\mathfrak{A}}_0(M',g')$ and $\mathcal{D}'_n$ the corresponding multilinear maps, we have $\psi^* \circ \mathcal{D}'_n\left[ F^{\otimes n} \right] = \mathcal{D}_n\left[ \left( \psi^* F \right)^{\otimes n} \right]$.
\item \textbf{Graded symmetry:} For elements $F,G \in \mathcal{F}$ with definite Grassmann parity, we have $\mathcal{D}_n\left[ \cdots \otimes F \otimes G \otimes \cdots \right] = (-1)^{\epsilon_F \epsilon_G} \mathcal{D}_n\left[ \cdots \otimes G \otimes F \otimes \cdots \right]$.
\item \textbf{Neutral element:} $\mathcal{D}_n\left[ F^{\otimes (n-1)} \otimes 1 \right] = \mathcal{D}_{n-1}\left[ F^{\otimes (n-1)} \right]$.
\item \textbf{Quantum number preservation:} For any grading $G = \mathbb{Z}/k \mathbb{Z}$ on the free-field algebra $\overline{\mathfrak{A}}_0 = \bigoplus_{g \in G} \overline{\mathfrak{A}}_0^g$ which is preserved by $D$ in the sense that $D\colon \overline{\mathfrak{A}}_0^g \to \overline{\mathfrak{A}}_0^{g+d}$ for a fixed $d = g(D)$, $\mathcal{D}_n$ preserves the corresponding classical grading on $\mathcal{F} = \bigoplus_{g \in G} \mathcal{F}^g$: \begin{equation*} g\left[ \mathcal{D}_n\left( F_1 \otimes \cdots \otimes F_n \right) \right] = d + \sum_{i=1}^n g(F_i) \eqend{,}
\end{equation*}
whenever the $F_i$ are homogeneous elements (with a definite grading).
\item \textbf{Support on the total diagonal:} If there exist $i$, $j$ such that $\supp F_i \cap \supp F_j = \emptyset$, we have $\mathcal{D}_n\left( F_1 \otimes \cdots \otimes F_n \right) = 0$.
\item \textbf{Order in $\hbar$:} If $F = \bigo{\hbar^0}$, then $\mathcal{D}_n\left( F^{\otimes n} \right) = \bigo{\hbar^0}$ and $\mathcal{A}_n\left( F^{\otimes n} \right) = \bigo{\hbar}$.
\end{enumerate}
\end{theorem}
\begin{remark*}
Acting in a locally and covariant way means that $D A$ must transform in the same way as $A \in \overline{\mathfrak{A}}_0$ under causality-preserving embeddings, i.e., we must have $\alpha_\psi \circ D = D' \circ \alpha_\psi$. Preserving spacetime locality means that $\supp (D F) \subseteq \supp F$. Property $4$ applies, for example, in the BV formalism with the grading being the ghost number. One can decompose the (free) BRST differential into parts with definite ghost number, and obtains a anomalous Ward identity for each part. Property $6$ is the only one that distinguishes $\mathcal{D}_n$ and $\mathcal{A}_n$, and in fact one can define $\mathcal{D}_n\left( F^{\otimes n} \right) = \lim_{\hbar \to 0} \left[ \mathcal{D}_n\left( F^{\otimes n} \right) + \mathcal{A}_n\left( F^{\otimes n} \right) \right]$ when $F$ is independent of $\hbar$. Analogously to the case of time-ordered products, from $\mathcal{D}_n\left( F^{\otimes n} \right)$ with bosonic $F$ using the multilinearity one can recover the general expression for different functionals of arbitrary Grassmann parity by polarisation. $\mathcal{A}$ is called the anomaly, and $\mathcal{D}$ the classical part of the anomalous Ward identity.
\end{remark*}
\begin{proof}
For ease of notation, we set $\hat{\mathcal{D}}_n \equiv \mathcal{D}_n + \mathcal{A}_n$, and define $\mathcal{D}_n$ as in the remark, such that property $6$ holds if we can show that $\hat{\mathcal{D}}_n\left( F^{\otimes n} \right) = \bigo{\hbar^0}$ if $F = \bigo{\hbar^0}$. However, this follows by using the connected time-ordered products~\eqref{sec_aqft_tc_def}, in terms of which the anomalous Ward identity~\eqref{thm_anomward_identity} reads
\begin{equation}
\label{proof_anomward_identity_connected}
D \mathcal{T}^\text{c}\left[ \mathe_\otimes^F \right] = \mathcal{T}^\text{c}\left[ \hat{\mathcal{D}}\left( \mathe_\otimes^F \right) \otimes \mathe_\otimes^F \right] \eqend{.}
\end{equation}
While this relation may seem surprising given the complicated relation~\eqref{sec_aqft_tc_def} between the ordinary and connected time-ordered products, it is easy to prove: We take $F \to F + \alpha G$ in the relation~\eqref{sec_aqft_tc_def}, expand to first order in $\alpha$ using the power series expansion of the $\star_\hbar$-exponential and set $G = \hat{\mathcal{D}}\left( \mathe_\otimes^F \right)$. On the other hand, we apply $D$ on the relation~\eqref{sec_aqft_tc_def} and use that $D$ is a derivation; comparing the two expressions order by order in $F$ yields equation~\eqref{proof_anomward_identity_connected}. Since we have shown that the connected time-ordered products do not contain negative factors of $\hbar$, it follows that $\hat{\mathcal{D}}_n\left( F^{\otimes n} \right) = \bigo{\hbar^0}$, except if $\hat{\mathcal{D}}\left( \mathe_\otimes^F \right) \otimes \mathe_\otimes^F$ would lie in the kernel of the connected time-ordered products. From the explicit relation between the time-ordered products and the connected ones~\eqref{sec_aqft_tc_expanded} it follows that the connected time-ordered products $\mathcal{T}^\text{c}_n\left( F_1 \otimes \cdots \otimes F_n \right)$ vanish whenever the time-ordered product of at least two of the $F_i$ factorises, which by the causal factorisation property happens when the support of those $F_i$ can be separated by a Cauchy surface. However, since $\hat{\mathcal{D}}_n\left( F^{\otimes n} \right)$ is supported on the total diagonal as will be shown below, this does not happen, and it follows that $\hat{\mathcal{D}}\left( \mathe_\otimes^F \right) \otimes \mathe_\otimes^F$ does not lie in the kernel of the connected time-ordered products, and is thus of order $\bigo{\hbar^0}$. Multilinearity of the maps $\hat{\mathcal{D}}_n$ holds because it holds for the time-ordered products and $D$ (by definition of a derivation). The graded symmetry property $2$ holds because it holds for the time-ordered products, and taking both together we only need to prove the identity for bosonic $F_k$. Expanding the anomalous Ward identity~\eqref{thm_anomward_identity} in powers of $F$, we obtain
\begin{equation}
\label{proof_anomward_identity}
D \mathcal{T}_n\left( F^{\otimes n} \right) = \sum_{k=1}^n \frac{n!}{k! (n-k)!} \left( \frac{\hbar}{\mathi} \right)^{k-1} \mathcal{T}_{n-k+1}\left[ \hat{\mathcal{D}}_k\left( F^{\otimes k} \right) \otimes F^{\otimes (n-k)} \right] \eqend{,}
\end{equation}
and the remaining properties are then proven by induction in $n$. Let us set
\begin{splitequation}
\label{proof_anomward_nndef}
N_n\left( F^{\otimes n} \right) &\equiv D \mathcal{T}_n\left( F^{\otimes n} \right) \\
&\quad- \sum_{k=1}^{n-1} \frac{n!}{k! (n-k)!} \left( \frac{\hbar}{\mathi} \right)^{k-1} \mathcal{T}_{n-k+1}\left[ \hat{\mathcal{D}}_k\left( F^{\otimes k} \right) \otimes F^{\otimes (n-k)} \right] \eqend{,}
\end{splitequation}
the difference between the left- and right-hand sides of the anomalous Ward identity without the last term where $k = n$. The anomalous Ward identity~\eqref{proof_anomward_identity} then reads
\begin{equation}
\label{proof_anomward_dndef}
\mathcal{T}_1\left[ \hat{\mathcal{D}}_n\left( F^{\otimes n} \right) \right] = \left( \frac{\mathi}{\hbar} \right)^{n-1} N_n\left( F^{\otimes n} \right) \eqend{,}
\end{equation}
and we show below that $N_n$ is supported on the total diagonal. By the expansion~\eqref{sec_aqft_local_sum_wick}, any such element of $\overline{\mathfrak{A}}_0$ can be written as a sum of locally covariant Wick products, therefore as a sum of time-ordered products with one factor, and we simply define $\hat{\mathcal{D}}_n$ by equation~\eqref{proof_anomward_dndef}. The properties $1$ and $4$ are then fulfilled because the time-ordered products fulfil them, and property $3$ holds using the neutral-element property of time-ordered products and $D \unitmatrix = 0$ (since $D$ is a derivation).

To show that $N_n$ and thus $\hat{\mathcal{D}}_n$ is supported on the total diagonal, property $5$, we proceed by induction in $n$. For $n = 1$, it is automatically fulfilled; assume thus that it holds for all $k < n$ and assume that there are $i$ and $j$ such that $\supp F_i \cap \supp F_j = \emptyset$. For each pair of points $x_i \in \supp F_i$, $x_j \in \supp F_j$ there exists then a spacelike Cauchy surface separating them, and by continuity and compactness of the $\supp F_k$ this will still be true for $\epsilon$-neighbourhoods $U_\epsilon(x_i)$, $U_\epsilon(x_j)$ for some $\epsilon > 0$ uniformly over all points $x_i$, $x_j$. Using a subordinate partition of unity, we can thus obtain a finite decomposition of $F_i$ and $F_j$ into $\epsilon$-neighbourhoods of which each pair can be separated by a spacelike Cauchy surface (i.e., we decompose the test functions appearing in $F_i$ and $F_j$), and by the multilinearity of $N_n$ it is enough to show that $N_n$ vanishes whenever $\supp F_i$ and $\supp F_j$ can be separated by a spacelike Cauchy surface. We then proceed as in the proof of Theorem~\ref{thm_tint_factor}, fattening this Cauchy surface and partitioning the support of the other $F_k = \sum_{m=-n-1}^{n+1} F_k^m$ with $i \neq k \neq j$ accordingly, and then using multilinearity it is sufficient to show the vanishing of $N_n\left( F_i \otimes F_j \otimes \bigotimes_{i \neq k \neq j} F_k^{m_k} \right)$ for arbitrary $m_k \in \{ -n-1, \ldots, n+1 \}$. It follows again that for any choice of the $m_k$ there exists a separating Cauchy surface $\Sigma$, and setting $F_+ = \alpha_i F_i + \sum_{k\colon \supp F_k^{m_k} \subseteq J^+(\Sigma)} \alpha_k F_k^{m_k}$ and $F_- = \alpha_j F_j + \sum_{k\colon \supp F_k^{m_k} \subseteq J^-(\Sigma)} \alpha_k F_k^{m_k}$ (where we assumed w.l.o.g. that $F_i$ lies to the future of $\Sigma$ and $F_j$ to the past) we can recover $N_n\left( F_i \otimes F_j \otimes \bigotimes_{i \neq k \neq j} F_k^{m_k} \right)$ from $N_n\left( F_+^{\otimes \ell} \otimes F_-^{\otimes (n-\ell)} \right)$ (for some $1 \leq \ell \leq n-1$) by taking derivatives w.r.t. the $\alpha_k$. Using the causal factorisation of the time-ordered products, the fact that $D$ is a derivation, that
\begin{equation}
\hat{\mathcal{D}}_k\left( F_+^{\otimes \ell} \otimes F_-^{\otimes (k-\ell)} \right) = 0 \qquad (\ell \geq 1)
\end{equation}
holds for all $k < n$ by assumption from property $5$ because $\supp F_+ \cap \supp F_- = \emptyset$, and that $\supp \hat{\mathcal{D}}_k\left( F^{\otimes k} \right) \subset \supp F$ by property $5$, we obtain
\begin{splitequation}
&N_n\left( F_+^{\otimes \ell} \otimes F_-^{\otimes (n-\ell)} \right) = \bigg[ D \mathcal{T}_\ell\left( F_+^{\otimes \ell} \right) - \sum_{k=1}^\ell \frac{\ell!}{k! (\ell-k)!} \left( \frac{\hbar}{\mathi} \right)^{k-1} \\
&\hspace{12em}\times \mathcal{T}_{\ell-k+1}\left[ \hat{\mathcal{D}}_k\left( F_+^{\otimes k} \right) \otimes F_+^{\otimes (\ell-k)} \right] \bigg] \star_\hbar \mathcal{T}_{n-\ell}\left( F_-^{\otimes (n-\ell)} \right) \\
&\qquad+ \mathcal{T}_\ell\left( F_+^{\otimes \ell} \right) \star_\hbar \bigg[ D \mathcal{T}_{n-\ell}\left( F_-^{\otimes (n-\ell)} \right) - \sum_{k=1}^{n-\ell} \frac{(n-\ell)!}{k! (n-\ell-k)!} \left( \frac{\hbar}{\mathi} \right)^{k-1} \\
&\hspace{12em}\times \mathcal{T}_{n-\ell-k+1}\left[ \hat{\mathcal{D}}_k\left( F_-^{\otimes k} \right) \otimes F_-^{\otimes (n-\ell-k)} \bigg] \right] \eqend{.}
\raisetag{2em}
\end{splitequation}
The terms in brackets are the anomalous Ward identities for $\ell$ and $n-\ell$ factors, which hold by assumption, such that $N_n\left( F_+^{\otimes \ell} \otimes F_-^{\otimes (n-\ell)} \right) = 0$. It follows from equation~\eqref{proof_anomward_dndef} that $\hat{\mathcal{D}}_n\left( F_+^{\otimes \ell} \otimes F_-^{\otimes (n-\ell)} \right) = 0$, and property $5$ holds. It is only at this point that the fact that $D$ is a derivation enters; for all other properties only its linearity is necessary.
\end{proof}

\begin{theorem}
\label{thm_classanom}
If $D$ is an inner derivation, whose action is given by the graded commutator with an element of the free-field algebra $Q \in \overline{\mathfrak{A}}_0$ (of arbitrary Grassmann parity), the following holds:
\begin{enumerate}
\item Identifying $Q_\text{cl} = \lim_{\hbar \to 0} Q$ with an element of $\mathcal{F}$, we have
\begin{equation*} \mathcal{D}_1\left( F \right) = \{ Q_\text{cl}, F \} = \iint \frac{\delta_\text{R} Q_\text{cl}}{\delta \phi_M(x)} \Delta_{MN}(x,y) \frac{\delta_\text{L} F}{\delta \phi_N(y)} \total x \total y \eqend{.} \end{equation*}
\item At second order, we have
\begin{splitequation*} \mathcal{D}_2\left( F \otimes F \right) &= \left\{ Q_\text{cl}, R^\text{cl}_1\left( F; F \right) \right\} - R^\text{cl}_1\left( F; \{ Q_\text{cl}, F \} \right) - R^\text{cl}_1\left( \{ Q_\text{cl}, F \}; F \right) \\ &= \iint \frac{\delta_\text{R} F}{\delta \phi_K(x)} \left[ G^\text{ret}_{KL}(x,y) + G^\text{adv}_{KL}(x,y) \right] \left\{ \frac{\delta_\text{L} Q_\text{cl}}{\delta \phi_L(y)}, F \right\} \total x \total y \eqend{.} \end{splitequation*}
\item $\mathcal{D}_k\left( F^{\otimes k} \right) = 0$ for all $k \geq 3$ if $Q_\text{cl}$ is at most of second order in fields, i.e. if $\delta^3 Q_\text{cl} / [ \delta \phi_K(x) \delta \phi_L(y) \delta \phi_M(z) ] = 0$ for all $K,L,M$.
\end{enumerate}
\end{theorem}
\begin{remark*} In AQFT, the local C${}^*$ algebras are generically of type III${}_1$~\cite{buchholzdantonifredenhagen1987,buchholzverch1995,yngvason2005} for which all derivations are inner~\cite{sakai1966,kadison1966}, such that being inner is no restriction on $D$ in this case. However, the algebra $\overline{\mathfrak{A}}_0$ is only a $*$-algebra since the generators are unbounded. Property (3) applies in particular to all symmetries of the free field theory, which are implemented in the quantum theory by graded commutators with the corresponding Noether charge. Apart from the expected term (1), the corresponding Ward identity for time-ordered products has then at most one additional contribution given by (2).
\end{remark*}
\begin{proof}
We assume without loss of generality that $Q$ is bosonic; for fermionic $Q$ we introduce an auxiliary Grassmann parameter $\beta$ and consider instead the bosonic $\beta Q$. Using the properties of left and right derivatives~\eqref{sec_aqft_leftrightder_props}, it is easily seen that $\beta$ can be taken out on the left without introducing any extra sign for all the formulas in the theorem, and that they thus hold without change for fermionic $Q$ if they hold in the bosonic case.

Applying $D$ on the retarded product
\begin{equation}
\mathcal{R}\left( \mathe_\otimes^{\alpha F}; F \right) = \frac{\hbar}{\mathi} \mathcal{T}\left[ \exp_\otimes\left( \frac{\mathi}{\hbar} \alpha F \right) \right]^{\star_\hbar (-1)} \star_\hbar \frac{\total}{\total \alpha} \mathcal{T}\left[ \exp_\otimes\left( \frac{\mathi}{\hbar} \alpha F \right) \right] \eqend{,}
\end{equation}
using that $D$ is a derivation, using differentiation by parts (in $\alpha$), and using the anomalous Ward identity~\eqref{thm_anomward_identity}, we obtain
\begin{splitequation}
D \mathcal{R}\left( \mathe_\otimes^{\alpha F}; F \right) &= \mathcal{R}\left[ \mathe_\otimes^{\alpha F}; \mathcal{D}\left( \mathe_\otimes^{\alpha F} \otimes F \right) \right] + \mathcal{R}\left[ \mathe_\otimes^{\alpha F} \otimes F; \mathcal{D}\left( \mathe_\otimes^{\alpha F} \right) \right] \\
&\quad+ \frac{\mathi}{\hbar} \left[ \mathcal{R}\left( \mathe_\otimes^{\alpha F}; F \right), \mathcal{R}\left[ \mathe_\otimes^{\alpha F}; \mathcal{D}\left( \mathe_\otimes^{\alpha F} \right) \right] \right]_{\star_\hbar} \\
&\quad+ \mathcal{R}\left[ \mathe_\otimes^{\alpha F}; \mathcal{A}\left( \mathe_\otimes^{\alpha F} \otimes F \right) \right] + \mathcal{R}\left[ \mathe_\otimes^{\alpha F} \otimes F; \mathcal{A}\left( \mathe_\otimes^{\alpha F} \right) \right] \\
&\quad+ \frac{\mathi}{\hbar} \left[ \mathcal{R}\left( \mathe_\otimes^{\alpha F}; F \right), \mathcal{R}\left[ \mathe_\otimes^{\alpha F}; \mathcal{A}\left( \mathe_\otimes^{\alpha F} \right) \right] \right]_{\star_\hbar} \eqend{.}
\end{splitequation}
Using the limits~\eqref{sec_aqft_rclass_limit} for the commutator and~\eqref{sec_aqft_poissonbracket_limit} for the retarded products, we can now take the classical limit $\hbar \to 0$, and since the anomaly $\mathcal{A}$ is at least of order $\hbar$ the corresponding terms vanish in this limit. Expanding in powers of $\alpha$, we obtain
\begin{splitequation}
\label{proof_classanom_actionretarded}
&D_\text{cl} R^\text{cl}_n\left( F^{\otimes n}; F \right) = \sum_{k=0}^n \frac{n!}{k! (n-k)!} R^\text{cl}\left[ F^{\otimes (n-k)}; \mathcal{D}_{k+1}\left( F^{\otimes (k+1)} \right) \right] \\
&\hspace{8em}+ \sum_{k=0}^n \frac{n!}{k! (n-k)!} R^\text{cl}_{n-k+1}\left[ F^{\otimes (n-k+1)}; \mathcal{D}_k\left( F^{\otimes k} \right) \right] \\
&\quad- \sum_{k=0}^n \sum_{\ell=0}^k \frac{n!}{(n-k)! (k-\ell)! \ell!} \left\{ R^\text{cl}_{n-k}\left( F^{\otimes (n-k)}; F \right), R^\text{cl}_{k-\ell}\left[ F^{\otimes (k-\ell)}; \mathcal{D}_\ell\left( F^{\otimes \ell} \right) \right] \right\} \eqend{.}
\end{splitequation}
In particular, using that $R^\text{cl}_0\left( -; F \right) = F$~\eqref{sec_aqft_rcl_recursion} it follows that
\begin{equations}
\mathcal{D}_1\left( F \right) &= D_\text{cl} F \eqend{,} \\
\begin{split}
\mathcal{D}_2\left( F \otimes F \right) &= D_\text{cl} R^\text{cl}_1\left( F; F \right) - 2 R^\text{cl}_1\left( F; D_\text{cl} F \right) + \left\{ F, D_\text{cl} F \right\} \\
&= D_\text{cl} R^\text{cl}_1\left( F; F \right) - R^\text{cl}_1\left( F; D_\text{cl} F \right) - R^\text{cl}_1\left( D_\text{cl} F; F \right) \eqend{.}
\end{split}
\end{equations}
Using the explicit form of $D_\text{cl}$ (noting that $F$ necessarily is bosonic for $\mathcal{D}_2\left( F \otimes F \right)$ because of the graded symmetry), the explicit formula~\eqref{sec_aqft_rcl_firstorder} for the classical retarded product, the symmetry of the retarded and advanced propagators~\eqref{sec_aqft_retadv_rel}
\begin{equation}
G^\text{ret}_{MN}(x,y) = (-1)^{\epsilon_M \epsilon_N} G^\text{adv}_{NM}(y,x) \eqend{,}
\end{equation}
and that the Poisson bracket commutes with functional derivatives and factorises according to
\begin{equations}[poisson_funcder]
\frac{\delta_\text{L}}{\delta \phi_K(x)} \left\{ F, G \right\} &= \left\{ \frac{\delta_\text{L} F}{\delta \phi_K(x)}, G \right\} + (-1)^{\epsilon_F \epsilon_K} \left\{ F, \frac{\delta_\text{L} G}{\delta \phi_K(x)} \right\} \eqend{,} \\
\frac{\delta_\text{R}}{\delta \phi_K(x)} \left\{ F, G \right\} &= \left\{ F, \frac{\delta_\text{R} G}{\delta \phi_K(x)} \right\} + (-1)^{\epsilon_G \epsilon_K} \left\{ \frac{\delta_\text{R} F}{\delta \phi_K(x)}, G \right\} \eqend{,} \\
\left\{ F, G H \right\} &= \left\{ F, G \right\} H + (-1)^{\epsilon_F \epsilon_G} G \left\{ F, H \right\} \eqend{,}
\end{equations}
we obtain the given formulas for $\mathcal{D}_1$ and $\mathcal{D}_2$.

To prove that all $\mathcal{D}_k\left( F^{\otimes k} \right)$ with $k \geq 3$ vanish if $Q_\text{cl}$ is at most of second order in fields, we shift summation indices in equation~\eqref{proof_classanom_actionretarded} to obtain
\begin{splitequation}
&\mathcal{D}_{n+1}\left( F^{\otimes (n+1)} \right) = D_\text{cl} R^\text{cl}_n\left( F^{\otimes n}; F \right) \\
&\quad- \sum_{k=1}^n \frac{n!}{k! (n-k)!} \bigg[ \frac{n+1}{n+1-k} R^\text{cl}_{n-k+1}\left[ F^{\otimes (n-k+1)}; \mathcal{D}_k\left( F^{\otimes k} \right) \right] \\
&\qquad- \sum_{\ell=1}^k \frac{k!}{(k-\ell)! \ell!} \left\{ R^\text{cl}_{n-k}\left( F^{\otimes (n-k)}; F \right), R^\text{cl}_{k-\ell}\left[ F^{\otimes (k-\ell)}; \mathcal{D}_\ell\left( F^{\otimes \ell} \right) \right] \right\} \bigg] \eqend{.}
\end{splitequation}
Since on the right-hand side only $\mathcal{D}_k\left( F^{\otimes k} \right)$ with $k \leq n$ appear, we can proceed by induction. Defining for $n \geq 2$
\begin{splitequation}
\mathcal{N}_{n+1} &\equiv 2 D_\text{cl} R^\text{cl}_n\left( F^{\otimes n}; F \right) - 2 (n+1) R^\text{cl}_n\left[ F^{\otimes n}; \mathcal{D}_1\left( F \right) \right] \\
&\quad- n (n+1) R^\text{cl}_{n-1}\left[ F^{\otimes (n-1)}; \mathcal{D}_2\left( F \otimes F \right) \right] \\
&\quad+ 2 \sum_{k=1}^n \frac{n!}{(k-1)! (n-k)!} \left\{ R^\text{cl}_{n-k}\left( F^{\otimes (n-k)}; F \right), R^\text{cl}_{k-1}\left[ F^{\otimes (k-1)}; \mathcal{D}_1\left( F \right) \right] \right\} \\
&\quad+ \sum_{k=2}^n \frac{n!}{(k-2)! (n-k)!} \left\{ R^\text{cl}_{n-k}\left( F^{\otimes (n-k)}; F \right), R^\text{cl}_{k-2}\left[ F^{\otimes (k-2)}; \mathcal{D}_2\left( F \otimes F \right) \right] \right\} \eqend{,}
\end{splitequation}
it follows inductively in $n$ that $\mathcal{D}_n\left( F^{\otimes n} \right) = 0$ for all $n \geq 3$ if $\mathcal{N}_n$ vanishes for all $n \geq 3$.

Using that the classical retarded products commute with functional derivatives~\eqref{sec_aqft_rcl_funcder} and the GLZ relation~\eqref{sec_aqft_rcl_glz}
\begin{splitequation}
R^\text{cl}_{n+1}\left( H^{\otimes n} \otimes F; G \right) &= R^\text{cl}_{n+1}\left( H^{\otimes n} \otimes G; F \right) \\
&\quad+ \sum_{k=0}^n \frac{n!}{(n-k)! k!} \left\{ R^\text{cl}_k\left( H^{\otimes k}; F \right), R^\text{cl}_{n-k}\left( H^{\otimes (n-k)}; G \right) \right\} \eqend{,}
\end{splitequation}
we obtain
\begin{splitequation}
&D_\text{cl} R^\text{cl}_n\left( F^{\otimes n}; F \right) = \left\{ Q_\text{cl}, R^\text{cl}_n\left( F^{\otimes n}; F \right) \right\} \\
&= \iint \Delta_{MN}(x,y) \frac{\delta_\text{R} Q_\text{cl}}{\delta \phi_M(x)} \bigg[ n R^\text{cl}_n\left( F^{\otimes (n-1)} \otimes \frac{\delta_\text{L} F}{\delta \phi_N(y)}; F \right) \\
&\hspace{12em}+ R^\text{cl}_n\left( F^{\otimes n}; \frac{\delta_\text{L} F}{\delta \phi_N(y)} \right) \bigg] \total x \total y \\
&= \iint \Delta_{MN}(x,y) \frac{\delta_\text{R} Q_\text{cl}}{\delta \phi_M(x)} \bigg[ (n+1) R^\text{cl}_n\left( F^{\otimes n}; \frac{\delta_\text{L} F}{\delta \phi_N(y)} \right) \\
&\quad+ \sum_{k=1}^n \frac{n!}{(n-k)! (k-1)!} \left\{ R^\text{cl}_{k-1}\left( F^{\otimes (k-1)}; \frac{\delta_\text{L} F}{\delta \phi_N(y)} \right), R^\text{cl}_{n-k}\left( F^{\otimes (n-k)}; F \right) \right\} \bigg] \total x \total y \eqend{.}
\end{splitequation}
Using also the previous result
\begin{equations}
\mathcal{D}_1\left( F \right) &= \left\{ Q_\text{cl}, F \right\} = \iint \Delta_{MN}(x,y) \frac{\delta_\text{R} Q_\text{cl}}{\delta \phi_M(x)} \frac{\delta_\text{L} F}{\delta \phi_N(y)} \total x \total y \eqend{,} \\
\begin{split}
\mathcal{D}_2\left( F \otimes F \right) &= \left\{ Q_\text{cl}, R^\text{cl}_1\left( F; F \right) \right\} - 2 R^\text{cl}_1\left( F; \left\{ Q_\text{cl}, F \right\} \right) + \left\{ F, \left\{ Q_\text{cl}, F \right\} \right\} \\
&= \iint \Delta_{MN}(x,y) \bigg[ \frac{\delta_\text{R} Q_\text{cl}}{\delta \phi_M(x)} \left[ R^\text{cl}_1\left( \frac{\delta_\text{L} F}{\delta \phi_N(y)}; F \right) + R^\text{cl}_1\left( F; \frac{\delta_\text{L} F}{\delta \phi_N(y)} \right) \right] \\
&\qquad- 2 R^\text{cl}_1\left( F; \frac{\delta_\text{R} Q_\text{cl}}{\delta \phi_M(x)} \frac{\delta_\text{L} F}{\delta \phi_N(y)} \right) + \left\{ F, \frac{\delta_\text{R} Q_\text{cl}}{\delta \phi_M(x)} \frac{\delta_\text{L} F}{\delta \phi_N(y)} \right\} \bigg] \total x \total y
\end{split}
\end{equations}
and the linearity of the retarded products in their second argument~\eqref{sec_aqft_rcl_linearity}, it follows that
\begin{splitequation}
\mathcal{N}_{n+1} &= \iint \Delta_{MN}(x,y) \bigg[ (n+1) \mathcal{N}^{n+1}_{MN}(x,y) \\
&\qquad- \sum_{k=1}^n \frac{n!}{(k-1)! (n-k)!} \left\{ R^\text{cl}_{n-k}\left( F^{\otimes (n-k)}; F \right), \mathcal{N}^k_{MN}(x,y) \right\} \\
&\qquad+ 2 \sum_{k=1}^n \frac{n!}{(k-1)! (n-k)!} \left\{ R^\text{cl}_{n-k}\left( F^{\otimes (n-k)}; F \right), \frac{\delta_\text{R} Q_\text{cl}}{\delta \phi_M(x)} \right\} \\
&\qquad\qquad\times R^\text{cl}_{k-1}\left( F^{\otimes (k-1)}; \frac{\delta_\text{L} F}{\delta \phi_N(y)} \right) \bigg] \total x \total y
\end{splitequation}
with
\begin{splitequation}
\mathcal{N}^{n+1}_{MN}(x,y) &= 2 \frac{\delta_\text{R} Q_\text{cl}}{\delta \phi_M(x)} R^\text{cl}_n\left( F^{\otimes n}; \frac{\delta_\text{L} F}{\delta \phi_N(y)} \right) - 2 R^\text{cl}_n\left( F^{\otimes n}; \frac{\delta_\text{R} Q_\text{cl}}{\delta \phi_M(x)} \frac{\delta_\text{L} F}{\delta \phi_N(y)} \right) \\
&\quad+ n R^\text{cl}_{n-1}\bigg[ F^{\otimes (n-1)}; 2 R^\text{cl}_1\left( F; \frac{\delta_\text{R} Q_\text{cl}}{\delta \phi_M(x)} \frac{\delta_\text{L} F}{\delta \phi_N(y)} \right) \\
&\hspace{4em}- \left\{ F, \frac{\delta_\text{R} Q_\text{cl}}{\delta \phi_M(x)} \frac{\delta_\text{L} F}{\delta \phi_N(y)} \right\} \\
&\hspace{4em}- \frac{\delta_\text{R} Q_\text{cl}}{\delta \phi_M(x)} \left[ R^\text{cl}_1\left( \frac{\delta_\text{L} F}{\delta \phi_N(y)}; F \right) + R^\text{cl}_1\left( F; \frac{\delta_\text{L} F}{\delta \phi_N(y)} \right) \right] \bigg] \eqend{.}
\end{splitequation}

Using the factorisation~\eqref{sec_aqft_rcl_factorisation} of classical retarded products, we obtain
\begin{splitequation}
&\mathcal{N}^{n+1}_{MN}(x,y) = 2 \frac{\delta_\text{R} Q_\text{cl}}{\delta \phi_M(x)} R^\text{cl}_n\left( F^{\otimes n}; \frac{\delta_\text{L} F}{\delta \phi_N(y)} \right) - 2 R^\text{cl}_n\left( F^{\otimes n}; \frac{\delta_\text{R} Q_\text{cl}}{\delta \phi_M(x)} \frac{\delta_\text{L} F}{\delta \phi_N(y)} \right) \\
&\qquad\quad+ n R^\text{cl}_{n-1}\left[ F^{\otimes (n-1)}; \left[ R^\text{cl}_1\left( F; \frac{\delta_\text{R} Q_\text{cl}}{\delta \phi_M(x)} \right) + R^\text{cl}_1\left( \frac{\delta_\text{R} Q_\text{cl}}{\delta \phi_M(x)}; F \right) \right] \frac{\delta_\text{L} F}{\delta \phi_N(y)} \right] \\
&\quad= \sum_{k=1}^n \frac{n!}{k! (n-k)!} \bigg[ k R^\text{cl}_{k-1}\left[ F^{\otimes (k-1)}; R^\text{cl}_1\left( F; \frac{\delta_\text{R} Q_\text{cl}}{\delta \phi_M(x)} \right) + R^\text{cl}_1\left( \frac{\delta_\text{R} Q_\text{cl}}{\delta \phi_M(x)}; F \right) \right] \\
&\hspace{9em}- 2 R^\text{cl}_k\left( F^{\otimes k}; \frac{\delta_\text{R} Q_\text{cl}}{\delta \phi_M(x)} \right) \bigg] R^\text{cl}_{n-k}\left( F^{\otimes (n-k)}; \frac{\delta_\text{L} F}{\delta \phi_N(y)} \right) \eqend{.}
\end{splitequation}
Since $Q_\text{cl}$ is at most quadratic in the fields, its first functional derivative is at most linear, and we can use the relation~\eqref{sec_aqft_rcl_glinear_rel}
\begin{equation}
R^\text{cl}_k\left( F^{\otimes k}; \frac{\delta_\text{R} Q_\text{cl}}{\delta \phi_M(x)} \right) = k R^\text{cl}_{k-1}\left[ F^{\otimes (k-1)}; R^\text{cl}_1\left( F; \frac{\delta_\text{R} Q_\text{cl}}{\delta \phi_M(x)} \right) \right] \eqend{.}
\end{equation}
This gives
\begin{splitequation}
\mathcal{N}^{n+1}_{MN}(x,y) &= \sum_{k=1}^n \frac{n!}{(k-1)! (n-k)!} R^\text{cl}_{k-1}\left( F^{\otimes (k-1)}; \left\{ \frac{\delta_\text{R} Q_\text{cl}}{\delta \phi_M(x)}, F \right\} \right) \\
&\qquad\times R^\text{cl}_{n-k}\left( F^{\otimes (n-k)}; \frac{\delta_\text{L} F}{\delta \phi_N(y)} \right) \\
&= \sum_{k=1}^n \frac{n!}{(k-1)! (n-k)!} \iint \Delta_{KL}(u,v) \frac{\delta_\text{R}^2 Q_\text{cl}}{\delta \phi_K(u) \delta \phi_M(x)} \\
&\qquad\times R^\text{cl}_{k-1}\left( F^{\otimes (k-1)}; \frac{\delta_\text{L} F}{\delta \phi_L(v)} \right) R^\text{cl}_{n-k}\left( F^{\otimes (n-k)}; \frac{\delta_\text{L} F}{\delta \phi_N(y)} \right) \total u \total v \eqend{,}
\end{splitequation}
where we could take the second functional derivative of $Q_\text{cl}$ outside the classical retarded product since it is field-independent by assumption. Similarly, we calculate using that the classical retarded products commute with functional derivatives~\eqref{sec_aqft_rcl_funcder} and the GLZ relation~\eqref{sec_aqft_rcl_glz}
\begin{splitequation}
&\left\{ R^\text{cl}_{n-k}\left( F^{\otimes (n-k)}; F \right), \frac{\delta_\text{R} Q_\text{cl}}{\delta \phi_M(x)} \right\} \\
&\quad= - \iint \Delta_{KL}(u,v) \frac{\delta_\text{R}^2 Q_\text{cl}}{\delta \phi_K(u) \delta \phi_M(x)} \frac{\delta_\text{L}}{\delta \phi_L(v)} R^\text{cl}_{n-k}\left( F^{\otimes (n-k)}; F \right) \total u \total v \\
&\quad= - \iint \Delta_{KL}(u,v) \frac{\delta_\text{R}^2 Q_\text{cl}}{\delta \phi_K(u) \delta \phi_M(x)} \bigg[ (n-k+1) R^\text{cl}_{n-k}\left( F^{\otimes (n-k)}; \frac{\delta_\text{L} F}{\delta \phi_L(v)} \right) \\
&\qquad+ \sum_{\ell=1}^{n-k} \frac{(n-k)!}{(n-k-\ell)! (\ell-1)!} \left\{ R^\text{cl}_{\ell-1}\left( F^{\otimes (\ell-1)}; \frac{\delta_\text{L} F}{\delta \phi_L(v)} \right), R^\text{cl}_{n-k-\ell}\left( F^{\otimes (n-k-\ell)}; F \right) \right\} \bigg] \total u \total v \eqend{.}
\end{splitequation}
Factorising the remaining Poisson bracket and renaming and shifting summation indices, it follows that
\begin{splitequation}
\mathcal{N}_{n+1} &= \iiiint \Delta_{MN}(x,y) \Delta_{KL}(u,v) \frac{\delta_\text{R}^2 Q_\text{cl}}{\delta \phi_K(u) \delta \phi_M(x)} \bigg[ \\
&\qquad+ \sum_{k=1}^n \frac{n! (n+1-2k)}{(k-1)! (n-k)!} R^\text{cl}_{k-1}\left( F^{\otimes (k-1)}; \frac{\delta_\text{L} F}{\delta \phi_L(v)} \right) R^\text{cl}_{n-k}\left( F^{\otimes (n-k)}; \frac{\delta_\text{L} F}{\delta \phi_N(y)} \right) \\
&\qquad- \sum_{\ell=1}^n \sum_{k=1+\ell}^n \frac{n!}{(\ell-1)! (k-1-\ell)! (n-k)!} \bigg[ R^\text{cl}_{\ell-1}\left( F^{\otimes (\ell-1)}; \frac{\delta_\text{L} F}{\delta \phi_L(v)} \right) \\
&\hspace{8em}\times \bigg\{ R^\text{cl}_{n-k}\left( F^{\otimes (n-k)}; F \right), R^\text{cl}_{k-1-\ell}\left( F^{\otimes (k-1-\ell)}; \frac{\delta_\text{L} F}{\delta \phi_N(y)} \right) \bigg\} \\
&\hspace{4em}- \left\{ R^\text{cl}_{n-k}\left( F^{\otimes (n-k)}; F \right), R^\text{cl}_{k-\ell-1}\left( F^{\otimes (k-\ell-1)}; \frac{\delta_\text{L} F}{\delta \phi_L(v)} \right) \right\} \\
&\hspace{8em}\times R^\text{cl}_{\ell-1}\left( F^{\otimes (\ell-1)}; \frac{\delta_\text{L} F}{\delta \phi_N(y)} \right) \bigg] \bigg] \total x \total y \total u \total v \eqend{.}
\raisetag{2em}
\end{splitequation}
This expression now vanishes by symmetry, using that $\Delta_{KL} = 0$ if $\epsilon_K \neq \epsilon_L$: if at least one of $\phi_K$ and $\phi_M$ is bosonic, both the functional derivatives acting on $Q_\text{cl}$ and the retarded products involving $\delta_\text{L} F/\delta \phi_L(v)$ and $\delta_\text{L} F/\delta \phi_N(y)$ commute, while otherwise both anticommute, and no extra sign is introduced either way. In the double sum the two terms thus cancel each other, while we take the single sum over $k$ twice, change the summation index $k \to n+1-k$ in the second one, and then the sum vanishes owing to the factor $(n+1-2k)$.
\end{proof}

\section{Quantum gauge theories}
\label{sec_gauge}

The main problem in applying the general apparatus of pAQFT (as described in section~\ref{sec_aqft}) directly to gauge theories is that no unique retarded and advanced Green's functions~\eqref{sec_aqft_retadv} exist because of the gauge invariance. A very general approach to treat this problem is the BRST/BV (Becchi--Rouet--Stora--Tyutin/Batalin--Vilkovisky) or field--antifield formalism~\cite{becchirouetstora1976,batalinvilkovisky1981,batalinvilkovisky1983,batalinvilkovisky1984}, in which one introduces extra fields and extends the action in such a way that unique retarded and advanced Green's functions can be found, and then recovers the original (classical) observables as elements of the cohomology of a nilpotent fermionic operator, the \emph{BRST differential} $\brst$. This structure must then be transported to the quantum theory; in particular $\brst$ must be promoted to a graded differential $\st$ on the free-field algebra $\overline{\mathfrak{A}}_0$, such that the quantum observables may be obtained as elements of the cohomology of $\st$. That this programme can be successfully executed is not automatic, and it might be impossible to perform the extension $\brst \to \st$. In this case there is no analogue of the classical symmetry in the quantum theory, it is anomalous.

We begin with the construction of the classical extended theory. Again, since extensive reviews of the BRST/BV formalism exist~\cite{henneaux1990,gomisparissamuel1995,barnichetal2000,fredenhagenrejzner2012b}, we are brief and only list the essential steps. For each symmetry transformation $\delta_\xi$ with parameter $\xi_M$ acting on the fields $\{ \phi_K \}$, we introduce a \emph{ghost field} $c_M$, an \emph{antighost field} $\bar{c}_M$ (where the bar is purely notational) and an \emph{auxiliary} or \emph{Nakanishi--Lautrup field}~\cite{nakanishi1966,lautrup1967} $B_M$. Note that for fermionic symmetries such as supersymmetry, the parameter $\xi$ must be taken to be fermionic such that $\delta_\xi$ does not change the Grassmann parity. For global symmetries~\cite{brandthenneauxwilch1996,brandthenneauxwilch1998}, the antighost and auxiliary field (the \emph{non-minimal fields}) can be dispensed with as we will see later, and for reducible symmetries (which become symmetries of the ghost fields) one has to repeat the procedure, leading to ``ghosts for ghosts''~\cite{townsend1979,namaziestorey1980,thierrymieg1990,siegel1980,kimura1980,kimura1981}. The original fields of the theory together with the ghosts, antighosts and auxiliary fields are then grouped together and denoted by $\{ \Phi_K \} = \{ \phi_L, c_M, \bar{c}_M, B_M \}$, and for each of them an antifield $\Phi^\ddag_K$ is introduced. Apart from the Grassmann parity $\epsilon$, one assigns two additional gradings (quantum numbers), the ghost number $g$ and the antifield number $a$, to the fields and monomials thereof (with the grading of a product being the sum of the gradings of its factors) in the following way:
\begin{splitequation}
&\epsilon(c_K) = \epsilon(\bar{c}_K) = \epsilon(\xi_K) + 1 \eqend{,} \quad \epsilon(B_K) = \epsilon(\xi_K) \eqend{,} \quad \epsilon(\Phi^\ddag_K) = \epsilon(\Phi_K) + 1 \eqend{,} \\
&g(\phi_K) = 0 \eqend{,} \quad g(c_K) = 1 \eqend{,} \quad g(\bar{c}_K) = -1 \eqend{,} \quad g(B_K) = 0 \eqend{,} \\
&g(\Phi^\ddag_K) = - 1 - g(\Phi_K) \eqend{,} \quad a(\Phi_K) = 0 \eqend{,} \quad a(\Phi^\ddag_K) = 1 \eqend{.}
\end{splitequation}
In the space $\mathcal{F}$ of local smeared field polynomials~\eqref{sec_aqft_calf_expr} one now has to include all fields $\Phi_K$ and antifields $\Phi^\ddag_K$. For $F,G \in \mathcal{F}$ one defines the \emph{antibracket}
\begin{equation}
\label{sec_gauge_antibracket_def}
\left( F, G \right) \equiv \int \left( \frac{\delta_\text{R} F}{\delta \Phi_K(x)} \frac{\delta_\text{L} G}{\delta \Phi^\ddag_K(x)} - \frac{\delta_\text{R} F}{\delta \Phi^\ddag_K(x)} \frac{\delta_\text{L} G}{\delta \Phi_K(x)} \right) \total x \eqend{.}
\end{equation}
Using the properties of left- and right-differentiation~\eqref{sec_aqft_leftrightder_props}, one checks that it is graded symmetric
\begin{equation}
\label{sec_gauge_antibracket_gradsym}
\left( F, G \right) = (-1)^{\epsilon_F + \epsilon_G + \epsilon_F \epsilon_G} \left( G, F \right) \eqend{,}
\end{equation}
commutes with functional derivatives
\begin{equations}[sec_gauge_antibracket_funcder]
\frac{\delta_\text{L}}{\delta \Phi_K(x)} \left( F, G \right) &= \left( \frac{\delta_\text{L} F}{\delta \Phi_K(x)}, G \right) + (-1)^{(1+\epsilon_F) \epsilon_K} \left( F, \frac{\delta_\text{L} G}{\delta \Phi_K(x)} \right) \eqend{,} \\
\frac{\delta_\text{R}}{\delta \Phi_K(x)} \left( F, G \right) &= \left( F, \frac{\delta_\text{R} G}{\delta \Phi_K(x)} \right) + (-1)^{(1+\epsilon_G) \epsilon_K} \left( \frac{\delta_\text{R} F}{\delta \Phi_K(x)}, G \right) \eqend{,}
\end{equations}
fulfils a graded Leibniz rule
\begin{equation}
\label{sec_gauge_antibracket_leibniz}
\left( F, G H \right) = \left( F, G \right) H + (-1)^{(1+\epsilon_F) \epsilon_G} G \left( F, H \right) \eqend{,}
\end{equation}
and fulfils the graded Jacobi identity
\begin{splitequation}
\label{sec_gauge_antibracket_jacobi}
&(-1)^{(\epsilon_F+1)(\epsilon_H+1)} \left( F, \left( G, H \right) \right) + (-1)^{(\epsilon_G+1)(\epsilon_F+1)} \left( G, \left( H, F \right) \right) \\
&\quad+ (-1)^{(\epsilon_H+1)(\epsilon_G+1)} \left( H, \left( F, G \right) \right) = 0 \eqend{.}
\end{splitequation}
Furthermore, the antibracket respects the grading: it is of ghost number $1$, antifield number $-1$ and Grassmann odd:
\begin{splitequation}
\label{sec_gauge_antibracket_grading}
&g\left[ \left( F,G \right) \right] = g(F)+g(G)+1 \eqend{,} \quad a\left[ \left( F,G \right) \right] = a(F)+a(G)-1 \eqend{,} \\
&\epsilon\left[ \left( F,G \right) \right] = \epsilon(F)+\epsilon(G)+1 \eqend{.}
\end{splitequation}

Using the antibracket, one defines the \emph{BRST differential}
\begin{equation}
\brst F \equiv \left( S_\text{tot}, F \right) \eqend{,}
\end{equation}
where the total action $S_\text{tot}$ is the sum of the original action $S$ and the extension $S_\text{ext}$, which needs to be chosen such that
\begin{enumerate}
\item the \emph{BV master equation} $\left( S_\text{tot}, S_\text{tot} \right) = 0$ is fulfilled,
\item the original symmetries are recovered by a BRST transformation, with the transformation parameter replaced by the ghost: \\ $\brst \phi_M = \sum_k \delta_{c} \phi_M + \text{terms containing antifields}$,
\item the non-minimal fields form \emph{trivial pairs}: $\brst \bar{c}_M = B_M$, $\brst B_M = 0$,
\item the (formally self-adjoint) differential operator $P_{KL}$ appearing in the antifield-independent free part of the total action (quadratic in fields)
\begin{equation}
\label{sec_gauge_freeaction_fieldspkl}
S^{(0)}_{\text{tot},0} = \frac{1}{2} \int \Phi_K(x) P_{KL}(x) \Phi_L(x) \total x
\end{equation}
needs to possess unique retarded and advanced Green's functions~\eqref{sec_aqft_retadv}.
\end{enumerate}
The last condition is of course the reason for the introduction of the antifield formalism, and to satisfy it one needs the non-minimal fields. Since global symmetries do not affect the invertibility of the differential operators $P_{KL}$ and thus the (non-)existence of retarded and advanced Green's functions, no non-minimal fields are necessary for them as stated previously. However, one can still treat them in a unified way with the gauge symmetries by introducing constant ghosts while leaving $P_{KL}$ unchanged, for example for supersymmetric theories. These constant ghosts have associated constant antifields that appear in the extended action in the same way as for normal ghosts, and ensure the validity of the BV master equation. However, for consistency one must replace functional derivatives with respect to them by ordinary derivatives (but still fulfilling the graded Leibniz rules~\eqref{sec_aqft_leftrightder_props}), for example in the definition of the antibracket~\eqref{sec_gauge_antibracket_def}. Because the antibracket is Grassmann odd, $\brst$ is an odd (fermionic) differential, and using the Leibniz rule~\eqref{sec_gauge_antibracket_leibniz} one sees that it is a graded left derivation. Using the Jacobi identity~\eqref{sec_gauge_antibracket_jacobi}, one finds that the BV master equation is equivalent to its nilpotency $\brst^2 = 0$, and since it also respects the grading by ghost number (augmenting it by $1$) we can define the cohomology classes at ghost number $g$:\footnote{One usually considers the cohomology for \emph{unsmeared} local field polynomials, but the extension to smeared ones is obvious.}
\begin{equation}
\label{sec_gauge_cohom_def}
H^g(\brst) \equiv \frac{\mathrm{Ker}(\brst\colon \mathcal{F}^g \to \mathcal{F}^{g+1})}{\mathrm{Im}(\brst\colon \mathcal{F}^{g-1} \to \mathcal{F}^g)} \eqend{,}
\end{equation}
where $\mathcal{F}^g$ is the subspace of $\mathcal{F}$ of homogeneous elements of ghost number $g$:
\begin{equation}
F \in \mathcal{F}^g \Leftrightarrow g(F) = g \eqend{,} \qquad \mathcal{F} = \bigoplus_{g \in \mathbb{N}} \mathcal{F}^g \eqend{.}
\end{equation}
Since one can choose representatives independent of trivial pairs~\cite{barnichetal2000,piguetsorella} and the ghost fields have positive ghost number, the condition that the non-minimal fields form such pairs ensures that all representatives of elements in $H^0(\brst)$ which are antifield-independent are indeed gauge-invariant classical observables. That all observables can be obtained in this way, and that one can indeed choose the representatives to be independent of antifields is a more subtle point, and needs to be checked separately for each specific theory.

An extension $S_\text{ext}$ can be constructed in two steps: first finding a solution to conditions 1 and 2 ignoring the non-minimal fields (the minimal extension), and then performing a canonical transformation on this solution to also satisfy conditions 3 and 4. The second step is easier; assume that an extension $S_\text{ext}$ only depending on the original fields, the ghosts and their antifields has been found such that the first two conditions are satisfied. Condition 3 can then be satisfied by replacing
\begin{equation}
\label{sec_gauge_nonminimal}
S_\text{ext} \to S_\text{ext} - \int B_M \bar{c}^\ddag_M \total x \eqend{,}
\end{equation}
i.e., the non-minimal extension is obtained by adding the trivial pair for each symmetry to the minimal extension. Since the minimal extension does not depend on the non-minimal fields and antifields, conditions 1 and 2 are maintained by this addition, and it is easy to check that condition 3 now also holds. One then chooses a so-called \emph{gauge-fixing fermion} $\Psi \in \mathcal{F}^{-1}$ with $\epsilon(\Psi) = 1$ depending on the fields $\Phi_K$ and new antifields $\hat{\Phi}^\ddag_K$ such that the system
\begin{equation}
\label{sec_gauge_canonicaltrafo}
\hat{\Phi}_K(x) = \Phi_K(x) - \frac{\delta_\text{L} \Psi}{\delta \hat{\Phi}^\ddag_K(x)} \eqend{,} \qquad \Phi^\ddag_K(x) = \hat{\Phi}^\ddag_K(x) - \frac{\delta_\text{R} \Psi}{\delta \Phi_K(x)}
\end{equation}
can be solved for $\Phi_K$. Replacing $\Phi_K$ and $\Phi^\ddag_K$, one checks that the new solution $\hat{S}_\text{tot}\left( \hat{\Phi}_K, \hat{\Phi}^\ddag_K \right) = S_\text{tot}\left( \Phi_K, \Phi^\ddag_K \right)$ still satisfies condition 1 since the transformation~\eqref{sec_gauge_canonicaltrafo} is a canonical transformation that leaves the antibracket unchanged~\cite{batalinvilkovisky1984b}. Conditions 2 and 3 are in general only maintained if we restrict $\Psi$ to be independent of antifields; antifield-dependent terms in $\Psi$ correspond to adding equation-of-motion terms to the symmetry transformations $\delta_{\xi}$ and mixing between them~\cite{batalinvilkovisky1984b}. If $\Psi$ is properly chosen, condition 4 is also satisfied, and in particular the choice $\Psi = \bar{c}_M \Psi_M\left( \phi_K \right)$ corresponds to imposing the gauge conditions $\Psi_M = 0$ strictly. Finally, the solution to conditions 1 and 2 can in general only be constructed as a formal power series in antifields $S_\text{ext} = \sum_{k=1}^\infty S_\text{ext}^\text{(k)}$ with $S_\text{ext}^\text{(k)}$ of order $k$ in antifields. The lowest-order terms of the extension are fixed~\cite{batalinvilkovisky1983,batalinvilkovisky1984,batalinvilkovisky1985}, and we have
\begin{equation}
\label{sec_gauge_sextlin}
S_\text{ext}^{(1)} = - \int \left[ \delta_c \phi_M \right] \phi_M^\ddag \total x - \int K_M c_M^\ddag \total x \eqend{,}
\end{equation}
where we set $\delta_c \equiv \int c_M(y) \delta_\text{L} / \delta \xi_M(y) \total y \, \delta_\xi$, and $K_M$ is determined such that
\begin{equation}
\label{sec_gauge_kmdef}
\left( \delta_c + \int K_M(y) \frac{\delta_\text{L}}{\delta c_M(y)} \total y \right) \delta_c \phi_N = 0
\end{equation}
holds, up to terms which vanish by the equations of motion for the original fields $\phi_K$. If no such terms arise, one speaks of a closed gauge algebra and the extension $S_\text{ext}^{(1)}$ is sufficient to fulfil condition 1 (as can be checked straightforwardly). Otherwise, one has an open gauge algebra and needs to recursively add terms with higher and higher powers of antifields (starting with the quadratic term). A formal solution always exists~\cite{batalinvilkovisky1985}, and in many specific theories the recursion must terminate because of dimensional constraints (e.g., when all antifields and background fields have positive engineering dimension). The functions $K_M$ determined by~\eqref{sec_gauge_kmdef} are nothing else but the structure functions of the gauge algebra; if the gauge algebra is open (and thus does not form a Lie algebra) the higher-order structure functions are determined from similar conditions involving antifields.

Expanding $S_\text{tot}$ with respect to the grading by antifield number,
\begin{equation}
S_\text{tot} = \sum_{k=0}^\infty S_\text{tot}^{(k)} \eqend{,}
\end{equation}
since the antibracket respects the grading~\eqref{sec_gauge_antibracket_grading} one obtains a corresponding expansion of the BRST differential:
\begin{equation}
\brst = \sum_{k=-1}^\infty \brst^{(k)} \eqend{,} \qquad \brst^{(k)} F = \left( S_\text{tot}^{(k+1)}, F \right) \eqend{.}
\end{equation}
Expanding the nilpotency relation $\brst^2 = 0$, it follows that
\begin{equation}
\label{sec_gauge_brstnilpot_expanded}
0 = \sum_{\ell=-1}^k \brst^{(\ell)} \brst^{(k-\ell-1)} \quad (k = -1,0,\ldots) \eqend{.}
\end{equation}
The lowest-order component $\brst^{(-1)}$ annihilates fields and transforms antifields into the (gauge-fixed) equations of motion of the corresponding field; it is known as the \emph{Koszul--Tate differential} and is nilpotent itself, following from the $k = -1$ term of the expanded nilpotency relation~\eqref{sec_gauge_brstnilpot_expanded}. The next component $\brst^{(0)}$ (called the longitudinal exterior derivative along the gauge orbits) is not a differential in general because the $k = 1$ term of the expanded nilpotency relation~\eqref{sec_gauge_brstnilpot_expanded} reads
\begin{equation}
\left( \brst^{(0)} \right)^2 = - \brst^{(-1)} \brst^{(1)} - \brst^{(1)} \brst^{(-1)} \neq 0 \eqend{.}
\end{equation}
However, for a closed gauge algebra we have $S_\text{tot}^{(2)} = 0$, thus $\brst^{(1)} = 0$ and in this case $\brst^{(0)}$ is also a differential, whose action on fields coincides with the original BRST differential of~\cite{becchirouetstora1976,tyutin1975}, before the introduction of the BV formalism.

Lastly, we note that the BV master equation can be rewritten as a condition for the existence of a divergence-free current. This can be done as follows: one rewrites the antibracket of two functionals as
\begin{equation}
\label{sec_gauge_antibracket_rewrite}
(F,G) = \int \left[ \left( F, \Phi^\ddag_K(x) \right) \left( \Phi^{\vphantom{\ddag}}_K(x), G \right) + \left( G, \Phi^\ddag_K(x) \right) \left( \Phi^{\vphantom{\ddag}}_K(x), F \right) \right] \total x \eqend{.}
\end{equation}
The fulfillment of the BV master equation $\left( S_\text{tot}, S_\text{tot} \right) = 0$ is thus equivalent to
\begin{equation}
\label{sec_gauge_brst_current}
\frac{\delta_\text{R} S_\text{tot}}{\delta \Phi_K(x)} \frac{\delta_\text{L} S_\text{tot}}{\delta \Phi^\ddag_K(x)} = \left( S_\text{tot}, \Phi^\ddag_K(x) \right) \left( \Phi^{\vphantom{\ddag}}_K(x), S_\text{tot} \right) = \nabla_\mu J^\mu(x)
\end{equation}
for some $J^\mu$ constructed locally from fields and antifields, which is nothing else but the Noether current associated to the BRST symmetry~\cite{hollands2008,taslimitehrani2017}.

\begin{example*}
To illustrate the general theory, we consider Yang--Mills theory~\cite{yangmills1954} where the basic field is a bosonic Lie algebra-valued one form, and we take its components $A_\mu^a$ with respect to a fixed basis in a given representation of the (semi-simple) Lie algebra and a fixed basis of one-forms on spacetime. The action is given by
\begin{equation}
S = - \frac{1}{4} \int F_{\mu\nu}^a F^{\mu\nu a} \total x
\end{equation}
with the field strength tensor $F_{\mu\nu}^a \equiv \nabla_\mu A_\nu^a - \nabla_\nu A_\mu^a + \mathi g f_{abc} A_\mu^b A_\nu^c$ and the coupling constant $g$. We have normalised the basis elements such that the Cartan--Killing form is the identity (so that Lie algebra indices $a,b,c,\ldots$ are summed over regardless of their position), and $f_{abc}$ are the totally antisymmetric structure constants in this basis that fulfil the Jacobi identity
\begin{equation}
f_{abs} f_{cds} - f_{acs} f_{bds} + f_{ads} f_{bcs} = 0 \eqend{.}
\end{equation}
The action is invariant under the symmetry transformation
\begin{equation}
\delta_\xi A_\mu^a = \left( D_\mu \xi \right)^a \equiv \nabla_\mu \xi^a + \mathi g f_{abc} A_\mu^b \xi^c
\end{equation}
with a Lie algebra-valued parameter $\xi$ expanded in the same basis, and one checks that the choice
\begin{equation}
K^a = - \frac{1}{2} \mathi g f_{abc} c^b c^c
\end{equation}
satisfies equation~\eqref{sec_gauge_kmdef} without any extra terms. The minimal extension thus reads
\begin{equation}
S_\text{ext} = S_\text{ext}^{(1)} = - \int \left( D^\mu c \right)^a A_\mu^{a\ddag} \total x + \frac{\mathi}{2} g f_{abc} \int c^b c^c c_a^\ddag \total x \eqend{,}
\end{equation}
and adding the non-minimal terms~\eqref{sec_gauge_nonminimal} and choosing a gauge-fixing fermion $\Psi$ that only depends on fields, the total action after the transformation~\eqref{sec_gauge_canonicaltrafo} reads
\begin{splitequation}
S_\text{tot} &= - \frac{1}{4} \int F_{\mu\nu}^a F^{\mu\nu a} \total x + \int \left( D_\mu c \right)^a \frac{\delta_\text{R} \Psi}{\delta A_\mu^a} \total x - \frac{\mathi}{2} g f_{abc} \int c^b c^c \frac{\delta_\text{R} \Psi}{\delta c^a} \total x \\
&\quad+ \int B_a \frac{\delta_\text{R} \Psi}{\delta \bar{c}^a} \total x - \int \left( D^\mu c \right)^a A_\mu^{a\ddag} \total x + \frac{\mathi}{2} g f_{abc} \int c^b c^c c_a^\ddag \total x - \int B_a \bar{c}^\ddag_a \total x \eqend{.}
\end{splitequation}
A suitable choice for $\Psi$ is given by
\begin{equation}
\Psi = \int \bar{c}_a \left( \frac{\xi}{2} B^a - \nabla^\mu A_\mu^a \right) \total x
\end{equation}
with a parameter $\xi \in \mathbb{R}$, and the total action reduces to
\begin{splitequation}
S_\text{tot} &= \frac{1}{2} \int A_\mu^a \left( g^{\mu\nu} \nabla^2 - R^{\mu\nu} - \frac{\xi-1}{\xi} \nabla^\mu \nabla^\nu \right) A_\nu^a \total x + \int \bar{c}_a \nabla^2 c^a \total x \\
&\quad+ \frac{\xi}{2} \int \left( B^a - \frac{1}{\xi} \nabla^\mu A_\mu^a \right) \left( B^a - \frac{1}{\xi} \nabla^\nu A_\nu^a \right) \total x \\
&\quad- \mathi g f_{abc} \int \left[ A_\mu^b A_\nu^c \nabla^\mu A^\nu_a + \left( \nabla^\mu \bar{c}_a \right) A_\mu^b c^c + \frac{\mathi}{4} g f_{ade} A_\mu^b A_\nu^c A^\mu_d A^\nu_e \right] \total x \\
&\quad- \int \left( D^\mu c \right)^a A_\mu^{a\ddag} \total x + \frac{\mathi}{2} g f_{abc} \int c^b c^c c_a^\ddag \total x - \int B_a \bar{c}^\ddag_a \total x \eqend{.}
\end{splitequation}
The antifield-independent part of the free action $S^{(0)}_{\text{tot},0}$ (the first two lines) now involves a differential operator that is invertible and possesses unique retarded and advanced Green's functions for all $\abs{\xi} < \infty$~\cite{froebtaslimitehrani2018}. The construction of quantum Yang--Mills theory in~\cite{hollands2008} was done in Feynman gauge $\xi = 1$.
\end{example*}

\begin{example*}
The second example is $\mathcal{N} = 1$ Super-Yang--Mills theory in flat space~\cite{ferrarazumino1974,salamstrathdee1974,dewitfreedman1975}, which in addition to the Yang--Mills vector boson contains a Majorana spinor $\chi^a$ in the adjoint representation. The action reads
\begin{equation}
S = - \frac{1}{4} \int F_{\mu\nu}^a F^{\mu\nu a} \total x - \frac{1}{2} \int \bar{\chi}^a \gamma^\mu \left( D_\mu \chi \right)^a \total x \eqend{,}
\end{equation}
and is invariant under the gauge transformation
\begin{equation}
\delta^\text{gauge}_\xi A_\mu^a = \left( D_\mu \xi \right)^a \eqend{,} \qquad \delta^\text{gauge}_\xi \chi^a = - \mathi g f_{abc} \xi^b \chi^c
\end{equation}
and the supersymmetry transformation
\begin{equation}
\delta^\text{susy}_\epsilon A_\mu^a = - \bar{\epsilon} \gamma_\mu \chi^a \eqend{,} \qquad \delta^\text{susy}_\epsilon \chi^a = \frac{1}{2} \gamma^\mu \gamma^\nu F_{\mu\nu}^a \epsilon \eqend{,}
\end{equation}
where $\epsilon$ is a constant Grassmann-odd Majorana spinor. To check invariance under the supersymmetry transformation, one needs to use Fierz rearrangement identities~\cite{fierz1937,freedmanvanproeyen}\footnote{Fierz actually attributes these identities to Pauli: ``The content of this 1. section originates from Prof. W. Pauli and I am indebted to him for ceding me his calculations.''~\cite{fierz1937}, footnote 2.} and the Bianchi identities for the field strength tensor; we omit the long but essentially straightforward calculation. Two supersymmetry transformations close into a gauge transformation and a translation plus equation-of-motion terms,
\begin{equations}
\left[ \delta^\text{susy}_2, \delta^\text{susy}_1 \right] A_\mu^a &= - 2 F_{\mu\nu}^a \bar{\epsilon}_1 \gamma^\nu \epsilon_2 = b^\nu \nabla_\nu A_\mu^a + \delta_\xi A_\mu^a \eqend{,} \\
\begin{split}
\left[ \delta^\text{susy}_2, \delta^\text{susy}_1 \right] \chi^a &= b^\nu \nabla_\nu \chi^a + \delta_\xi \chi^a - \frac{1}{4} \left( b^\rho \gamma_\rho + \bar{\epsilon}_1 \gamma^{[\rho} \gamma^{\sigma]} \epsilon_2 \gamma_\rho \gamma_\sigma \right) \gamma^\nu \left( D_\nu \chi \right)^a
\end{split}
\end{equations}
with the translation parameter $b^\nu$ and the field-dependent gauge transformation parameter $\xi^a$ defined by
\begin{equation}
b^\nu \equiv 2 \bar{\epsilon}_1 \gamma^\nu \epsilon_2 \eqend{,} \qquad \xi^a \equiv - b^\nu A_\nu^a \eqend{.}
\end{equation}
(Again, Fierz rearrangement identities are needed to check this.) Therefore, we need to introduce three ghosts: the gauge ghost $c^a$ (a scalar of odd Grassmann parity), the supersymmetry ghost $\theta$ (a constant Majorana spinor of even Grassmann parity), and the translation ghost $\alpha^\mu$ (a constant vector of odd Grassmann parity). The total symmetry transformation is then the sum of these three:
\begin{equation}
\delta_c \phi_K = \delta^\text{gauge}_c \phi_K + \delta^\text{susy}_\theta \phi_K + \alpha^\mu \nabla_\mu \phi_K \eqend{.}
\end{equation}
Equation~\eqref{sec_gauge_kmdef} can be satisfied for $A_\mu^a$ by taking
\begin{equations}[sec_gauge_example_sym_k]
K_{c^a} &= - \frac{\mathi}{2} g f_{abc} c^b c^c + \bar{\theta} \gamma^\rho \theta A_\rho^a + \alpha^\rho \nabla_\rho c^a \eqend{,} \\
K_\theta &= 0 \eqend{,} \\
K_{\alpha^\rho} &= - \bar{\theta} \gamma^\rho \theta \eqend{,}
\end{equations}
but acting on $\chi^a$ an equation-of-motion term remains:
\begin{equation}
\left( \delta_c + \int K_M(y) \frac{\delta_\text{L}}{\delta c_M(y)} \total y \right) \delta_c \chi^a = - \frac{1}{8} \left( 2 \bar{\theta} \gamma^\rho \theta \gamma_\rho + \bar{\theta} \gamma^{[\rho} \gamma^{\sigma]} \theta \gamma_\rho \gamma_\sigma \right) \gamma^\mu \left( D_\mu \chi \right)^a \eqend{.}
\end{equation}
We therefore need to add a term quadratic in antifields (and possibly higher-order terms) to the extended action; one checks that
\begin{equation}
S_\text{ext} = S_\text{ext}^{(1)} + \frac{1}{16} \int \chi_a^\ddag \left( 2 \bar{\theta} \gamma^\rho \theta \gamma_\rho + \bar{\theta} \gamma^{[\rho} \gamma^{\sigma]} \theta \gamma_\rho \gamma_\sigma \right) \bar{\chi}_a^\ddag \total x
\end{equation}
is sufficient to fulfil the BV master equation $\left( S_\text{tot}, S_\text{tot} \right) = 0$, with $S_\text{ext}^{(1)}$ given by~\eqref{sec_gauge_sextlin}.
\end{example*}
\begin{example*}
As a last (somewhat curious) example, we show that one can also treat non-gauge theories in the same framework. Considering a real scalar field $\phi$ with action $S_\phi$, we introduce the antifield $\phi^\ddag$ and define the total action $S_\text{tot} = S_\phi$. The action of the BRST differential reads
\begin{equation}
\brst \phi = 0 \eqend{,} \qquad \brst \phi^\ddag = \frac{\delta_\text{R} S}{\delta \phi} \eqend{,}
\end{equation}
and the assignment of Grassmann parity and ghost numbers is
\begin{equation}
\epsilon(\phi) = 0 \eqend{,} \qquad \epsilon\left( \phi^\ddag \right) = 1 \eqend{,} \qquad g(\phi) = 0 \eqend{,} \qquad g\left( \phi^\ddag \right) = -1 \eqend{.}
\end{equation}
Since there are no fields with positive ghost number, the cohomologies at positive ghost number are empty: $H^g(\brst) = \{0\}$ for $g > 0$, and any function of $\phi$ (and its derivatives) is an element of $H^0(\brst)$. All conditions of the remaining theorems are thus fulfilled.
\end{example*}

In the quantum theory, fields and antifields are treated on the same footing. However, since the antifields are not dynamical but fixed background fields, there are no retarded and advanced Green's functions associated with them, and thus no commutator function. Consequently, they (anti-)commute with all other fields and among themselves,
\begin{equation}
\left[ \Phi_K(f), \Phi^\ddag_L(g) \right]_{\star_\hbar} = 0 = \left[ \Phi^\ddag_K(f), \Phi^\ddag_L(g) \right]_{\star_\hbar} \eqend{,}
\end{equation}
can be taken out of normal-ordered products (where the hat denotes omission),
\begin{splitequation}
&\normord{\Phi_{K_1} \cdots \Phi^\ddag_{K_\ell} \cdots \Phi_{K_n}}_G(f_1 \otimes \cdots \otimes f_n) = (-1)^{(1+\epsilon_{K_\ell}) \sum_{k=1}^{\ell-1} \epsilon_{K_k}} \\
&\quad\times \Phi^\ddag_{K_\ell}(f_\ell) \star_\hbar \normord{\Phi_{K_1} \cdots \hat{\Phi}^\ddag_{K_\ell} \cdots \Phi_{K_n}}_G(f_1 \otimes \cdots \otimes \hat{f}_\ell \otimes \cdots \otimes f_n) \eqend{,}
\end{splitequation}
and most importantly also out of time-ordered products:
\begin{splitequation}
\label{sec_gauge_antifields_timeordered}
&\mathcal{T}_n\left( F_1 \otimes \cdots \otimes F_n \right) = (-1)^{\epsilon_{F_\ell} \sum_{k=1}^{\ell-1} \epsilon_{F_k}} F_\ell \star_\hbar \mathcal{T}_{n-1}\left( F_1 \otimes \cdots \otimes \hat{F}_\ell \otimes \cdots \otimes F_n \right) \\
&\qquad= (-1)^{\epsilon_{F_\ell} \sum_{k=\ell+1}^{n} \epsilon_{F_k}} \mathcal{T}_{n-1}\left( F_1 \otimes \cdots \otimes \hat{F}_\ell \otimes \cdots \otimes F_n \right) \star_\hbar F_\ell
\raisetag{1.6em}
\end{splitequation}
if $F_\ell$ only contains antifields.

We now study the anomalous Ward identities for gauge theories, and first need an analogue of Theorem~\ref{thm_classanom}. For this, we take any derivation $D_\mathcal{F}$ acting on $\mathcal{F}$ by the antibracket with an element $Q \in \mathcal{F}$ at most of second order in fields (or antifields), whose action on basic fields and antifields is given by
\begin{equation}
\label{thm_classanom_bv_deract}
D_\mathcal{F} \Phi_K(x) = - \frac{\delta_\text{R} Q}{\delta \Phi^\ddag_K(x)} \eqend{,} \qquad D_\mathcal{F} \Phi^\ddag_K(x) = \frac{\delta_\text{R} Q}{\delta \Phi_K(x)} \eqend{.}
\end{equation}
Since the right-hand sides are at most linear in fields (or antifields), we can obtain a corresponding differential $D$ on the free-field algebra $\overline{\mathfrak{A}}_0$ by defining its action on the generators $\Phi_K(f),\Phi^\ddag_K(f)$ of $\mathfrak{A}_0$ to be given by~\eqref{thm_classanom_bv_deract} with the right-hand sides identified as (linear combinations of) generators of $\mathfrak{A}_0$, and then extending it to general $A \in \overline{\mathfrak{A}}_0$ by linearity and a graded Leibniz rule. We then obtain
\begin{theorem}
\label{thm_classanom_bv}
Given a derivation $D$ acting on the free-field algebra $\overline{\mathfrak{A}}_0$, obtained from a derivation $D_\mathcal{F}$ acting on $\mathcal{F}$ by the antibracket with an element $Q \in \mathcal{F}$ at most of second order in fields (or antifields) as described above, the following holds for the maps $\mathcal{D}_n$ appearing in the anomalous Ward identity~\eqref{thm_anomward_identity}:
\begin{enumerate}
\item At first order, we have
\begin{equation*} \mathcal{D}_1\left( F \right) = (Q,F) \eqend{.} \end{equation*}
\item At second order, we have
\begin{splitequation*} \mathcal{D}_2\left( F \otimes F \right) &= \left( Q, R^\text{cl}_1\left( F; F \right) \right) - R^\text{cl}_1\left[ F; \left( Q, F \right) \right] - R^\text{cl}_1\left[ \left( Q, F \right); F \right] \\ &= \iint \frac{\delta_\text{R} F}{\delta \Phi_K(x)} \left[ G^\text{ret}_{KL}(x,y) + G^\text{adv}_{KL}(x,y) \right] \left( \frac{\delta_\text{L} Q}{\delta \Phi_L(y)}, F \right) \total x \total y \eqend{.} \end{splitequation*}
\item $\mathcal{D}_k\left( F^{\otimes k} \right) = 0$ for all $k \geq 3$.
\end{enumerate}
\end{theorem}
\begin{proof}
The proof is completely analogous to the proof of Theorem~\ref{thm_classanom}, using that the antibracket commutes with functional derivatives~\eqref{sec_gauge_antibracket_funcder} and factorises~\eqref{sec_gauge_antibracket_leibniz}; we therefore omit the details.
\end{proof}

For gauge theories, this result is needed for the free BRST differential $\brst_0$, which is obtained by taking the antibracket with the free action $S_0$ (which is quadratic in fields and antifields): $\brst_0 F = (S_0,F)$. Expanding the free action with respect to antifield number $S_0 = S_0^{(0)} + S_0^{(1)} + S_0^{(2)}$, we obtain an expansion of the free BRST differential $\brst_0$ into the free Koszul--Tate differential $\brst_0^{(-1)}$ and the operators $\brst_0^{(0)}$ and $\brst_0^{(1)}$, acting according to $\brst_0^{(k)} F = \left( S_0^{(k+1)}, F \right)$. We would like to define corresponding derivations, denoted by $\st_0$, $\st_0^{(k)}$ and acting on $\overline{\mathfrak{A}}_0$ according to~\eqref{thm_classanom_bv_deract}. This is obviously unproblematic for basic fields (and antifields), but in order for $\st_0$ to act consistently on $\overline{\mathfrak{A}}_0$ it must be compatible with the (anti-)commutation relations~\eqref{sec_aqft_commutator}, which leads to conditions on the commutator function $\Delta_{KL}(x,y)$. These conditions can be seen as Ward identities in the free theory, and can in fact be easily derived from the BRST invariance of the free action: Expanding the BV master equation for the free action $(S_0,S_0) = 0$ in antifield number, at first order we obtain
\begin{equation}
0 = 2 \left( S_0^{(1)}, S_0^{(0)} \right) = - 2 \int \frac{\delta_\text{R} S_0^{(1)}}{\delta \Phi^\ddag_K(x)} \frac{\delta_\text{L} S_0^{(0)}}{\delta \Phi_K(x)} \total x = - 2 \int \frac{\delta_\text{R} S_0^{(1)}}{\delta \Phi^\ddag_K(x)} P_{KL} \Phi_L(x) \total x
\end{equation}
since $S_0^{(0)}$ does not depend on antifields. We write $S_0^{(1)}$ in the form
\begin{equation}
S_0^{(1)} = - \int \left( Q_{MN} \Phi_N \right) \Phi^\ddag_M \total x
\end{equation}
with a differential operator $Q_{MN}$. Since the action was assumed to be Grassmann even and the Grassmann parity of the antifield is opposite to the one of the corresponding field, it follows that $Q_{MN} = 0$ if $\epsilon_M = \epsilon_N$. We then obtain
\begin{equation}
\int \Phi_N Q^*_{KN} P_{KM} \Phi_M \total x = 0 \eqend{,}
\end{equation}
and since this must hold for all $\Phi$, the symmetric part of the differential operator must vanish:
\begin{equation}
Q^*_{KN} P_{KM} + (-1)^{\epsilon_M \epsilon_N} P^*_{KN} Q_{KM} = 0 \eqend{,}
\end{equation}
where we used that $P_{KL} = 0$ if $\epsilon_K \neq \epsilon_L$. We multiply this condition from the left by $\delta_\text{L} F/\delta \Phi_P(y) G^\text{ret}_{NP}(x,y)$ and from the right by $G^\text{adv}_{MQ}(x,z) \delta_\text{L} G/\delta \Phi_Q(z)$ for some $F,G \in \mathcal{F}$ and integrate over $x$. Because of the support properties of the retarded and advanced Green's functions, the integration domain is compact and we can integrate by parts to obtain
\begin{equation}
\frac{\delta_\text{L} F}{\delta \Phi_P(y)} \left[ Q_{QN}(z) G^\text{ret}_{NP}(z,y) + (-1)^{\epsilon_M \epsilon_P} Q_{PM}(y) G^\text{adv}_{MQ}(y,z) \right] \frac{\delta_\text{L} G}{\delta \Phi_Q(z)} = 0 \eqend{.}
\end{equation}
Since $F$ and $G$ are arbitrary, the combination in brackets must vanish, and since $Q_{MN}$ vanishes unless $\Phi_M$ and $\Phi_N$ have opposite Grassmann parity, we get
\begin{equation}
\label{sec_gauge_compatibilitycond}
Q_{QN}(z) G^\text{ret}_{NP}(z,y) = - Q_{PM}(y) G^\text{adv}_{MQ}(y,z) \eqend{,}
\raisetag{1.2em}
\end{equation}
and using the relation~\eqref{sec_aqft_retadv_rel} between retarded and advanced Green's functions we also have
\begin{equation}
\label{sec_gauge_compatibilitycond2}
Q_{PM}(y) G^\text{ret}_{QM}(z,y) = Q_{QN}(z) G^\text{adv}_{PN}(y,z) \eqend{.}
\end{equation}
These two conditions are the free-theory Ward identities that must be fulfilled for a consistent quantum gauge theory, and they are also used in the proof of the next theorem~\ref{thm_freebrstward}. It has been shown in~\cite{froebtaslimitehrani2018} that they are fulfilled for the retarded and advanced Green's functions [and thus for the commutator~\eqref{sec_aqft_delta_def}] in general linear covariant gauges for vector and tensor fields and the corresponding ghosts. However, in order to perform the completion of $\mathfrak{A}_0$ as explained in section~\ref{sec_aqft} and to obtain the on-shell algebra, these conditions must in principle also be fulfilled for the two-point functions. This is much harder to prove; a construction of suitable two-point functions for Yang--Mills theory in Feynman gauge was given in~\cite{hollands2008,hollands_rev}.

For the corresponding anomalous Ward identities, we obtain:
\begin{theorem}
\label{thm_freebrstward}
For the free Koszul--Tate differential $\st_0^{(-1)}$, the classical part of the anomalous Ward identity is given by $\mathcal{D}_2\left( F \otimes F \right) = (F,F)$. For the remaining parts of the free BRST differential $\st_0^{(0)}$ and $\st_0^{(1)}$ we have $\mathcal{D}_2\left( F \otimes F \right) = 0$.
\end{theorem}
\begin{remark*}
The anomalous Ward identity introduced by Hollands~\cite{hollands2008} for Yang--Mills theories can be reformulated as the statement that $\mathcal{D}_2\left( F \otimes F \right) = (F,F)$ for the full (free) BRST differential $\st_0 = \st_0^{(-1)} + \st_0^{(0)} + \st_0^{(1)}$, while the anomalous Master Ward Identity of Fredenhagen and Rejzner~\cite{fredenhagenrejzner2013,rejzner2015} for general gauge theories in the BV framework is the same result for the free Koszul--Tate differential only. We see that both results are compatible, and in particular no additional terms arise for open gauge algebras.
\end{remark*}
\begin{proof}
We use the explicit formula for the second-order term given in Theorem~\ref{thm_classanom_bv}. Note that since antifields are not dynamic, and consequently have no retarded or advanced propagator, one only sums over fields in the formulas given there. In the formulas of Theorem~\ref{thm_classanom_bv}, we need the left derivative of $S_0^{(k)}$ with respect to fields, which vanishes for $S_0^{(2)}$ such that $\mathcal{D}_2\left( F \otimes F \right) = 0$ for $\st_0^{(1)}$. For the free Koszul--Tate differential $\st_0^{(-1)}$ we calculate
\begin{equation}
\left( \frac{\delta_\text{L} S_0^{(0)}}{\delta \Phi_L(y)}, F \right) = \left( P_{LM}(y) \Phi_M(y), F \right) = P_{LM}(y) \frac{\delta_\text{L} F}{\delta \Phi^\ddag_M(y)} \eqend{,}
\end{equation}
and for $\st_0^{(0)}$ we obtain
\begin{equation}
\left( \frac{\delta_\text{L} S_0^{(1)}}{\delta \Phi_L(y)}, F \right) = - \left( Q^*_{ML}(y) \Phi_M^\ddag(y), F \right) = Q^*_{ML}(y) \frac{\delta_\text{L} F}{\delta \Phi_M(y)} \eqend{.}
\end{equation}
Since $F$ is compactly supported, we can integrate by parts, and using that retarded and advanced Green's functions are solutions to the inhomogeneous equations of motion~\eqref{sec_aqft_retadv}, we obtain
\begin{equation}
\mathcal{D}_2\left( F \otimes F \right) = 2 \int \frac{\delta_\text{R} F}{\delta \Phi_K(x)} \frac{\delta_\text{L} F}{\delta \Phi^\ddag_K(x)} \total x = (F,F)
\end{equation}
for the Koszul--Tate differential $\st_0^{(-1)}$. For $\st_0^{(0)}$, we integrate by parts to obtain
\begin{splitequation}
&\mathcal{D}_2\left( F \otimes F \right) = \iint \frac{\delta_\text{R} F}{\delta \Phi_K(x)} Q_{ML}(y) \left[ G^\text{ret}_{KL}(x,y) + G^\text{adv}_{KL}(x,y) \right] \frac{\delta_\text{L} F}{\delta \Phi_M(y)} \total x \total y \\
&\quad= \frac{1}{2} \iint \frac{\delta_\text{R} F}{\delta \Phi_K(x)} \frac{\delta_\text{L} F}{\delta \Phi_M(y)} \bigg[ Q_{ML}(y) G^\text{ret}_{KL}(x,y) + Q_{ML}(y) G^\text{adv}_{KL}(x,y) \\
&\hspace{6em}+ (-1)^{\epsilon_K \epsilon_M + \epsilon_K + \epsilon_M} Q_{KL}(x) \left[ G^\text{ret}_{ML}(y,x) + G^\text{adv}_{ML}(y,x) \right] \bigg] \total x \total y \eqend{,}
\end{splitequation}
where in the second step we commuted the two functional derivatives (using that $F$ is bosonic), changed left into right derivatives according to the relations~\eqref{sec_aqft_leftrightder_props} and renamed indices and integration variables. Since $Q_{KL}$ vanishes unless $\Phi_K$ and $\Phi_L$ have opposite Grassmann parity and the retarded and advanced Green's functions vanish unless $\Phi_M$ and $\Phi_L$ have the same Grassmann parity, the terms in the second line reduce to
\begin{equation}
- Q_{KL}(x) \left[ G^\text{ret}_{ML}(y,x) + G^\text{adv}_{ML}(y,x) \right] \eqend{,}
\end{equation}
and then the terms in brackets cancel according to the compatibility conditions~\eqref{sec_gauge_compatibilitycond} and~\eqref{sec_gauge_compatibilitycond2}, such that $\mathcal{D}_2\left( F \otimes F \right) = 0$ follows for this case.
\end{proof}
\begin{remark*}
Since $\left( S_0^{(0)}, S_0^{(0)} \right) = 0$, the anomalous Ward identity for the Koszul--Tate differential can be written in the form
\begin{equation}
\label{remark_freebrstward_koszultateanomward}
\st_0^{(-1)} \mathcal{T}\left[ \exp_\otimes\left( \frac{\mathi}{\hbar} F \right) \right] = \frac{\mathi}{\hbar} \mathcal{T}\left[ \left[ \frac{1}{2} \left( S_0^{(0)} + F, S_0^{(0)} + F \right) + \mathcal{A}^{(0)}\left( \mathe_\otimes^F \right) \right] \otimes \exp_\otimes\left( \frac{\mathi}{\hbar} F \right) \right] \eqend{,}
\end{equation}
and the anomalous Ward identity for the full free BRST differential $\st_0$ is
\begin{splitequation}
\label{remark_freebrstward_brstanomward}
\st_0 \mathcal{T}\left[ \exp_\otimes\left( \frac{\mathi}{\hbar} F \right) \right] = \frac{\mathi}{\hbar} \mathcal{T}\left[ \left[ \frac{1}{2} \left( S_0 + F, S_0 + F \right) + \mathcal{A}\left( \mathe_\otimes^F \right) \right] \otimes \exp_\otimes\left( \frac{\mathi}{\hbar} F \right) \right] \eqend{,}
\end{splitequation}
with the antifield-independent part of the free action $S_0^{(0)}$ replaced by the full free action $S_0$ and the anomaly $\mathcal{A}^{(0)}$ replaced by $\mathcal{A}$. However, time-ordered products are only well-defined for functionals with compact support, and both the free action $S_0$ and the antibracket involve an integration over the whole spacetime. Therefore, $\left( S_0 + F, S_0 + F \right)$ is only a notationally convenient shorthand for $2 \brst_0 F + (F,F)$, which is compactly supported whenever $F$ is, and the same applies in the following for all similar expressions.
\end{remark*}

The anomalous terms $\mathcal{A}$ appearing in the anomalous Ward identities are highly constrained. It was first shown by Wess and Zumino~\cite{wesszumino1971} for Yang--Mills thoeory that the gauge structure of the theory is reflected in the anomaly, i.e. that potential anomalies have to satisfy a consistency condition. It was realised soon after that this condition can be reformulated using the BRST differential, and that it follows from the nilpotency of that differential~\cite{becchirouetstora1976}. Namely, we have:
\begin{theorem}
\label{thm_freebrstward_cons}
For the Koszul--Tate differential $\st_0^{(-1)}$, the anomaly satisfies the consistency condition (in the sense of generating functionals)
\begin{equation*}
\left( S_0^{(0)} + F, \mathcal{A}^{(0)}\left[ \mathe_\otimes^F \right] \right) = \frac{1}{2} \mathcal{A}^{(0)}\left[ \left( S_0^{(0)} + F, S_0^{(0)} + F \right) \otimes \mathe_\otimes^F \right] + \mathcal{A}^{(0)}\left[ \mathcal{A}^{(0)}\left[ \mathe_\otimes^F \right] \otimes \mathe_\otimes^F \right] \eqend{,}
\end{equation*}
and for the full free BRST differential $\st_0$ the anomaly satisfies the analoguous condition with the antifield-independent part of the free action $S_0^{(0)}$ replaced by the full free action $S_0$ and the anomaly $\mathcal{A}^{(0)}$ replaced by $\mathcal{A}$.
\end{theorem}
\begin{proof}
Applying $\st_0^{(-1)}$ again on the anomalous Ward identity~\eqref{remark_freebrstward_koszultateanomward} and using that $\left( \st_0^{(-1)} \right)^2 = 0$, we obtain
\begin{splitequation}
0 &= \st_0^{(-1)} \mathcal{T}\left[ \left[ \frac{1}{2} \left( S_0^{(0)} + F, S_0^{(0)} + F \right) + \mathcal{A}^{(0)}\left( \mathe_\otimes^F \right) \right] \otimes \exp_\otimes\left( \frac{\mathi}{\hbar} F \right) \right] \\
&= - \mathcal{T}\left[ \left( \left[ \frac{1}{2} \left( S_0^{(0)} + F, S_0^{(0)} + F \right) + \mathcal{A}^{(0)}\left( \mathe_\otimes^F \right) \right], S_0^{(0)} + F \right) \otimes \exp_\otimes\left( \frac{\mathi}{\hbar} F \right) \right] \\
&\quad- \mathcal{T}\left[ \mathcal{A}^{(0)}\left[ \mathe_\otimes^F \otimes \left[ \frac{1}{2} \left( S_0^{(0)} + F, S_0^{(0)} + F \right) + \mathcal{A}^{(0)}\left( \mathe_\otimes^F \right) \right] \right]\otimes \exp_\otimes\left( \frac{\mathi}{\hbar} F \right) \right] \eqend{,}
\end{splitequation}
using the anomalous Ward identity~\eqref{remark_freebrstward_koszultateanomward} again and using that $\mathcal{T}\left[ A \otimes A \otimes \cdots \right] = 0$ for fermionic $A$ (such as $\mathcal{A}^{(0)}\left[ \exp_\otimes(F) \right]$) because of the graded symmetry of the time-ordered products. Since the antibracket satisfies the graded Jacobi identity~\eqref{sec_gauge_antibracket_jacobi}, we have
\begin{equation}
\left( \left( S_0^{(0)} + F, S_0^{(0)} + F \right), S_0^{(0)} + F \right) = 0 \eqend{,}
\end{equation}
and by the graded symmetry~\eqref{sec_gauge_antibracket_gradsym} of the antibracket also
\begin{equation}
\left( \mathcal{A}^{(0)}\left[ \mathe_\otimes^F \right], S_0^{(0)} + F \right) = - \left( S_0^{(0)} + F, \mathcal{A}^{(0)}\left[ \mathe_\otimes^F \right] \right) \eqend{.}
\end{equation}
Since a time-ordered product vanishes only if its argument vanishes, the consistency condition follows. For the full free BRST differential, the same proof applies with $S_0^{(0)}$ replaced by $S_0$ and $\mathcal{A}^{(0)}$ replaced by $\mathcal{A}$.
\end{proof}

Furthermore, since antifields can be taken out of time-ordered products~\eqref{sec_gauge_antifields_timeordered}, there should be no anomaly associated to them. While this was asserted in~\cite{demedeiroshollands2013}, no proof was given. Closing this gap, we have the following:
\begin{theorem}
\label{thm_anom_antifields}
Given a derivation $D$ of the form given in Theorem~\ref{thm_classanom_bv}, if perturbative agreement holds for all fields (if any) appearing in $\delta_\text{R} Q/\delta \Phi_K(x)$ and if there exists $k$ such that $F_k \in \mathcal{F}$ only contains antifields, i.e. if $\delta F_k/\delta \Phi_K(x) = 0$ for all $K$, the anomaly vanishes: $\mathcal{A}_n\left[ F_1 \otimes \cdots \otimes F_n \right] = 0$.
\end{theorem}
\begin{remark*}
Perturbative agreement~\cite{hollandswald2005,zahn2015,dragohackpinamonti2017} is the statement that terms that are at most quadratic in the fields can be shifted freely between the free action and the interaction part. While this is usually assumed to hold, it actually results in non-trivial conditions on the time-ordered products, and one needs to show that these conditions can be fulfilled by a suitable field redefinition~\eqref{sec_aqft_renormfreedom}. For the above theorem, one needs perturbative agreement for terms linear in the fields, which results in the conditions~\cite{hollandswald2005}
\begin{equations}[thm_anom_antifields_pa]
\begin{split}
\mathcal{T}_{n+1}\left[ F^{\otimes n} \otimes \Phi_K(x) \right] &= \mathcal{T}_n\left( F^{\otimes n} \right) \star_\hbar \Phi_K(x) \\
&\quad+ \mathi \hbar \, n \, \int G^\text{adv}_{LK}(y,x) \mathcal{T}_n\left[ F^{\otimes (n-1)} \otimes \frac{\delta_\text{R} F}{\delta \Phi_L(y)} \right] \total y \eqend{,}
\label{thm_anom_antifields_pa_1}
\end{split} \\
\begin{split}
\mathcal{T}_{n+1}\left[ \Phi_K(x) \otimes F^{\otimes n} \right] &= \Phi_K(x) \star_\hbar \mathcal{T}_n\left( F^{\otimes n} \right) \\
&\quad+ \mathi \hbar \, n \, \int G^\text{ret}_{LK}(y,x) \mathcal{T}_n\left[ F^{\otimes (n-1)} \otimes \frac{\delta_\text{R} F}{\delta \Phi_L(y)} \right] \total y \eqend{.}
\label{thm_anom_antifields_pa_2}
\end{split}
\end{equations}
Using the commutation relation for time-ordered products~\eqref{sec_aqft_timeord_commutator}, either condition follows from the other.
\end{remark*}
\begin{proof}
Since the general case can be recovered by polarisation using linearity and the graded symmetry of the anomaly, we can restrict to the case where $k = n$ and all the $F_i$ with $1 \leq i < n$ are equal, setting $F_n = G$ and $F_i = F$ for $1 \leq i < n$, and furthermore assume that $F$ and $G$ are bosonic. We furthermore set $D_\text{cl} F \equiv \left( Q, F \right) = \mathcal{D}_1\left( F \right)$ as stated in Theorem~\ref{thm_classanom_bv}. By polarisation (replacing $F \to F + \alpha G$, taking a derivative with respect to $\alpha$ and setting $\alpha = 0$), the anomalous Ward identity~\eqref{proof_anomward_identity} gives
\begin{splitequation}
\label{proof_anom_antifields_wardidentity}
&D \mathcal{T}_n\left( F^{\otimes (n-1)} \otimes G \right) = \left( \frac{\hbar}{\mathi} \right)^{n-1} \mathcal{T}_1\left[ \hat{\mathcal{D}}_n\left( F^{\otimes (n-1)} \otimes G \right) \right] \\
&\qquad+ \sum_{k=1}^{n-1} \frac{(n-1)!}{(k-1)! (n-k)!} \left( \frac{\hbar}{\mathi} \right)^{k-1} \mathcal{T}_{n-k+1}\left[ \hat{\mathcal{D}}_k\left( F^{\otimes (k-1)} \otimes G \right) \otimes F^{\otimes (n-k)} \right] \\
&\qquad+ \sum_{k=1}^{n-1} \frac{(n-1)!}{k! (n-k-1)!} \left( \frac{\hbar}{\mathi} \right)^{k-1} \mathcal{T}_{n-k+1}\left[ \hat{\mathcal{D}}_k\left( F^{\otimes k} \right) \otimes F^{\otimes (n-k-1)} \otimes G \right] \eqend{,}
\end{splitequation}
where $\hat{\mathcal{D}}_k = \mathcal{D}_k + \mathcal{A}_k$ and we have isolated the terms with $k = n$ in the sum.

For $n = 1$, this reads
\begin{equation}
\label{proof_anom_antifields_ward1}
D \mathcal{T}_1\left( G \right) = \mathcal{T}_1\left[ \mathcal{D}_1\left( G \right) \right] + \mathcal{T}_1\left[ \mathcal{A}_1\left( G \right) \right] \eqend{.}
\end{equation}
Since $G$ does not depend on fields, we have $\mathcal{T}_1\left( G \right) = G$. Because $\mathcal{D}_1\left( G \right) = D_\text{cl} G$, which can be at most linear in fields (and of arbitrary order in antifields) since $Q$ is at most quadratic in fields and antifields by assumption, and no Hadamard parametrix for antifields exists, it follows that
\begin{equation}
D \mathcal{T}_1\left( G \right) = D_\text{cl} G = \normord{ D_\text{cl} G }_H = \mathcal{T}_1\left( D_\text{cl} G \right) = \mathcal{T}_1\left[ \mathcal{D}_1\left( G \right) \right] \eqend{,}
\end{equation}
and since a time-ordered product vanishes only if its argument vanishes we obtain $\mathcal{A}_1\left( G \right) = 0$. For $n = 2$, the anomalous Ward identity~\eqref{proof_anom_antifields_wardidentity} gives
\begin{splitequation}
\mathcal{T}_1\left[ \mathcal{A}_2\left( F \otimes G \right) \right] &= \frac{\mathi}{\hbar} D \mathcal{T}_2\left( F \otimes G \right) - \frac{\mathi}{\hbar} \mathcal{T}_2\left( F \otimes D_\text{cl} G \right) - \frac{\mathi}{\hbar} \mathcal{T}_2\left( D_\text{cl} F \otimes G \right) \\
&\quad- \frac{\mathi}{\hbar} \mathcal{T}_2\left[ \mathcal{A}_1\left( F \right) \otimes G \right] - \mathcal{T}_1\left[ \mathcal{D}_2\left( F \otimes G \right) \right] \eqend{.}
\end{splitequation}
We then use that antifields can be taken out of time-ordered products~\eqref{sec_gauge_antifields_timeordered}, that $D$ is a derivation and the anomalous Ward identity~\eqref{proof_anom_antifields_ward1} to obtain
\begin{splitequation}
D \mathcal{T}_2\left( F \otimes G \right) &= D \mathcal{T}_1\left( F \right) \star_\hbar G + \mathcal{T}_1\left( F \right) \star_\hbar D G \\
&= \mathcal{T}_1\left( D_\text{cl} F \right) \star_\hbar G + \mathcal{T}_1\left[ \mathcal{A}_1\left( F \right) \right] \star_\hbar G + \mathcal{T}_1\left( F \right) \star_\hbar D G \\
&= \mathcal{T}_2\left( D_\text{cl} F \otimes G \right) + \mathcal{T}_2\left[ \mathcal{A}_1\left( F \right) \otimes G \right] + \mathcal{T}_1\left( F \right) \star_\hbar \mathcal{T}_1\left( D_\text{cl} G \right) \eqend{,}
\end{splitequation}
and therefore
\begin{equation}
\mathcal{T}_1\left[ \mathcal{A}_2\left( F \otimes G \right) \right] = \frac{\mathi}{\hbar} \mathcal{T}_1\left( F \right) \star_\hbar \mathcal{T}_1\left( D_\text{cl} G \right) - \frac{\mathi}{\hbar} \mathcal{T}_2\left( F \otimes D_\text{cl} G \right) - \mathcal{T}_1\left[ \mathcal{D}_2\left( F \otimes G \right) \right] \eqend{.}
\end{equation}
In the same way, by taking $G$ out on the left we obtain the same formula with $\mathcal{T}_1\left( D_\text{cl} G \right) \star_\hbar \mathcal{T}_1\left( F \right)$ instead of $\mathcal{T}_1\left( F \right) \star_\hbar \mathcal{T}_1\left( D_\text{cl} G \right)$, and symmetrising we get
\begin{splitequation}
\mathcal{T}_1\left[ \mathcal{A}_2\left( F \otimes G \right) \right] &= \frac{\mathi}{2 \hbar} \mathcal{T}_1\left( F \right) \star_\hbar \mathcal{T}_1\left( D_\text{cl} G \right) + \frac{\mathi}{2 \hbar} \mathcal{T}_1\left( D_\text{cl} G \right) \star_\hbar \mathcal{T}_1\left( F \right) \\
&\quad- \frac{\mathi}{\hbar} \mathcal{T}_2\left( F \otimes D_\text{cl} G \right) - \mathcal{T}_1\left[ \mathcal{D}_2\left( F \otimes G \right) \right] \eqend{.}
\end{splitequation}
From the explicit formula of Theorem~\ref{thm_classanom_bv}, we obtain by polarisation
\begin{splitequation}
\label{proof_anom_antifields_d2}
2 \mathcal{D}_2\left( F \otimes G \right) &= D_\text{cl} R^\text{cl}_1\left( F; G \right) + D_\text{cl} R^\text{cl}_1\left( G; F \right) - R^\text{cl}_1\left( F; D_\text{cl} G \right) \\
&\quad- R^\text{cl}_1\left( D_\text{cl} F; G \right) - R^\text{cl}_1\left( G; D_\text{cl} F \right) - R^\text{cl}_1\left( D_\text{cl} G; F \right) \eqend{.}
\end{splitequation}
Since the antifields are not dynamical, in the formula~\eqref{sec_aqft_rcl_firstorder} for the classical retarded product one only sums over fields, and it follows that $R^\text{cl}_1\left( H; G \right) = R^\text{cl}_1\left( G; H \right) = 0$ for any $H \in \mathcal{F}$. Therefore,
\begin{equation}
2 \mathcal{D}_2\left( F \otimes G \right) = - R^\text{cl}_1\left( F; D_\text{cl} G \right)  - R^\text{cl}_1\left( D_\text{cl} G; F \right) \eqend{,}
\end{equation}
and we obtain
\begin{splitequation}
\mathcal{T}_1\left[ \mathcal{A}_2\left( F \otimes G \right) \right] &= \frac{\mathi}{2 \hbar} \mathcal{T}_1\left( F \right) \star_\hbar \mathcal{T}_1\left( D_\text{cl} G \right) + \frac{\mathi}{2 \hbar} \mathcal{T}_1\left( D_\text{cl} G \right) \star_\hbar \mathcal{T}_1\left( F \right) \\
&\quad- \frac{\mathi}{\hbar} \mathcal{T}_2\left( F \otimes D_\text{cl} G \right) + \frac{1}{2} \mathcal{T}_1\left[ R^\text{cl}_1\left( F; D_\text{cl} G \right) \right] + \frac{1}{2} \mathcal{T}_1\left[ R^\text{cl}_1\left( D_\text{cl} G; F \right) \right] \eqend{.}
\end{splitequation}
If $D_\text{cl} G$ does not contain any fields, we have $\mathcal{T}_2\left( F \otimes D_\text{cl} G \right) = \mathcal{T}_1\left( F \right) \star_\hbar \mathcal{T}_1\left( D_\text{cl} G \right) = \mathcal{T}_1\left( D_\text{cl} G \right) \star_\hbar \mathcal{T}_1\left( F \right)$ and $R^\text{cl}_1\left( F; D_\text{cl} G \right) = R^\text{cl}_1\left( D_\text{cl} G; F \right) = 0$, and the anomaly $\mathcal{A}_2\left( F \otimes G \right)$ vanishes. We thus concentrate on the terms of $D_\text{cl} G$ that are exactly linear in fields, writing $D_\text{cl} G = \hat{D}_\text{cl} G + {}$terms only depending on antifields with
\begin{equation}
\hat{D}_\text{cl} G = \iint \Phi_L(y) \frac{\delta_\text{L} \delta_\text{R} Q}{\delta \Phi_L(y) \delta \Phi_K(x)} \frac{\delta_\text{L} G}{\delta \Phi^\ddag_K(x)} \total x \total y \eqend{.}
\end{equation}
For these terms perturbative agreement~\eqref{thm_anom_antifields_pa} holds by assumption, which gives
\begin{splitequation}
&\mathcal{T}_2\left( F \otimes \hat{D}_\text{cl} G \right) = \iint \mathcal{T}_2\left[ F \otimes \Phi_L(y) \right] \frac{\delta_\text{L} \delta_\text{R} Q}{\delta \Phi_L(y) \delta \Phi_K(x)} \frac{\delta_\text{L} G}{\delta \Phi^\ddag_K(x)} \total x \total y \\
&\quad= \iint \mathcal{T}_1\left( F \right) \star_\hbar \Phi_L(y) \frac{\delta_\text{L} \delta_\text{R} Q}{\delta \Phi_L(y) \delta \Phi_K(x)} \frac{\delta_\text{L} G}{\delta \Phi^\ddag_K(x)} \total x \total y \\
&\qquad\quad+ \mathi \hbar \iiint G^\text{adv}_{ML}(z,y) \mathcal{T}_1\left( \frac{\delta_\text{R} F}{\delta \Phi_M(z)} \right) \frac{\delta_\text{L} \delta_\text{R} Q}{\delta \Phi_L(y) \delta \Phi_K(x)} \frac{\delta_\text{L} G}{\delta \Phi^\ddag_K(x)} \total x \total y \total z \\
&\quad= \mathcal{T}_1\left( F \right) \star_\hbar \hat{D}_\text{cl} G + \mathi \hbar \mathcal{T}_1\left[ \iint \frac{\delta_\text{R} F}{\delta \Phi_M(z)} G^\text{adv}_{ML}(z,y) \frac{\delta_\text{L} \left( \hat{D}_\text{cl} G \right)}{\delta \Phi_L(y)} \total y \total z \right] \\
&\quad= \mathcal{T}_1\left( F \right) \star_\hbar \mathcal{T}_1\left( \hat{D}_\text{cl} G \right) - \mathi \hbar \mathcal{T}_1\left[ R^\text{cl}_1\left( F; \hat{D}_\text{cl} G \right) \right] \eqend{,}
\raisetag{1.3em}
\end{splitequation}
where we used the explicit formula~\eqref{sec_aqft_rcl_firstorder} for the classical retarded product to obtain the last equality. Similarly, perturbative agreement for $\mathcal{T}_2\left[ \Phi_L(y) \otimes F \right]$ [taking out $\Phi_L(y)$ to the left~\eqref{thm_anom_antifields_pa}] leads to
\begin{equation}
\mathcal{T}_2\left( F \otimes \hat{D}_\text{cl} G \right) = \mathcal{T}_1\left( \hat{D}_\text{cl} G \right) \star_\hbar \mathcal{T}_1\left( F \right) - \mathi \hbar \mathcal{T}_1\left[ R^\text{cl}_1\left( \hat{D}_\text{cl} G; F \right) \right] \eqend{,}
\end{equation}
and by symmetrising we obtain $\mathcal{T}_1\left[ \mathcal{A}_2\left( F \otimes G \right) \right] = 0$ and thus $\mathcal{A}_2\left( F \otimes G \right) = 0$ also in this case.

In the general case for $n > 2$ we proceed in the same way, isolating the term containing $\mathcal{A}_n$ in the anomalous Ward identity~\eqref{proof_anom_antifields_wardidentity}, taking $G$ out of the time-ordered products because it only contains antifields, using that $D$ is a derivation and then using the anomalous Ward identity~\eqref{proof_anomward_identity} for the terms on the right-hand side that do not contain $G$. Since $\mathcal{D}_k = 0$ for $k > 2$, and we have already proven that $\mathcal{A}_1\left( G \right) = 0 = \mathcal{A}_2\left( F \otimes G \right)$, this results in
\begin{splitequation}
&\mathcal{T}_1\left[ \mathcal{A}_n\left( F^{\otimes (n-1)} \otimes G \right) \right] = \left( \frac{\mathi}{\hbar} \right)^{n-1} N_n\left( F^{\otimes (n-1)} \otimes G \right) \\
&\qquad- \sum_{k=3}^{n-1} \frac{(n-1)!}{(k-1)! (n-k)!} \left( \frac{\hbar}{\mathi} \right)^{k-n} \mathcal{T}_{n-k+1}\left[ \mathcal{A}_k\left( F^{\otimes (k-1)} \otimes G \right) \otimes F^{\otimes (n-k)} \right] \eqend{,}
\end{splitequation}
where the last sum only appears for $n > 3$ and where we set
\begin{splitequation}
N_n\left( F^{\otimes (n-1)} \otimes G \right) &= \frac{1}{2} \mathcal{T}_1\left( D_\text{cl} G \right) \star_\hbar \mathcal{T}_{n-1}\left( F^{\otimes (n-1)} \right) \\
&\quad+ \frac{1}{2} \mathcal{T}_{n-1}\left( F^{\otimes (n-1)} \right) \star_\hbar \mathcal{T}_1\left( D_\text{cl} G \right) - \mathcal{T}_n\left[ D_\text{cl} G \otimes F^{\otimes (n-1)} \right] \\
&\quad- \frac{\hbar}{\mathi} (n-1) \mathcal{T}_{n-1}\left[ \mathcal{D}_2\left( F \otimes G \right) \otimes F^{\otimes (n-2)} \right] \eqend{.}
\raisetag{1.6em}
\end{splitequation}
It follows by induction in $n$ that $\mathcal{A}_n\left( F^{\otimes (n-1)} \otimes G \right) = 0$ if we can show that $N_n\left( F^{\otimes (n-1)} \otimes G \right) = 0$ for all $n$. However, this is seen in the same way as for $n = 2$: by the explicit formula~\eqref{proof_anom_antifields_d2} for $\mathcal{D}_2\left( F \otimes G \right)$, if $D_\text{cl} G$ does not contain any fields, $\mathcal{D}_2\left( F \otimes G \right)$ vanishes and $D_\text{cl} G$ can be taken out of the time-ordered product such that $N_n\left( F^{\otimes (n-1)} \otimes G \right) = 0$, while if $D_\text{cl} G$ is linear in fields, we have to use perturbative agreement~\eqref{thm_anom_antifields_pa} and arrive at the same conclusion.
\end{proof}

However, perturbative agreement for terms linear in the fields can be always fulfilled by a field redefinition. This has been shown explicitly for scalar and spinor fields~\cite{hollandswald2005,demedeiroshollands2013,zahn2015,dragohackpinamonti2017}, but it also follows in the general case with antifields:
\begin{theorem}
\label{thm_pertagreement}
One can choose renormalisation conditions such that~\eqref{thm_anom_antifields_pa} holds.
\end{theorem}
\begin{proof}
We consider the difference between left- and right-hand side of the condition~\eqref{thm_anom_antifields_pa_1}:
\begin{splitequation}
\label{proof_pertagreement_n_def}
\mathcal{N}_n\left[ \Phi_K(x); F^{\otimes n} \right] &\equiv \mathcal{T}_{n+1}\left[ \Phi_K(x) \otimes F^{\otimes n} \right] - \mathcal{T}_n\left( F^{\otimes n} \right) \star_\hbar \Phi_K(x) \\
&\quad- \mathi \hbar \, n \, \int G^\text{adv}_{LK}(y,x) \mathcal{T}_n\left[ F^{\otimes (n-1)} \otimes \frac{\delta_\text{R} F}{\delta \Phi_L(y)} \right] \total y \eqend{.}
\end{splitequation}
By the multilinearity of the time-ordered products, $\mathcal{N}_n$ is also multilinear, and by the locality and covariance property it is also local and covariant. It inherits further the neutral element and graded symmetry properties, and we can thus recover the general case (for different $F_i$) by polarisation. Moreover, since antifields can be taken out of time-ordered products (including individual $F_i$) we can assume without loss of generality that $F$ only contains fields.


To show that we can make $\mathcal{N}_n$ vanish by a field redefinition, we have to proceed by induction in the total number of fields $N_\Phi$ contained in $F^{\otimes n}$~\cite{hollandswald2005,zahn2015}, and for fixed $N_\Phi$ ascend in $n$. The induction starts with $N_\Phi = 0$ and $n = 0$, where
\begin{equation}
\mathcal{N}_0\left[ \Phi_K(x); - \right] = \mathcal{T}_1\left[ \Phi_K(x) \right] - \Phi_K(x) = 0 \eqend{.}
\end{equation}
Assume thus that $\mathcal{N}$ vanishes for all $N_\Phi' < N_\Phi$, and all $n' < n$. Using the field independence of the time-ordered products~\eqref{sec_aqft_timeord_fieldindep}, we calculate
\begin{equation}
\label{proof_pertagreement_n_fieldder}
\frac{\delta_\text{L}}{\delta \Phi_L(y)} \mathcal{N}_n\left[ \Phi_K(x); F^{\otimes n} \right] = n \mathcal{N}_n\left[ \Phi_K(x); F^{\otimes (n-1)} \otimes \frac{\delta_\text{L} F}{\delta \Phi_L(y)} \right]
\end{equation}
where we used the commutation relations for left and right derivatives~\eqref{sec_aqft_leftrightder_props}, and that by polarisation
\begin{splitequation}
&\mathcal{N}_n\left[ \Phi_K(x); F^{\otimes (n-1)} \otimes G \right] = \mathcal{T}_{n+1}\left[ F^{\otimes (n-1)} \otimes G \otimes \Phi_K(x) \right] \\
&\quad- \mathcal{T}_n\left[ F^{\otimes (n-1)} \otimes G \right] \star_\hbar \Phi_K(x) - \mathi \hbar \int G^\text{adv}_{MK}(z,x) \mathcal{T}_n\left[ F^{\otimes (n-1)} \otimes \frac{\delta_\text{R} G}{\delta \Phi_M(z)} \right] \total z \\
&\quad- \mathi \hbar (n-1) \int G^\text{adv}_{MK}(z,x) \mathcal{T}_n\left[ F^{\otimes (n-2)} \otimes G \otimes \frac{\delta_\text{R} F}{\delta \Phi_M(z)} \right] \total z
\raisetag{1.9em}
\end{splitequation}
for $G$ of arbitrary Grassmann parity. Since the right-hand side of~\eqref{proof_pertagreement_n_fieldder} contains a smaller total number of fields $N_\Phi$ (and possibly a smaller $n$ by the neutral element property if $F$ is linear in fields), it vanishes by assumption, and it follows that $\mathcal{N}_n\left[ \Phi_K(x); F^{\otimes n} \right]$ is proportional to the identity operator $\unitmatrix$. Assume now that $x \not\in \supp F$; then either $J^+(\{ x \}) \cap J^-(\supp F) = \emptyset$ or $J^-(\{ x \}) \cap J^+(\supp F) = \emptyset$ (or both), and $\mathcal{N}_n\left[ \Phi_K(x); F^{\otimes n} \right]$ vanishes by the factorisation property of the time-ordered products [for $J^+(\{ x \}) \cap J^-(\supp F) = \emptyset$ one additionally needs to use the commutation relation for time-ordered products~\eqref{sec_aqft_timeord_commutator}]. In the general case for $\mathcal{N}_n\left[ \Phi_K(x); F_1 \otimes \cdots \otimes F_n \right]$ with distinct $F_i$, one shows in complete analogy to the proof of Theorem~\ref{thm_anomward} that $\mathcal{N}_n = 0$ if there exist $i$, $j$ with $\supp F_i \cap \supp F_j = \emptyset$, or $i$ with $\{ x \} \cap \supp F_i = \emptyset$, i.e., it is supported on the total diagonal.

Consider now a field redefinition according to~\eqref{sec_aqft_renormfreedom} with $\mathcal{Z}_k = 0$ for all $k \leq n$. All time-ordered products with $k \leq n$ entries are thus unchanged, and for the time-ordered products with $n+1$ entries we obtain by polarisation
\begin{equation}
\hat{\mathcal{T}}_{n+1}\left( F^{\otimes n} \otimes G \right) = \mathcal{T}_{n+1}\left( F^{\otimes n} \otimes G \right) + \left( \frac{\hbar}{\mathi} \right)^n \mathcal{T}_1\left[ \mathcal{Z}_{n+1}\left( F^{\otimes n} \otimes G \right) \right] \eqend{,}
\end{equation}
such that the new $\hat{\mathcal{N}}$ after the redefinition is given by
\begin{equation}
\hat{\mathcal{N}}_n\left[ \Phi_K(x); F^{\otimes n} \right] = \mathcal{N}_n\left[ \Phi_K(x); F^{\otimes n} \right] + \left( \frac{\hbar}{\mathi} \right)^n \mathcal{T}_1\left[ \mathcal{Z}_{n+1}\left[ F^{\otimes n} \otimes \Phi_K(x) \right] \right] \eqend{.}
\end{equation}
Since $\mathcal{N}_n$ is proportional to the identity, we can identify it with a classical functional and make the right-hand side vanish by setting
\begin{equation}
\label{proof_pertagreement_z_from_n}
\mathcal{Z}_{n+1}\left[ F^{\otimes n} \otimes \Phi_K(x) \right] = - \left( \frac{\mathi}{\hbar} \right)^n \mathcal{N}_n\left[ \Phi_K(x); F^{\otimes n} \right] \eqend{.}
\end{equation}
In order for this definition to be a valid field redefinition, we thus need to further show that $\mathcal{N}_n$ is (graded) symmetric if one of the $F$ is a basic field, that it has the right order in $\hbar$, and that we can fulfil the field-independence property of $\mathcal{Z}_{n+1}$ (properties of the field redefinition which we have not stated explicitly, such as a wave front set condition and almost homogeneous scaling, are automatically fulfilled from the definition of $\mathcal{N}_n$~\cite{hollandswald2005,zahn2015}).

For the graded symmetry, we calculate using the commutation relation for basic fields~\eqref{sec_aqft_commutator}, the relation between retarded Green's function and commutator~\eqref{sec_aqft_delta_def} and the properties of left and right derivatives~\eqref{sec_aqft_leftrightder_props} that
\begin{splitequation}
&\mathcal{N}_n\left[ \Phi_K(x); \Phi_L(y) \otimes F^{\otimes (n-1)} \right] - (-1)^{\epsilon_K \epsilon_L} \mathcal{N}_n\left[ \Phi_L(y); \Phi_K(x) \otimes F^{\otimes (n-1)} \right] = \\
&\quad+ (-1)^{\epsilon_K \epsilon_L} \mathcal{N}_{n-1}\left[ \Phi_K(x); F^{\otimes (n-1)} \right] \star_\hbar \Phi_L(y) \\
&\quad- \mathcal{N}_{n-1}\left[ \Phi_L(y); F^{\otimes (n-1)} \right] \star_\hbar \Phi_K(x) \\
&\quad+ \mathi \hbar \, (n-1) \, \int G^\text{adv}_{ML}(z,y) \mathcal{N}_{n-1}\left[ \Phi_K(x); F^{\otimes (n-2)} \otimes \frac{\delta_\text{R} F}{\delta \Phi_M(z)} \right] \total z \\
&\quad- \mathi \hbar \, (n-1) \, \int (-1)^{\epsilon_K \epsilon_L} G^\text{adv}_{MK}(z,x) \mathcal{N}_{n-1}\left[ \Phi_L(y); F^{\otimes (n-2)} \otimes \frac{\delta_\text{R} F}{\delta \Phi_M(z)} \right] \total z \eqend{.}
\end{splitequation}
Since by induction all $\mathcal{N}$ with a smaller total number of fields $N_\Phi' < N_\Phi$ vanish, the right-hand side is zero and the graded symmetry property is fulfilled. Field independence is similarly easy to show: since $\delta_\text{L}/\delta \Phi_M(z) \mathcal{N}_n\left[ \Phi_K(x); F^{\otimes n} \right] = 0$, we have to set all $\mathcal{Z}$ with a smaller total number of fields to zero, which is consistent with the induction and the assumption that the $\mathcal{Z}_k$ with $k \leq n$ vanish (because of the neutral element property). We also set all $\mathcal{Z}_k$ with the same total number of fields $N_\Phi$, but all entries at least quadratic in fields, to zero. In the other direction, we need to define $\mathcal{Z}$ with a higher number of fields, which is straightforward to do by ``integrating up'' the lower-order expressions: first one defines the $\mathcal{Z}_k\left[ F_1 \otimes \cdots \otimes F_k \right]$ with $k \geq n+1$, $N_\Phi+1$ total fields and all $F_i$ at least quadratic in fields. Acting with a functional derivative on these, one obtains $\mathcal{Z}_k$ with $k \geq n$ and $N_\Phi$ total fields, and either all entries at least quadratic in fields (in which case they vanish by definition), or one entry linear in fields, in which case they are defined by~\eqref{proof_pertagreement_z_from_n}. These can thus be defined such that field independence holds, where possible ``integration constants'' can be chosen arbitrarily. One can then define the $\mathcal{Z}_k\left[ F_1 \otimes \cdots \otimes F_k \right]$ with $k \geq n+1$, $N_\Phi+1$ total fields and at most one $F_i$ linear in fields; acting with a functional derivative one obtains a non-vanishing term only when the entry linear in fields is preserved, and we can again ``integrate up'' this term. The procedure is continued ascending in the number of entries that are linear in fields, and stops with $\mathcal{Z}_{N_\Phi+1}\left[ F_1 \otimes \cdots \otimes F_{N_\Phi+1} \right]$ where all $F_i$ are linear in fields. One then continues to define the $\mathcal{Z}_k$ with $N_\Phi+2$ total fields, etc.

It remains to show that $\mathcal{N}_n$ is at least of order $\hbar^{n+1}$ to make~\eqref{proof_pertagreement_z_from_n} a valid field redefinition. This is more involved, and does not seem to have been given much consideration before. Multiplying the definition of $\mathcal{N}_n$~\eqref{proof_pertagreement_n_def} by $1/n! \, \left( \mathi/\hbar \right)^n$ and summing over $n$, we obtain the generating functional
\begin{splitequation}
\mathcal{N}\left[ \Phi_K(x); \exp_\otimes\left( \frac{\mathi}{\hbar} F \right) \right] &= \mathcal{T}\left[ \Phi_K(x) \otimes \exp_\otimes\left( \frac{\mathi}{\hbar} F \right) \right] - \mathcal{T}\left[ \exp_\otimes\left( \frac{\mathi}{\hbar} F \right) \right] \star_\hbar \Phi_K(x) \\
&\quad+ \int G^\text{adv}_{LK}(y,x) \mathcal{T}\left[ \exp_\otimes\left( \frac{\mathi}{\hbar} F \right) \otimes \frac{\delta_\text{R} F}{\delta \Phi_L(y)} \right] \total y \eqend{.}
\end{splitequation}
Multiplying from the left by $\mathcal{T}\left[ \exp_\otimes\left( \frac{\mathi}{\hbar} F \right) \right]^{\star_\hbar (-1)}$, we can express the right-hand side using retarded products~\eqref{sec_aqft_retprod_def} to get
\begin{splitequation}
&\mathcal{T}\left[ \exp_\otimes\left( \frac{\mathi}{\hbar} F \right) \right]^{\star_\hbar (-1)} \star_\hbar \mathcal{N}\left[ \Phi_K(x); \exp_\otimes\left( \frac{\mathi}{\hbar} F \right) \right] \\
&\quad= \mathcal{R}\left[ \mathe_\otimes^F; \Phi_K(x) \right] - \Phi_K(x) + \int G^\text{adv}_{LK}(y,x) \mathcal{R}\left[ \mathe_\otimes^F; \frac{\delta_\text{R} F}{\delta \Phi_L(y)} \right] \total y \eqend{.}
\end{splitequation}
We then expand both sides in powers of $F$; since all $\mathcal{N}_k$ with $k < n$ vanish by induction, we simply obtain (for $n \geq 1$)
\begin{splitequation}
\mathcal{N}_n\left[ \Phi_K(x); F^{\otimes n} \right] = \left( \frac{\hbar}{\mathi} \right)^n \bigg[ &\mathcal{R}_n\left[ F^{\otimes n}; \Phi_K(x) \right] \\
&+ n \int G^\text{adv}_{LK}(y,x) \mathcal{R}_{n-1}\left[ F^{\otimes (n-1)}; \frac{\delta_\text{R} F}{\delta \Phi_L(y)} \right] \total y \bigg] \eqend{.}
\end{splitequation}
Since the retarded products have a well-defined classical limit as $\hbar \to 0$, it follows that $\mathcal{N}_n = \bigo{\hbar^{n+1}}$ if the classical limit of the retarded products inside the brackets vanishes. This, however, follows straightforwardly: from the relation~\eqref{sec_aqft_rcl_glinear_rel} and explicit formula for the first-order classical retarded products~\eqref{sec_aqft_rcl_firstorder} we obtain
\begin{splitequation}
R^\text{cl}_n\left[ F^{\otimes n}; \Phi_K(x) \right] &= n R^\text{cl}_{n-1}\left[ F^{\otimes (n-1)}; R^\text{cl}_1\left[ F; \Phi_K(x) \right] \right] \\
&= - n \int G^\text{adv}_{LK}(y,x) R^\text{cl}_{n-1}\left[ F^{\otimes (n-1)}; \frac{\delta_\text{R} F}{\delta \phi_L(y)} \right] \total y \eqend{.}
\end{splitequation}
\end{proof}

One would now like to extend the above results to the interacting theory, and in particular obtain the derivation $\st$ corresponding to the classical interacting BRST differential $\brst$. However, this is not possible in general because of the presence of anomalous terms in the anomalous Ward identity~\eqref{remark_freebrstward_brstanomward}, in particular the anomaly $\mathcal{A}\left[ \exp_\otimes(L) \right]$ for the interaction $L$ (with cutoff function $g(x)$) itself. Since the anomalous Ward identity, and therefore the anomalous terms depend on the renormalisation conditions that one has imposed in the construction of the time-ordered products, it is possible in certain cases to remove the anomaly order by order in perturbation theory by a field redefinition~\eqref{sec_aqft_renormfreedom}. This is in particular the case if a certain cohomological condition holds. The explicit field redefinitions that one needs to perform to remove the anomaly in this case, and for which the consistency conditions of Theorem~\ref{thm_freebrstward_cons} are essential, are detailed in~\cite{hollands2008}. However, this condition is not fulfilled if the theory contains chiral fermions (the well-known axial anomaly~\cite{adler1969,belljackiw1969,fujikawa1979}), and it also can happen that while an anomalous term can exist in principle, its coefficient vanishes because of a specific choice of matter representation (as in the Standard Model~\cite{gengmarshak1989,minahanramondwarner1990}). In the following, we thus only assume that $\mathcal{A}\left[ \exp_\otimes(L) \right]$ vanishes, without worrying about the underlying reason.
\begin{theorem}
\label{thm_brstward}
If the anomaly of the interaction $L$ with cutoff function $g$ vanishes, $\mathcal{A}\left[ \exp_\otimes(L) \right] = 0$, one can deform the free BRST differential $\st_0$ into an interacting nilpotent BRST differential $\st = \st_0 + \bigo{g}$ as a graded derivation on the free-field algebra $\overline{\mathfrak{A}}_0$. On perturbatively interacting fields, we have the interacting anomalous Ward identity
\begin{equation}
\label{thm_brstward_anomward}
\st \mathcal{T}_L\left[ \exp_\otimes\left( \frac{\mathi}{\hbar} F \right) \right] = \frac{\mathi}{\hbar} \mathcal{T}_L\left[ \left( \brst F + \frac{1}{2} (F,F) + \mathcal{A}\left( \mathe_\otimes^{L+F} \right) \right) \otimes \exp_\otimes\left( \frac{\mathi}{\hbar} F \right) \right]
\end{equation}
as an equality between generating functionals, where $g \rvert_{\supp F} = 1$ and for simplicity we assume $F$ to be bosonic. On-shell and whenever the interacting BRST Noether charge $\mathcal{T}_L\left( Q \right)$ is well-defined, we have $\st F = (\mathi \hbar)^{-1} \left[ \mathcal{T}_L\left( Q \right), F \right]_{\star_\hbar}$ for all $F \in \overline{\mathfrak{A}}_0$.
\end{theorem}
\begin{proof}
We determine explicitly the generator of the difference $\st - \st_0$. For this, observe that for local functionals $F$ and $G$ we have~\cite{hollands2008,taslimitehrani2017}
\begin{equation}
\label{proof_brstward_antibracket}
(F,G) = \int f(x) \left[ \left( F, \Phi^\ddag_K(x) \right) \left( \Phi^{\vphantom{\ddag}}_K(x), G \right) + \left( G, \Phi^\ddag_K(x) \right) \left( \Phi^{\vphantom{\ddag}}_K(x), F \right) \right] \total x \eqend{,}
\end{equation}
where $f\rvert_{\supp F \cap \supp G} = 1$, as an extension of the relation~\eqref{sec_gauge_brst_current}. We recall that the fulfilment of the classical BV master equation is equivalent of the existence of the BRST current $J^\mu$ according to~\eqref{sec_gauge_brst_current} (for constant cutoff $g \neq g(x)$)
\begin{equation}
\label{proof_brstward_current}
\left. \left( S_0 + L, \Phi^\ddag_K(x) \right) \left( \Phi^{\vphantom{\ddag}}_K(x), S_0 + L \right) \right\rvert_{g = \text{const.}} = \nabla_\mu J^\mu(x) \eqend{,}
\end{equation}
and the free BRST current $J^\mu_0$ is defined by the analogue equation without $L$,
\begin{equation}
\label{proof_brstward_freecurrent}
\left( S_0, \Phi^\ddag_K(x) \right) \left( \Phi^{\vphantom{\ddag}}_K(x), S_0 \right) = \nabla_\mu J^\mu_0(x) \eqend{.}
\end{equation}
The difference $J^\mu - J^\mu_0$ admits an expansion in powers of the coupling $g$, and we define $J^\mu_g(x)$ by replacing in this expansion the constant $g$ by the cutoff function $g(x)$, $J^\mu_g \equiv \left( J^\mu - J^\mu_0 \right)_{g \to g(x)}$, and using this
\begin{splitequation}
q(x) &\equiv \left( S_0, \Phi^\ddag_K(x) \right) \left( \Phi^{\vphantom{\ddag}}_K(x), L \right) + \left( L, \Phi^\ddag_K(x) \right) \left( \Phi^{\vphantom{\ddag}}_K(x), S_0 \right) \\
&\quad+ \left( L, \Phi^\ddag_K(x) \right) \left( \Phi^{\vphantom{\ddag}}_K(x), L \right) - \nabla_\mu J^\mu_g(x) \eqend{.}
\end{splitequation}
For all $x$ such that $g(y) = \text{const.}$ for all $y$ in an open neighbourhood of $x$, we have $J^\mu_g(x) = J^\mu(x) - J^\mu_0(x)$ and thus $q(x) = 0$ using equations~\eqref{proof_brstward_current} and~\eqref{proof_brstward_freecurrent}, and therefore $\supp q \subseteq \supp \nabla_\mu g$. In particular, its integral is well-defined, and we obtain using equation~\eqref{proof_brstward_antibracket}
\begin{equation}
\Delta Q \equiv \int q(x) \total x = \left( S_0, L \right) + \frac{1}{2} \left( L, L \right) = \frac{1}{2} \left( S_0 + L, S_0 + L \right) \eqend{,}
\end{equation}
which is of order $\bigo{g}$ since it contains at least one factor of $L$. We decompose $\Delta Q = \Delta Q^+ - \Delta Q^-$ in such a way that
\begin{equation}
\label{proof_brstward_deltaqpm_supp}
J^\pm\left( \supp \Delta Q^\pm \right) \cap J^\mp \left( \supp F \right) = \emptyset
\end{equation}
(see figure~\ref{fig_brst_charge}).
\begin{figure}[ht]
\centering
\begin{tikzpicture}
\def\figaspect{1.4}
\def\figscale{0.8}
\colorlet{mdark}{blue!60} 
\colorlet{mmed}{blue!30} 
\colorlet{mlight}{green!20} 
\colorlet{mhatch}{gray!90} 
\colorlet{mcaus}{red!40} 
\fill[pattern=north east lines, pattern color=mcaus] (0,-2.0*\figscale) -- +(-1.5*\figscale*\figaspect,+1.5*\figscale) -- +(-4.5*\figscale*\figaspect,-1.5*\figscale) -- +(4.5*\figscale*\figaspect,-1.5*\figscale) -- +(+1.5*\figscale*\figaspect,+1.5*\figscale) -- cycle;
\fill[pattern=north west lines, pattern color=mcaus] (0,2.0*\figscale) -- +(-1.5*\figscale*\figaspect,-1.5*\figscale) -- +(-4.5*\figscale*\figaspect,+1.5*\figscale) -- +(4.5*\figscale*\figaspect,+1.5*\figscale) -- +(+1.5*\figscale*\figaspect,-1.5*\figscale) -- cycle;
\fill[color=mlight] (3.0*\figscale*\figaspect,0) -- (0,3.0*\figscale) -- (-3.0*\figscale*\figaspect,0) -- (0,-3.0*\figscale) -- cycle;
\fill[color=mmed] (2.0*\figscale*\figaspect,0) -- (0,2.0*\figscale) -- (-2.0*\figscale*\figaspect,0) -- (0,-2.0*\figscale) -- cycle;
\fill[pattern=north east lines, pattern color=mhatch] (2.0*\figscale*\figaspect,0) -- (3.0*\figscale*\figaspect,0) -- (0,3.0*\figscale) -- (-3.0*\figscale*\figaspect,0) -- (-2.0*\figscale*\figaspect,0) -- (0,2.0*\figscale) -- cycle;
\fill[pattern=north west lines, pattern color=mhatch] (2.0*\figscale*\figaspect,0) -- (3.0*\figscale*\figaspect,0) -- (0,-3.0*\figscale) -- (-3.0*\figscale*\figaspect,0) -- (-2.0*\figscale*\figaspect,0) -- (0,-2.0*\figscale) -- cycle;
\draw (3.0*\figscale*\figaspect,0) -- (0,3.0*\figscale) -- (-3.0*\figscale*\figaspect,0) -- (0,-3.0*\figscale) -- cycle;
\fill[color=mcaus] (\figscale*\figaspect,0) -- (0,\figscale) -- (-\figscale*\figaspect,0) -- (0,-\figscale) -- cycle;
\draw (\figscale*\figaspect,0) -- (0,\figscale) -- (-\figscale*\figaspect,0) -- (0,-\figscale) -- cycle;
\draw[dashed] (\figscale*\figaspect,0) -- +(3.5*\figscale*\figaspect,-3.5*\figscale);
\draw[dashed] (-\figscale*\figaspect,0) -- +(-3.5*\figscale*\figaspect,-3.5*\figscale);
\draw[dashed] (\figscale*\figaspect,0) -- +(3.5*\figscale*\figaspect,3.5*\figscale);
\draw[dashed] (-\figscale*\figaspect,0) -- +(-3.5*\figscale*\figaspect,3.5*\figscale);
\node at (0,0) {supp\,$F$};
\node at (1.8*\figscale*\figaspect,2.8*\figscale) {supp\,$g$};
\draw (1.8*\figscale*\figaspect,2.5*\figscale) -- +(-0.65*\figscale*\figaspect,-0.65*\figscale);
\node at (-2.3*\figscale*\figaspect,2.8*\figscale) {$J^+(\text{supp}\,F)$};
\node at (2.3*\figscale*\figaspect,-2.8*\figscale) {$J^-(\text{supp}\,F)$};
\fill[color=mmed] (-2.5*\figscale*\figaspect,-4.5*\figscale) circle [radius=0.2*\figscale];
\node at (-1.8*\figscale*\figaspect,-4.5*\figscale) {$g=1$};
\fill[color=mlight] (-0.5*\figscale*\figaspect,-4.5*\figscale) circle [radius=0.2*\figscale];
\fill[pattern=north east lines, pattern color=mhatch] (-0.5*\figscale*\figaspect,-4.5*\figscale) circle [radius=0.2*\figscale];
\node at (0.35*\figscale*\figaspect,-4.5*\figscale) {supp\,$\Delta Q^+$};
\fill[color=mlight] (1.5*\figscale*\figaspect,-4.5*\figscale) circle [radius=0.2*\figscale];
\fill[pattern=north west lines, pattern color=mhatch] (1.5*\figscale*\figaspect,-4.5*\figscale) circle [radius=0.2*\figscale];
\node at (2.35*\figscale*\figaspect,-4.5*\figscale) {supp\,$\Delta Q^-$};
\end{tikzpicture}
\caption{The cutoff functions involved in the construction of the interacting BRST differential.}
\label{fig_brst_charge}
\end{figure}
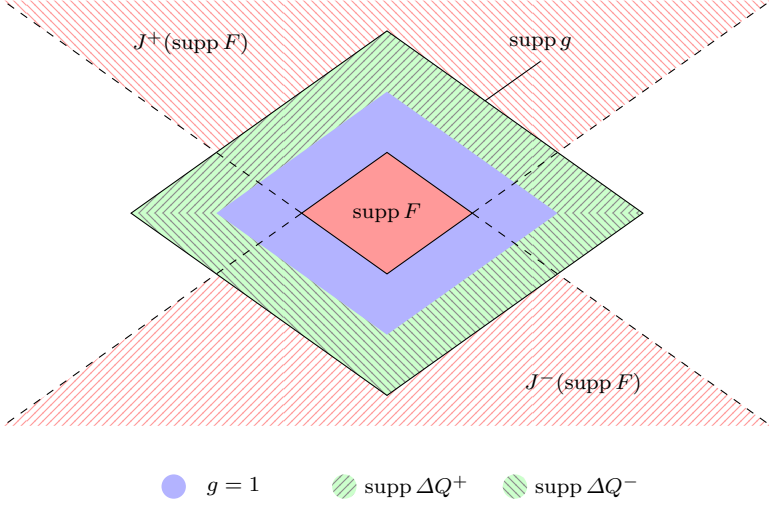
From the causal factorisation of interacting time-ordered products~\eqref{sec_aqft_inttimeord_fact}, it then follows that
\begin{splitequation}
\label{proof_brstward_qcomm}
&\left[ \mathcal{T}_L\left( \Delta Q^- \right), \mathcal{T}_L\left( F^{\otimes k} \right) \right]_{\star_\hbar} = \mathcal{T}_L\left[ \left( \Delta Q^+ - \Delta Q^- \right) \otimes F^{\otimes k} \right] \\
&\hspace{12em}- \mathcal{T}_L\left( \Delta Q^+ - \Delta Q^- \right) \star_\hbar \mathcal{T}_L\left( F^{\otimes k} \right) \\
&\quad= \frac{1}{2} \mathcal{T}_L\left[ \left( S_0 + L, S_0 + L \right) \otimes F^{\otimes k} \right] - \frac{1}{2} \mathcal{T}_L\left[ \left( S_0 + L, S_0 + L \right) \right] \star_\hbar \mathcal{T}_L\left( F^{\otimes k} \right) \eqend{.}
\end{splitequation}

We now define for any $\mathcal{G} \in \overline{\mathfrak{A}}_0$ [with the graded commutator~\eqref{sec_aqft_commutator}]
\begin{equation}
\label{proof_brstward_brstdef}
\st \mathcal{G} \equiv \st_0 \mathcal{G} + \frac{1}{\mathi \hbar} \left[ \mathcal{T}_L\left( \Delta Q^- \right), \mathcal{G} \right]_{\star_\hbar} \eqend{,}
\end{equation}
and since $\Delta Q^-$ is fermionic and time-ordered products preserve the grading, $\st$ is a graded derivation, and $\st - \st_0$ is of order $\bigo{g}$ because $\Delta Q^-$ is. While the decomposition of $\Delta Q$ into $\Delta Q^\pm$ is not unique, the action of $\st$ is nevertheless well defined on the interacting time-ordered products of $F$, because for any other decomposition $\Delta Q^-$ differs only by terms whose support is spacelike separated from the support of $F$ by condition~\eqref{proof_brstward_deltaqpm_supp}, and which therefore do not contribute to the commutator~\eqref{proof_brstward_qcomm} or~\eqref{proof_brstward_brstdef}. Using the anomalous Ward identity for the free BRST differential $\st_0$~\eqref{remark_freebrstward_brstanomward}, it follows that
\begin{splitequation}
\label{proof_brstward_s0_on_tl}
&\st_0 \mathcal{T}_L\left[ \exp_\otimes\left( \frac{\mathi}{\hbar} F \right) \right] = - \frac{\mathi}{\hbar} \mathcal{T}_L\left[ \left( \brst_0 L + \frac{1}{2} (L,L) + \mathcal{A}\left( \mathe_\otimes^L \right) \right) \right] \star_\hbar \mathcal{T}_L\left[ \exp_\otimes\left( \frac{\mathi}{\hbar} F \right) \right] \\
&\qquad+ \frac{\mathi}{\hbar} \mathcal{T}_L\left[ \left( \brst_0 (L+F) + \frac{1}{2} (L+F,L+F) + \mathcal{A}\left( \mathe_\otimes^{L+F} \right) \right) \otimes \exp_\otimes\left( \frac{\mathi}{\hbar} F \right) \right] \eqend{,}
\end{splitequation}
and because
\begin{equation}
\brst_0 G + \frac{1}{2} (G,G) = \frac{1}{2} \left( S_0 + G, S_0 + G \right)
\end{equation}
and we assume that $\mathcal{A}\left[ \exp_\otimes(L) \right] = 0$, we have
\begin{splitequation}
&\st_0 \mathcal{T}_L\left[ \exp_\otimes\left( \frac{\mathi}{\hbar} F \right) \right] = \frac{\mathi}{\hbar} \left[ \mathcal{T}_L\left( \Delta Q^- \right), \mathcal{T}_L\left[ \exp_\otimes\left( \frac{\mathi}{\hbar} F \right) \right] \right]_{\star_\hbar} \\
&\qquad+ \frac{\mathi}{\hbar} \mathcal{T}_L\left[ \left( (S_0+L,F) + \frac{1}{2} (F,F) + \mathcal{A}\left( \mathe_\otimes^{L+F} \right) \right) \otimes \exp_\otimes\left( \frac{\mathi}{\hbar} F \right) \right]
\end{splitequation}
using equation~\eqref{proof_brstward_qcomm}. Since $g(x) = 1$ on $\supp F$, we have furthermore $(S_0 + L, F) = \brst F$, and the interacting anomalous Ward identity~\eqref{thm_brstward_anomward} follows.

In various degrees of generality, it has been shown~\cite{duetschfredenhagen1999,duetschboas2002,duetschfredenhagen2003,hollands2008,rejzner2015,taslimitehrani2017} that on-shell $\st_0$ can be written as the commutator with the free BRST charge $\mathcal{T}( Q_0 )$ whenever this charge is well-defined, and that $\st$ can be written as the commutator with the interacting BRST charge $\mathcal{T}_L( Q )$ (when it is well-defined). For this to work, one first has to show that the BRST current is conserved, i.e., that $\mathcal{T}_L\left[ \nabla_\mu J^\mu(x) \right]$~\eqref{proof_brstward_current} vanishes on-shell, which either is a condition that needs to be assumed~\cite{rejzner2015} or can be fulfilled by a redefinition of time-ordered products if a certain cohomology class of the classical BRST differential is empty~\cite{hollands2008,taslimitehrani2017}. One then embeds the region of interest into a spacetime with compact Cauchy surfaces and defines the interacting BRST charge as the integral of the BRST current smeared with a closed $3$-form over a Cauchy surface. Since the Cauchy surface is compact and the BRST current is conserved, the interacting BRST charge is well-defined and independent of the choice of Cauchy surface and it follows that $\Delta Q^+ = \Delta Q^-$. The anomalous Ward identity~\eqref{thm_brstward_anomward} is then obtained by choosing one of the surfaces in the commutator to the future of $\supp F$, and one in the past~\cite{taslimitehrani2017}. We see, however, that this is in fact unnecessary, and that one can simply incorporate a non-vanishing divergence of the BRST current in the definition of $Q$. Moreover, one does not need to restrict to spacetime with compact Cauchy surfaces, and the resulting identity holds off-shell, analogously to the classical BRST invariance which also is an off-shell identity. The only question is thus whether $\Delta Q^-$ as defined above reduces on-shell to the literature definition, which might not be obvious. Let us denote a possible difference by $\delta Q$. Since both definitions give the same anomalous interacting Ward identity~\eqref{thm_brstward_anomward}, it follows that $\mathcal{T}_L\left( \delta Q \right)$ (anti-)commutes with all interacting time-ordered products. As a formal power series in the coupling $g$, we can invert the interacting time-ordered products and express $\Phi_K(x) = \mathcal{T}_L\left[ \Phi_K(x) \right] + \bigo{g}$, and since the free-field algebra $\overline{\mathfrak{A}}_0$ is complete, we can express the higher-order terms recursively as interacting time-ordered products and obtain $\Phi_K(x) = \mathcal{T}_L\left[ \Phi_K(x) + \bigo{g} \right]$. It follows that $\mathcal{T}_L\left( \delta Q \right)$ (anti-)commutes with all $\Phi_K$, and thus is proportional to the identity~\cite{hollandswald2001}. Since $\delta Q$ is fermionic and the time-ordered products preserve the grading, the proportionality factor must contain an odd number of fermionic antifields, but because all antifields vanish on-shell, we have $\mathcal{T}_L\left( \delta Q \right) = 0$ on-shell.\footnote{Note that since the definition of $\Delta Q$ and thus $\delta Q$ depend on $\supp F$, what we have really shown is that $\supp \delta Q \cap \supp F = \emptyset$, which however is all that we care about.}

It remains to show that $\st^2 = 0$, at least when acting on interacting time-ordered products. For this we notice that the functionals on the right-hand side of the interacting anomalous Ward identity~\eqref{thm_brstward_anomward} have the same or smaller support as $F$, and we can thus apply $\st$ again, which gives
\begin{equation}
\st^2 \mathcal{T}_L\left[ \exp_\otimes\left( \frac{\mathi}{\hbar} F \right) \right] = \frac{\mathi}{\hbar} \mathcal{T}_L\left[ N \otimes \exp_\otimes\left( \frac{\mathi}{\hbar} F \right) \right]
\end{equation}
with
\begin{splitequation}
N &= \brst \left[ \brst F + \frac{1}{2} (F,F) + \mathcal{A}\left( \mathe_\otimes^{L+F} \right) \right] - \left( \left[ \brst F + \frac{1}{2} (F,F) + \mathcal{A}\left( \mathe_\otimes^{L+F} \right) \right], F \right) \\
&\quad- \mathcal{A}\left[ \mathe_\otimes^{L+F} \otimes \left( \brst F + \frac{1}{2} (F,F) + \mathcal{A}\left( \mathe_\otimes^{L+F} \right) \right) \right] \\
&= \brst \mathcal{A}\left( \mathe_\otimes^{L+F} \right) - \left( \mathcal{A}\left( \mathe_\otimes^{L+F} \right), F \right) - \mathcal{A}\left[ \mathe_\otimes^{L+F} \otimes \left( \brst F + \frac{1}{2} (F,F) + \mathcal{A}\left( \mathe_\otimes^{L+F} \right) \right) \right] \eqend{.}
\end{splitequation}
where we used that $\brst^2 = 0$, that $((F,F),F) = 0$ because of the graded Jacobi identity~\eqref{sec_gauge_antibracket_jacobi} for the antibracket, and that $\brst (F,F) = \left( S, (F,F) \right) = 2 \left( (S,F), F \right) = 2 (\brst F,F)$ again because of the graded Jacobi identity. We now consider the consistency condition for the anomaly of Theorem~\ref{thm_freebrstward_cons} for $F \to L+F$, which reads
\begin{splitequation}
\brst \mathcal{A}\left[ \mathe_\otimes^{L+F} \right] &= - \left( F, \mathcal{A}\left[ \mathe_\otimes^{L+F} \right] \right) + \mathcal{A}\left[ \left( \brst F + \frac{1}{2} (F,F) \right) \otimes \mathe_\otimes^{L+F} \right] \\
&\quad+ \mathcal{A}\left[ \mathcal{A}\left( \mathe_\otimes^{L+F} \right) \otimes \mathe_\otimes^{L+F} \right] + \frac{1}{2} \mathcal{A}\left[ \left( S_0 + L, S_0 + L \right) \otimes \mathe_\otimes^L \right] \eqend{.}
\end{splitequation}
Here, we used that
\begin{equation}
\left( S_0 + L, \mathcal{A}\left[ \mathe_\otimes^{L+F} \right] \right) = \brst \mathcal{A}\left[ \mathe_\otimes^{L+F} \right]
\end{equation}
because the anomaly is supported on the total diagonal, $\mathcal{A}\left[ \exp_\otimes(L) \right] = 0$ and $g(x) = \text{const.}$ on $\supp F$, for the same reason $\left( S_0 + L, F \right) = \brst F$, and since $(S,S) = 0$ also $\mathcal{A}\left[ \left( S_0 + L, S_0 + L \right) \otimes \mathe_\otimes^{L+F} \right] = \mathcal{A}\left[ \left( S_0 + L, S_0 + L \right) \otimes \mathe_\otimes^L \right]$. It follows that
\begin{equation}
N = \frac{1}{2} \mathcal{A}\left[ \left( S_0 + L, S_0 + L \right) \otimes \mathe_\otimes^L \right] = 0 \eqend{,}
\end{equation}
where the last equality follows from the consistency condition for the anomaly of Theorem~\ref{thm_freebrstward_cons} for $F \to L$ and $\mathcal{A}\left[ \exp_\otimes(L) \right] = 0$.
\end{proof}
\begin{remark*}
Finally, we can compare this interacting anomalous Ward identity with the identities of Fredenhagen and Rejzner~\cite{fredenhagenrejzner2013,rejzner2015} and Taslimi Tehrani~\cite{taslimitehrani2017}. While the results of~\cite{taslimitehrani2017} only hold on-shell, the analogue of our interacting Ward identity is given by their Remark 20, which expresses the commutator of the interacting BRST charge with an interacting time-ordered product, and contains the same right-hand side. Our interacting Ward identity can thus be seen as an off-shell generalisation of this identity. The comparison with~\cite{fredenhagenrejzner2013,rejzner2015} is somewhat more involved: we first note that their interacting BRST operator is obtained by conjugating the free BRST operator with the interaction, and thus corresponds to our $\st_0$ acting on an interacting time-ordered product, with the corresponding Ward identity given by~\eqref{proof_brstward_s0_on_tl}. In the formal algebraic adiabatic limit, when the quantum master equation of~\cite{duetschboas2002,brenneckeduetsch2008,fredenhagenrejzner2013,rejzner2015} holds, we have $\brst_0 L + \frac{1}{2} (L,L) = 0$ such that $\Delta Q = 0$ and $\st$ reduces to $\st_0$. In this limit, our interacting BRST differential thus agrees with the one of~\cite{fredenhagenrejzner2013,rejzner2015}, and can thus be seen as a generalisation which holds without taking the algebraic adiabatic limit. Moreover, our construction is also applicable to theories with open gauge algebras (such as supersymmetric theories), which these other works do not cover.
\end{remark*}

We now analyse the structure of the terms appearing in the interacting Ward identity of Theorem~\ref{thm_brstward} in more detail:
\begin{definition}
\label{def_brackets_linfty}
Given the anomaly $\mathcal{A}$ for the full free BRST differential $\st_0$, we define the $n$-ary quantum brackets $[\cdot]_\hbar$ by
\begin{equations}[brackets_linfty]
[-]_\hbar &\equiv 0 \eqend{,} \\
[F_1]_\hbar &\equiv \brst F_1 + (-1)^{\epsilon_1} \mathcal{A}\left[ F_1 \otimes \mathe_\otimes^L \right] \eqend{,} \\
[F_1,F_2]_\hbar &\equiv (-1)^{\epsilon_1} (F_1,F_2) + (-1)^{\epsilon_1+\epsilon_2} \mathcal{A}\left[ F_1 \otimes F_2 \otimes \mathe_\otimes^L \right] \eqend{,} \\
[F_1,\ldots,F_k]_\hbar &\equiv (-1)^{\epsilon_1+\cdots+\epsilon_k} \mathcal{A}\left[ F_1 \otimes \cdots \otimes F_k \otimes \mathe_\otimes^L \right] \eqend{,} \qquad k \geq 3 \eqend{,}
\end{equations}
where we set $\epsilon_i = \epsilon(F_i)$. We alternatively write $\stq F_1 \equiv [F_1]_\hbar$, called the \emph{quantum BRST differential}, and $(F_1,F_2)_\hbar \equiv (-1)^{\epsilon_1} [F_1,F_2]_\hbar$, the \emph{quantum antibracket}.
\end{definition}
\begin{remark*}
To shorten notation, we also write $[F^n]_\hbar = [F,\ldots,F]_\hbar$ with $n$ factors of $F$ inside the last bracket. The signs in the definition ensure that the quantum brackets are graded symmetric when two arguments are exchanged, and that they are intrinsically odd in the sense that $[\alpha G,F^k]_\hbar = (-1)^{\epsilon_\alpha} \alpha [G,F^k]_\hbar$. We see that they are multilinear maps $[\cdot]_\hbar\colon \mathcal{F}^{\otimes n} \to \mathcal{F}$, and that they are jointly induced by the free BRST differential $\brst_0$ and the interaction $L$. Furthermore, they implicitly depend on the choice of the time-ordered products, which determines the anomaly $\mathcal{A}$. Since from its definition~\eqref{sec_gauge_antibracket_def} the antibracket $(\cdot,\cdot)$ is supported on the diagonal and by Theorem~\ref{thm_anomward} the anomaly $\mathcal{A}$ is supported on the total diagonal, the quantum antibrackets are also supported on the total diagonal.
\end{remark*}
Using the quantum brackets, we can write the interacting anomalous Ward identity~\eqref{thm_brstward_anomward} under the conditions of Theorem~\ref{thm_brstward} in the succinct form
\begin{equation}
\label{anomward_quantumbracket}
\st \mathcal{T}_L\left[ \exp_\otimes\left( \frac{\mathi}{\hbar} F \right) \right] = \frac{\mathi}{\hbar} \mathcal{T}_L\left[ \left[ \mathe^F \right]_\hbar \otimes \exp_\otimes\left( \frac{\mathi}{\hbar} F \right) \right] \eqend{.}
\end{equation}

\begin{theorem}
\label{thm_linfty}
If the anomaly of the interaction $L$ with cutoff function $g$ vanishes, $\mathcal{A}\left[ \exp_\otimes\left( L \right) \right] = 0$, the quantum brackets $[\cdot]_\hbar$ form an $L_\infty$ algebra (in the $b$-picture) over $\mathcal{F}_V$ for all $V$ such that $g\rvert_V = 1$. In particular, the quantum BRST differential is nilpotent, and the quantum antibracket is a well-defined graded symmetric map between the cohomology classes of $\stq$,
\begin{equation*}
(\cdot,\cdot)_\hbar \colon H^g(\stq) \otimes H^{g'}(\stq) \to H^{g+g'+1}(\stq) \eqend{,}
\end{equation*}
and fulfils the graded Jacobi identity in cohomology.
\end{theorem}
\begin{remark*}
The graded symmetry of the quantum antibracket is the same as for the classical antibracket~\eqref{sec_gauge_antibracket_gradsym}. While one might expect the natural grading of the $L_\infty$ algebra to be given by the ghost number, such a grading is incompatible with the graded symmetry of the anomaly $\mathcal{A}$ whenever fermionic symmetries (such as supersymmetry) exist, since their ghosts are bosonic. The grading which is respected by the $L_\infty$ algebra is instead just the Grassmann $\mathbb{Z}_2$ grading, inherited from the grading of the anomaly $\mathcal{A}$.
\end{remark*}
\begin{proof}
If $\mathcal{A}\left[ \exp_\otimes\left( L \right) \right] = 0$ and $g\rvert_{\supp F} = 1$, the term of order $F^k$ in the interacting anomalous Ward identity~\eqref{thm_brstward_anomward}, \eqref{anomward_quantumbracket} reads
\begin{equation}
\label{proof_linfty_anomward}
\st \mathcal{T}_L\left[ F^{\otimes k} \right] = \sum_{\ell=1}^k \frac{k!}{\ell! (k-\ell)!} \left( \frac{\hbar}{\mathi} \right)^{\ell-1} \mathcal{T}_L\left[ [F^\ell]_\hbar \otimes F^{\otimes (k-\ell)} \right] \eqend{,}
\end{equation}
where we assume $F$ to be bosonic for simplicity. By polarisation (replacing $F \to F + \alpha G$, taking a derivative with respect to $\alpha$ and setting $\alpha = 0$), it follows that
\begin{splitequation}
\label{proof_linfty_anomward2}
\st \mathcal{T}_L\left[ G \otimes F^{\otimes k} \right] &= \sum_{\ell=0}^k \frac{k!}{\ell! (k-\ell)!} \left( \frac{\hbar}{\mathi} \right)^\ell \mathcal{T}_L\left[ [F^\ell,G]_\hbar \otimes F^{\otimes (k-\ell)} \right] \\
&\quad+ \sum_{\ell=1}^k \frac{k!}{\ell! (k-\ell)!} \left( \frac{\hbar}{\mathi} \right)^{\ell-1} \mathcal{T}_L\left[ [F^\ell]_\hbar \otimes G \otimes F^{\otimes (k-\ell)} \right]
\end{splitequation}
for both fermionic and bosonic $G$, using the graded symmetry of the time-ordered products and the quantum brackets. Since $\st^2 = 0$ by Theorem~\ref{thm_brstward}, we apply $\st$ again on equation~\eqref{proof_linfty_anomward}, use equation~\eqref{proof_linfty_anomward2} and shift summation indices to obtain
\begin{splitequation}
0 &= \sum_{\ell=1}^k \sum_{n=\ell}^k \frac{k!}{(n-\ell)! \ell! (k-n)!} \left( \frac{\hbar}{\mathi} \right)^{n-1} \mathcal{T}_L\left[ [F^{n-\ell},[F^\ell]_\hbar]_\hbar \otimes F^{\otimes (k-n)} \right] \\
&\quad+ \sum_{\ell=1}^k \sum_{n=\ell+1}^k \frac{k!}{(n-\ell)! \ell! (k-n)!} \left( \frac{\hbar}{\mathi} \right)^{n-2} \mathcal{T}_L\left[ [F^{n-\ell}]_\hbar \otimes [F^\ell]_\hbar \otimes F^{\otimes (k-n)} \right] \eqend{.}
\end{splitequation}
We then interchange sums using
\begin{equation}
\sum_{\ell=1}^k \sum_{n=\ell+s}^k a_{\ell,n} = \sum_{n=1}^k \sum_{\ell=1}^{n-s} a_{\ell,n} \eqend{,}
\end{equation}
and since the $[\cdot]_\hbar$ are fermionic, the sum on the second line vanishes by symmetry: take it twice and change $\ell \to n-\ell$ in the second sum, where a single minus sign results from the interchange of $[F^\ell]_\hbar$ and$[F^{n-\ell}]_\hbar$. It follows that
\begin{equation}
\sum_{n=1}^k \left( \frac{\hbar}{\mathi} \right)^{n-1} \frac{k!}{n! (k-n)!} \sum_{\ell=1}^n \frac{n!}{(n-\ell)! \ell!} \mathcal{T}_L\left[ [F^{n-\ell},[F^\ell]_\hbar]_\hbar \otimes F^{\otimes (k-n)} \right] = 0 \eqend{,}
\end{equation}
and since a time-ordered product only vanishes when its argument vanishes, we obtain by induction in $k$ that
\begin{equation}
\label{proof_linfty_relation}
\sum_{\ell=1}^n \frac{n!}{(n-\ell)! \ell!} [F^{n-\ell},[F^\ell]_\hbar]_\hbar = 0 \eqend{.}
\end{equation}
This relation (for bosonic $F$), together with the graded symmetry under the interchange of two arguments and the intrinsic oddness of the quantum brackets as explained in the remark after Definition~\ref{def_brackets_linfty}, are the defining relations for an $L_\infty$ algebra (in the $b$-picture)~\cite{ladastesheff1993,hohmzwiebach2017}. In particular, for $n = 1,2,3$ we obtain
\begin{equations}[]
&[[F]_\hbar]_\hbar = 0 \eqend{,} \\
&2 [F,[F]_\hbar]_\hbar + [[F,F]_\hbar]_\hbar = 0 \eqend{,} \\
&3 [F,F,[F]_\hbar]_\hbar + 3 [F,[F,F]_\hbar]_\hbar + [[F^3]_\hbar]_\hbar = 0 \eqend{,}
\end{equations}
which express the nilpotency of the quantum BRST differential $\stq^2 F = 0$, a compatibility condition between quantum BRST differential and quantum antibracket $\stq (F,F)_\hbar = - 2 (F, \stq F)_\hbar$, and the failure of the Jacobi identity for the quantum antibracket, $(F,(F,F)_\hbar)_\hbar = - [F,F,\stq F]_\hbar - \frac{1}{3} \stq [F^3]_\hbar$. Since $\stq^2 = 0$ and by Theorem~\ref{thm_anomward} quantum numbers are preserved, analogously to the cohomology classes of the classical BRST differential~\eqref{sec_gauge_cohom_def} we can define the cohomology classes of the quantum BRST differential at ghost number $g$:
\begin{equation}
H^g(\stq) \equiv \frac{\mathrm{Ker}(\stq\colon \mathcal{F}^g \to \mathcal{F}^{g+1})}{\mathrm{Im}(\stq\colon \mathcal{F}^{g-1} \to \mathcal{F}^g)} \eqend{.}
\end{equation}
The compatibility condition ensures that the quantum antibracket is a well-defined map between cohomology classes as stated in the theorem. By polarisation and using the graded symmetry of $[F,F]_\hbar$ under the exchange of the arguments together with its intrinsic oddness, for arbitrary $F$ and $G$ we obtain
\begin{equation}
\stq ( F, G )_\hbar = ( \stq F, G )_\hbar - (-1)^{\epsilon_F} ( F, \stq G )_\hbar \eqend{.}
\end{equation}
Choosing further $F + \stq F'$, $G + \stq G'$ for some $F' \in \mathcal{F}^{g-1}$, $G' \in \mathcal{F}^{g'-1}$ and using repeatedly the compatibility condition it follows that
\begin{equation}
( F + \stq F', G + \stq G' )_\hbar = ( F, G )_\hbar + \stq \left[ ( F', G+\stq G' )_\hbar - (-1)^{\epsilon_F} ( F, G' )_\hbar \right] \eqend{.}
\end{equation}
Therefore, the quantum antibracket is well-defined as a map from $H^g(\stq) \otimes H^{g'}(\stq) \to H^{g+g'+1}(\stq)$, since it is annihilated by $\stq$ for $\stq$-invariant entries $F$ and $G$, and choosing a different representative of an element in $H^g(\stq)$ or $H^{g'}(\stq)$ only gives a different representative of the resulting element in $H^{g+g'+1}(\stq)$. Furthermore, since for representatives $F$ of elements in $H^g(\stq)$ we have $\stq F = 0$, the Jacobi identity holds in cohomology:
\begin{equation}
(F,(F,F)_\hbar)_\hbar = 0 \mod \stq \eqend{.}
\end{equation}
\end{proof}

Lastly, we want to determine the possible observables in the quantum theory. Classical observables are (representatives of) elements of the comology of the BRST operator at zero ghost number $H^0(\brst)$, but an interacting time-ordered product of a classical observable is not $\st$-invariant because of the anomaly $\mathcal{A}$, as can be seen from the interacting anomalous Ward identity~\eqref{thm_brstward_anomward} (at linear order in $F$). However, a $\stq$-invariant classical functional will give a $\st$-invariant interacting time-ordered product, and quantum observables are thus (representatives of) elements of the cohomology of the quantum BRST operator $H^0(\stq)$. Similarly, while classically the product of two observables is again an observable because $\brst$ is a derivation, the time-ordered product of two quantum observables is not necessarily $\st$-invariant because of the higher anomalies in the interacting anomalous Ward identity~\eqref{thm_brstward_anomward}. It is therefore necessary to know if all classical observables can be extended to the quantum theory, and whether it is possible to correct their time-ordered products to obtain $\st$-invariant algebra elements. Concretely, we have:
\begin{theorem}
\label{thm_obs}
If the anomaly of the interaction $L$ with cutoff function vanishes, $\mathcal{A}\left[ \exp_\otimes\left( L \right) \right] = 0$, if the cohomology of the classical BRST differential at ghost number $1$ vanishes, $H^1(\brst) = \{0\}$, and if there exists a contracting homotopy $\mathsf{h}\colon \mathcal{F}^k \to \mathcal{F}^{k-1}$ with respect to the classical BRST differential, the following holds:
\begin{enumerate}
\item To each classical observable corresponds an observable in the quantum theory, that is, each representative of an element in $H^0(\brst)$ can be extended to a representative of an element in $H^0(\stq)$.
\item The cohomology of the quantum BRST differential $\stq$ at ghost number $1$ vanishes: $H^1(\stq) = \{0\}$.
\item There exist multilinear maps $\mathcal{C}_n\colon \mathcal{F}^{0\otimes n} \to \mathcal{F}^0$ (the contact terms), such that the interacting time-ordered product (as generating functional)
\begin{equation}
\label{thm_obs_timeordered_contact}
\mathcal{T}_L\left[ \exp_\otimes\left[ \frac{\mathi}{\hbar} F - \frac{\mathi}{\hbar} \mathcal{C}\left( \mathe_\otimes^F \right) \right] \right] \qquad\text{with}\qquad \mathcal{C}\left( \mathe_\otimes^F \right) = \sum_{k=0}^\infty \frac{1}{k!} \mathcal{C}_k\left( F^{\otimes k} \right)
\end{equation}
is $\st$-invariant and independent of the choice of representative $F$ of an element in $H^0(\stq)$, up to $\st$-exact terms. They satisfy the identities (in the sense of generating functionals)
\begin{equation}
\label{thm_obs_contactdef}
\left[ \exp\left( F - \mathcal{C}\left( \mathe_\otimes^F \right) \right) \right]_\hbar = 0 \eqend{,}
\end{equation}
$\mathcal{C}_0(-) = 0 = \mathcal{C}_1(F)$, and (again in the sense of generating functionals)
\begin{equation}
\label{thm_obs_contactdef2}
\mathcal{C}\left( \mathe_\otimes^F \otimes \stq G \right) = \left[ 1 - \exp\left( F - \mathcal{C}\left( \mathe_\otimes^F \right) \right), G \right]_\hbar
\end{equation}
for $G \in \mathcal{F}^{-1}$, and can in fact be determined from these identities.
\end{enumerate}
\end{theorem}
\begin{remark*}
The existence of the contracting homotopy $\mathsf{h}$ is necessary in order to determine the maps $\mathcal{C}_n$ from the given identities as multilinear maps. $\mathsf{h}$ is in particular a right inverse for the BRST differential at ghost number $1$ for closed functionals, i.e., we have $\brst \mathsf{h} G = G$ for all $G \in \mathcal{F}^1$ such that $\brst G = 0$. Such a homotopy can be constructed perturbatively if one has a contracting homotopy $\mathsf{h}_0$ with respect to the free BRST differential $\brst_0$, in complete analogy to the construction of the quantum homology $\mathsf{h}_\hbar$ from $\mathsf{h}$ which is done in the proof. The homotopy $\mathsf{h}_0$ in turn can be constructed as follows~\cite{henneauxteitelboim1992,barnichetal2000}: one considers the jet space of fields, antifields and their derivatives, and chooses a basis $\{ u_i, v_i \}_{i \in I} \cup \{ w_j \}_{j \in J}$ of jet space such that $\brst_0 u_i = v_i$ and $\brst_0 w_j = 0$. We then define $\tilde{\mathsf{h}}_0$ by $\tilde{\mathsf{h}}_0 w_j = 0 = \tilde{\mathsf{h}}_0 u_i$, $\tilde{\mathsf{h}}_0 v_i = u_i$, linearity and a graded Leibniz rule. It follows that $\tilde{\mathsf{h}}_0$ satisfies $\brst_0 \tilde{\mathsf{h}}_0 + \tilde{\mathsf{h}}_0 \brst_0 = N_{u,v}$, where $N_{u,v}$ is the counting operator for the number of $u_i$ and $v_i$ basis elements. $N_{u,v}$ defines a grading on jet space, and $\mathsf{h}_0$ is defined by setting its action equal to the action of $\tilde{\mathsf{h}}_0/n$ on the subspace of jet space where $N_{u,v}$ has positive eigenvalue $n > 0$, and by defining $\mathsf{h}_0 w = 0$ for all $w$ in the subspace of eigenvalue $0$. It follows that the contracting homotopy relation $\brst_0 \mathsf{h}_0 + \mathsf{h}_0 \brst_0 = \operatorname{id}$ holds on all subspaces of positive eigenvalue $n > 0$. By construction, the elements of the subspace of vanishing eigenvalue $n = 0$ are representatives of elements of the cohomology of $\brst_0$, which for most theories will not be empty. However, if $H(\brst) = \{0\}$ (at the ghost number under consideration) these representatives cannot be extended to the (classical) interacting theory, and one can restrict to the subspace of positive eigenvalues $n > 0$. If moreover the BRST differential $\brst$ does map this subspace into itself (as it does for all the examples given previously), the construction of the proof is still applicable. Lastly, we note that the concrete choice of basis in jet space depends on the theory under consideration and its construction is not entirely straightforward; in particular antifields will belong to the $u_i$ and equation-of-motion terms to the $v_i$, while only those (BRST-invariant) derivatives of fields that are not determined by the equations of motion or by Noether identities belong to the $w_j$~\cite{barnichetal2000}. However, a systematic construction is possible for all theories that are regular in the sense of~\cite{barnichetal2000}, which in particular includes the usual free scalar, spinor, ... theories, such that the existence of $\mathsf{h}$ does not constitute an additional restriction for them.

The maps $\mathcal{C}_n$ are called contact terms because they are supported on the total diagonal, which follows because the quantum brackets $[\cdot]_\hbar$ are supported on the total diagonal by the Remark after Def.~\ref{def_brackets_linfty}. It follows moreover that $\mathcal{C}_n\left( F^{\otimes n} \right) = \bigo{\hbar}$ whenever $(F,F) = \bigo{\hbar}$, which in particular happens when $F$ can be chosen independent of antifields (such as for classically gauge-invariant observables). The contact terms then represent a true quantum effect. The form of the interacting time-ordered products including contact terms~\eqref{thm_obs_timeordered_contact} is reminiscent of a field redefinition~\eqref{sec_aqft_renormfreedom}, which could be obtained by the local map
\begin{equation}
\mathcal{Z}\left( F^{\otimes n} \right) = - \mathcal{C}\left( F^{\otimes n} \right) \eqend{,} \qquad \mathcal{Z}\left( F^{\otimes n} \otimes L^{\otimes k} \right) = 0 \qquad (k \geq 1) \eqend{.}
\end{equation}
In order to maintain the field-independence property of the maps $\mathcal{Z}_n$~\eqref{sec_aqft_renorm_fieldindependence}, one would also need to simultaneously redefine all products of $F$ with other fields, and all submonomials of $F$, i.e., all terms which are obtained by one or more derivatives of $F$ with respect to fields. While one does not expect problems in principle, it is not fully clear whether this can be done without destroying other important properties: For example, if $F$ is the product of the interaction Lagrangian with another invariant polynomial, one might introduce a non-vanishing anomaly of the interaction $L$ with cutoff function $\mathcal{A}\left[ \exp_\otimes\left( L \right) \right]$ by such a redefinition~\cite{zahnprivate}.
\end{remark*}
\begin{proof}
The first result is known in the more general context of filtrations of coboundary operators (see~\cite{piguetsorella} and references therein), and we include its proof only for completeness. We choose a representative $F$ of an element in $H^0(\stq)$ and expand it in powers of $\hbar$,
\begin{equation}
\label{proof_obs_f_expansion}
F = \sum_{k=0}^\infty \hbar^k F^{(k)} \eqend{.}
\end{equation}
We also expand the quantum BRST differential $\stq$ from Definition~\ref{def_brackets_linfty},
\begin{equation}
\stq = \sum_{k=0}^\infty \hbar^k \stq^{(k)} \eqend{,}
\end{equation}
and since the anomaly $\mathcal{A}$ is at least of order $\hbar$ according to Theorem~\ref{thm_anomward}, we have $\stq^{(0)} = \brst$. Comparing powers of $\hbar$, we thus obtain the relations
\begin{equation}
\label{proof_obs_fk_cohomconds}
\brst F^{(0)} = 0 \eqend{,} \qquad \brst F^{(k)} = - \sum_{\ell=1}^k \stq^{(\ell)} F^{(k-\ell)} \qquad (k \geq 1) \eqend{,}
\end{equation}
and thus any representative $F$ of an element in $H^0(\stq)$ gives a representative $F^{(0)}$ of an element in $H^0(\brst)$. To show the inverse, we note that the nilpotency of $\stq$ entails
\begin{equation}
\label{proof_obs_stq_nilpot}
\brst^2 = 0 \eqend{,} \qquad \brst \stq^{(k)} = - \sum_{\ell=1}^k \stq^{(\ell)} \stq^{(k-\ell)} \qquad (k \geq 1) \eqend{.}
\end{equation}
Given $F^{(0)}$ with $\brst F^{(0)} = 0$, we then have to construct $F^{(k)}$ such that the relations~\eqref{proof_obs_fk_cohomconds} are fulfilled, which can be done by induction. For $k = 1$, the relation~\eqref{proof_obs_stq_nilpot} reads $\brst \stq^{(1)} = - \stq^{(1)} \brst$, and we have therefore
\begin{equation}
\brst \stq^{(1)} F^{(0)} = - \stq^{(1)} \brst F^{(0)} = 0 \eqend{.}
\end{equation}
Since by assumption $H^1(\brst) = \{0\}$, it follows that
\begin{equation}
\stq^{(1)} F^{(0)} = \brst B^{(1)}
\end{equation}
for some local functional $B^{(1)}$, and we simply define $F^{(1)} = - B^{(1)}$, fulfilling the condition~\eqref{proof_obs_fk_cohomconds} for $k = 1$. Assume thus that all $F^{(k')}$ with $k' < k$ have been constructed, and calculate
\begin{splitequation}
\brst \sum_{\ell=1}^k \stq^{(\ell)} F^{(k-\ell)} &= - \sum_{\ell=1}^k \sum_{n=1}^\ell \stq^{(n)} \stq^{(\ell-n)} F^{(k-\ell)} \\
&= - \sum_{\ell=1}^k \sum_{n=1}^{\ell-1} \stq^{(n)} \stq^{(\ell-n)} F^{(k-\ell)} + \sum_{\ell=1}^{k-1} \sum_{n=1}^{k-\ell} \stq^{(\ell)} \stq^{(n)} F^{(k-\ell-n)} \eqend{,}
\end{splitequation}
where we used the nilpotency relation~\eqref{proof_obs_stq_nilpot} in the first and the conditions on the $F^{(k)}$~\eqref{proof_obs_fk_cohomconds} in the second step. The sums cancel upon rearranging of the summation indices, and since by assumption $H^1(\brst) = \{0\}$, we have
\begin{equation}
\sum_{\ell=1}^k \stq^{(\ell)} F^{(k-\ell)} = \brst B^{(k)}
\end{equation}
for some local functional $B^{(k)}$, and we then define $F^{(k)} = - B^{(k)}$.

By a similar expansion, we now show that the assumption $H^1(\brst) = \{0\}$ leads to $H^1(\stq) = \{0\}$, and that the existence of a contracting homotopy $\mathsf{h}$ for $\brst$ results in a contracting homotopy $\mathsf{h}_\hbar$ for $\stq$. We consider the expansion~\eqref{proof_obs_f_expansion} for a representative $F$ of an element in $H^1(\stq)$, and obtain as before~\eqref{proof_obs_fk_cohomconds} that $\brst F^{(0)} = 0$. Since $H^1(\brst) = \{0\}$, we must have $F^{(0)} = \brst G^{(0)}$ for some local functional $G^{(0)}$. We then consider instead of $F$ the representative $F' = F - \stq G^{(0)}$ of the same element in $H^1(\stq)$, which is of higher order in $\hbar$. Repeating the process, we obtain that $F = \stq \left( G^{(0)} + \hbar G^{(1)} + \cdots \right)$, and conclude that $H^1(\stq) = \{0\}$. We now set $\mathsf{h}_\hbar \equiv \sum_{k=0}^\infty \hbar^k \mathsf{h}^{(k)}$ and define recursively
\begin{equation}
\label{proof_obs_homotopy_def}
\mathsf{h}^{(0)} = \mathsf{h} \eqend{,} \qquad \mathsf{h}^{(k)} = - \mathsf{h} \sum_{m=1}^k \left[ \stq^{(m)} \mathsf{h}^{(k-m)} + \mathsf{h}^{(k-m)} \stq^{(m)} \right] \qquad (k \geq 1) \eqend{.}
\end{equation}
By induction we show that this satisfies the homotopy condition $\stq \mathsf{h}_\hbar + \mathsf{h}_\hbar \stq = \operatorname{id}$. By assumption, the homotopy condition holds for the classical homotopy $\mathsf{h}$ and the classical BRST differential $\brst$, and at order $\hbar^k$ we have to prove that
\begin{equation}
\label{proof_obs_homotopy_cond_k}
\sum_{m=0}^k \left( \stq^{(m)} \mathsf{h}^{(k-m)} + \mathsf{h}^{(k-m)} \stq^{(m)} \right) = 0 \qquad (k \geq 1) \eqend{.}
\end{equation}
We calculate
\begin{splitequation}
&\brst \mathsf{h}^{(k)} + \mathsf{h}^{(k)} \brst = - \brst \mathsf{h} \sum_{m=1}^k \left[ \stq^{(m)} \mathsf{h}^{(k-m)} + \mathsf{h}^{(k-m)} \stq^{(m)} \right] \\
&\hspace{8em}- \mathsf{h} \sum_{m=1}^k \left[ \stq^{(m)} \mathsf{h}^{(k-m)} + \mathsf{h}^{(k-m)} \stq^{(m)} \right] \brst \\
&\quad= - \sum_{m=1}^k \left[ \stq^{(m)} \mathsf{h}^{(k-m)} + \mathsf{h}^{(k-m)} \stq^{(m)} \right] + \mathsf{h} \sum_{m=1}^k \left[ \brst \mathsf{h}^{(k-m)} \stq^{(m)} - \stq^{(m)} \mathsf{h}^{(k-m)} \brst \right] \\
&\qquad- \mathsf{h} \sum_{m=1}^k \left[ \sum_{\ell=1}^m \stq^{(\ell)} \stq^{(m-\ell)} \mathsf{h}^{(k-m)} - \sum_{\ell=0}^{m-1} \mathsf{h}^{(k-m)} \stq^{(\ell)} \stq^{(m-\ell)} \right] \eqend{,}
\raisetag{2.2em}
\end{splitequation}
where we used~\eqref{proof_obs_stq_nilpot} and the analogous equality for $\stq^{(k)} \brst$. Exchanging the sums over $m$ and $\ell$ in the last line, shifting summation indices and rearranging, it follows that
\begin{splitequation}
\label{proof_obs_homotopy_result}
&\sum_{m=0}^k \left[ \stq^{(m)} \mathsf{h}^{(k-m)} + \mathsf{h}^{(k-m)} \stq^{(m)} \right] = \mathsf{h} \left( \brst \mathsf{h} + \mathsf{h} \brst \right) \stq^{(k)} - \mathsf{h} \stq^{(k)} \left( \brst \mathsf{h} + \mathsf{h} \brst \right) \\
&\hspace{8em}+ \mathsf{h} \sum_{\ell=1}^{k-1} \sum_{m=0}^{k-\ell} \left[ \stq^{(m)} \mathsf{h}^{(k-\ell-m)} + \mathsf{h}^{(k-\ell-m)} \stq^{(m)} \right] \stq^{(\ell)} \\
&\hspace{8em}- \mathsf{h} \sum_{\ell=1}^{k-1} \stq^{(\ell)} \sum_{m=0}^{k-\ell} \left[ \stq^{(m)} \mathsf{h}^{(k-\ell-m)} + \mathsf{h}^{(k-\ell-m)} \stq^{(m)} \right] \eqend{.}
\end{splitequation}
Since $1 \leq k-\ell < k$ for all $\ell$ in the sums in the last two lines, by induction the sums over $m$, which are the homotopy condition~\eqref{proof_obs_homotopy_cond_k} at order $k-\ell$, vanish. Using the classical homotopy condition $\brst \mathsf{h} + \mathsf{h} \brst = \operatorname{id}$, the first two terms also cancel, such that the homotopy condition~\eqref{proof_obs_homotopy_cond_k} also holds at order $k$. Taking both results together, it follows that for any $F \in \mathcal{F}^1$ with $\stq F = 0$ we have $F = \stq G$ for some $G \in \mathcal{F}^0$, and we can take $G = \mathsf{h}_\hbar F$.

In case the classical homotopy condition is not fulfilled when acting on all functionals, but only on a subset $\mathcal{G} \subset \mathcal{F}$, the proof still works as long as $\stq$ maps $\mathcal{G}$ into itself, i.e., $\stq G \in \mathcal{G}$ for $G \in \mathcal{G}$, since in this case the first two terms in~\eqref{proof_obs_homotopy_result} still cancel when acting on elements of $\mathcal{G}$. In particular, this is the case for the analogous construction of $\mathsf{h}$ from $\mathsf{h}_0$ for the free theory if $H^1(\brst_0)$ is not empty but $H^1(\brst)$ is, and if $\brst$ maps the subspace of jet space that does not contain fixed (chosen) representatives of all the elements of $H^1(\brst_0)$ into itself. For the examples given previously, it is known~\cite{piguetsibold1984,barnichetal2000,brandt2003} that all elements of $H^1(\brst_0)$ have a representative that is the product of an undifferentiated ghost with a representative of an element of $H^0(\brst_0)$. The counting operator $N_{u,v}$ referred to in the Remark counts derivatives of $A_\mu$ symmetrized over all indices, differentiated ghosts, some of the antifields and Majorana spinors in the case of Super-Yang--Mills theory, and it is straightforward to check that $\brst$ maps the subspaces of jet space with positive eigenvalues of $N_{u,v}$ into itself. Therefore, in all these cases $\mathsf{h}$ can be constructed from $\mathsf{h}_0$ by the analogue of the recursion~\eqref{proof_obs_homotopy_def}.

We now show that the contact terms $\mathcal{C}$ can be determined from the identities~\eqref{thm_obs_contactdef} and~\eqref{thm_obs_contactdef2} using the contracting homotopy $\mathsf{h}$. For ease of notation, we use the shorthand
\begin{equation}
\label{proof_obs_fcn_def}
F_\mathcal{C}^n \equiv \left. \frac{\partial^n}{\partial \alpha^n} \exp\left[ \alpha F - \mathcal{C}\left( \mathe_\otimes^{\alpha F} \right) \right] \right\rvert_{\alpha = 0} \eqend{,}
\end{equation}
and expanding the identities~\eqref{thm_obs_contactdef} and~\eqref{thm_obs_contactdef2} in powers of $F \in H^0(\stq)$, we obtain
\begin{equation}
\label{proof_obs_contactdef}
\left[ F_\mathcal{C}^n \right]_\hbar = 0 \eqend{,} \qquad \mathcal{C}_{n+1}\left( F^{\otimes n} \otimes \stq G \right) = \begin{cases} 0 & n = 0 \\ - \left[ F_\mathcal{C}^n, G \right]_\hbar & n > 0 \eqend{.} \end{cases}
\end{equation}
For $n = 1$, we have $F_\mathcal{C} = F - \mathcal{C}_1(F)$ and the first identity gives
\begin{equation}
\stq F = \stq \mathcal{C}_1(F) \eqend{,}
\end{equation}
and since $\stq F = 0$ we can choose $\mathcal{C}_1(F) = 0$. We now proceed by induction in $n$. Explicitly, we have
\begin{equation}
\label{proof_obs_fcn_def2}
F_\mathcal{C}^n = \sum_{\ell=1}^n \sum_{k_1 + \cdots + k_\ell = n} \prod_{i=1}^\ell \frac{1}{k_i!} \left[ \delta_{k_i,1} F - \mathcal{C}_{k_i}\left( F^{\otimes k_i} \right) \right] \eqend{,}
\end{equation}
and the first identity that needs to be satisfied for $n > 1$ can be written as
\begin{equation}
\stq \mathcal{C}_n\left( F^{\otimes n} \right) = \left[ n! \sum_{\ell=2}^n \sum_{k_1 + \cdots + k_\ell = n} \prod_{i=1}^\ell \frac{1}{k_i!} \left[ \delta_{k_i,1} F - \mathcal{C}_{k_i}\left( F^{\otimes k_i} \right) \right] \right]_\hbar \equiv K_n \eqend{.}
\end{equation}
If the right-hand side is $\stq$-closed, $\stq K_n = 0$, we can use the contracting homotopy $\mathsf{h}_\hbar$ constructed above to define $\mathcal{C}_n\left( F^{\otimes n} \right) \equiv \mathsf{h}_\hbar K_n$, which shows (since $\mathsf{h}_\hbar$ is linear) that the $\mathcal{C}_n$ are multilinear maps. To show that $\stq K_n = 0$, we first rewrite the $L_\infty$ identities~\eqref{proof_linfty_relation} for bosonic $F_j$ by polarisation as
\begin{equation}
\label{proof_obs_linftyrels_fi}
\sum_{I \cup J = \{1,\ldots,\ell\}, J \neq \emptyset} \left[ \prod_{i \in I} F_i, \left[ \prod_{j \in J} F_j \right]_\hbar \right]_\hbar = 0 \eqend{,}
\end{equation}
that is
\begin{equation}
\label{proof_obs_linftyrels_fi2}
\stq \left[ \prod_{j=1}^\ell F_j \right]_\hbar = - \sum_{I \cup J = \{1,\ldots,\ell\}, I \neq \emptyset \neq J} \left[ \prod_{i \in I} F_i, \left[ \prod_{j \in J} F_j \right]_\hbar \right]_\hbar \eqend{.}
\end{equation}
We then calculate
\begin{splitequation}
&\stq K_n = - n! \sum_{\ell=2}^n \sum_{k_1 + \cdots + k_\ell = n} \sum_{I \cup J = \{1,\ldots,\ell\}, I \neq \emptyset \neq J} \\
&\quad\times \left[ \prod_{i \in I} \frac{1}{k_i!} \left[ \delta_{k_i,1} F - \mathcal{C}_{k_i}\left( F^{\otimes k_i} \right) \right], \left[ \prod_{j \in J} \frac{1}{k_j!} \left[ \delta_{k_j,1} F - \mathcal{C}_{k_j}\left( F^{\otimes k_j} \right) \right] \right]_\hbar \right]_\hbar \eqend{.}
\end{splitequation}
For each choice of $I$ and $J$ we now set $n' = \sum_{j \in J} k_j$, and since $I \neq \emptyset$ we have $n' < n$. Rearranging sums and using the short notation~\eqref{proof_obs_fcn_def2}, it follows that
\begin{equation}
\stq K_n = - n! \sum_{\ell=2}^n \sum_{\substack{I \cup J = \{1,\ldots,\ell\} \\ I \neq \emptyset \neq J}} \sum_{k_i\colon i\in I} \left[ \prod_{i \in I} \frac{1}{k_i!} \left[ \delta_{k_i,1} F - \mathcal{C}_{k_i}\left( F^{\otimes k_i} \right) \right], \left[ F_\mathcal{C}^{n'} \right]_\hbar \right]_\hbar \eqend{,}
\end{equation}
and by induction $\left[ F_\mathcal{C}^{n'} \right]_\hbar = 0$.

This defines the contact terms up to the addition of a $\stq$-exact term at each order, and we show that one can use this freedom to also satisfy the second identity of~\eqref{proof_obs_contactdef}. For $n = 0$, this reads $\mathcal{C}_1(\stq G) = 0$ which is true since we have taken $\mathcal{C}_1(F) = 0$ for all representatives $F$ of an element in $H^0(\stq)$, in this case the zero element. Proceeding by induction, we may then assume that the identity holds for all $k < n$. By polarisation, from the first identity~\eqref{thm_obs_contactdef} we obtain to first order in $G$
\begin{equation}
\left[ \exp\left( F - \mathcal{C}\left( \mathe_\otimes^F \right) \right), \mathcal{C}\left( \mathe_\otimes^F \otimes \stq G \right) \right]_\hbar = \left[ \exp\left( F - \mathcal{C}\left( \mathe_\otimes^F \right) \right), \stq G \right]_\hbar \eqend{,}
\end{equation}
which expanding in powers of $F$ and using the shorthand~\eqref{proof_obs_fcn_def} can be written as
\begin{equation}
\sum_{k=0}^n \frac{n!}{k! (n-k)!} \left[ F_{\mathcal{C}}^k, \mathcal{C}_{n-k+1}\left( F^{\otimes (n-k)} \otimes \stq G \right) \right]_\hbar = \left[ F_{\mathcal{C}}^n, \stq G \right]_\hbar \eqend{.}
\end{equation}
Separating the term with $k = 0$ and using that the second identity of~\eqref{proof_obs_contactdef} holds for all $k < n$ by induction to replace terms on the right-hand side, we obtain
\begin{splitequation}
\label{proof_obs_stqcfstqg_rel}
\stq \mathcal{C}_{n+1}\left( F^{\otimes n} \otimes \stq G \right) &= \left[ F_{\mathcal{C}}^n, \stq G \right]_\hbar + \sum_{k=1}^{n-1} \frac{n!}{k! (n-k)!} \left[ F_{\mathcal{C}}^k, \left[ F_\mathcal{C}^{n-k}, G \right]_\hbar \right]_\hbar \\
&= - \stq \left[ F_\mathcal{C}^n, G \right]_\hbar + \sum_{k=0}^n \frac{n!}{k! (n-k)!} \left[ F_{\mathcal{C}}^k, \left[ F_\mathcal{C}^{n-k}, G \right]_\hbar \right]_\hbar \eqend{.}
\end{splitequation}
For one fermionic $G$, one obtains from the $L_\infty$ relations~\eqref{proof_obs_linftyrels_fi} by polarisation that
\begin{equation}
\label{proof_obs_linftyrels_fi_g}
\sum_{I \cup J = \{1,\ldots,\ell\}} \left[ \prod_{i \in I} F_i, \left[ G, \prod_{j \in J} F_j \right]_\hbar \right]_\hbar = \sum_{I \cup J = \{1,\ldots,\ell\}} \left[ G, \prod_{i \in I} F_i, \left[ \prod_{j \in J} F_j \right]_\hbar \right]_\hbar \eqend{,}
\end{equation}
and replacing the $F_{\mathcal{C}}^k$ in the right-hand side of~\eqref{proof_obs_stqcfstqg_rel} by their expansion~\eqref{proof_obs_fcn_def2}, using the relation~\eqref{proof_obs_linftyrels_fi_g} to shift $G$ to the first quantum bracket, and using the shorthand~\eqref{proof_obs_fcn_def2} again we obtain
\begin{equation}
\stq \mathcal{C}_{n+1}\left( F^{\otimes n} \otimes \stq G \right) = - \stq \left[ F_\mathcal{C}^n, G \right]_\hbar + \sum_{k=0}^n \frac{n!}{k! (n-k)!} \left[ G, F_{\mathcal{C}}^k, \left[ F_\mathcal{C}^{n-k} \right]_\hbar \right]_\hbar \eqend{.}
\end{equation}
By the first identity of~\eqref{proof_obs_contactdef}, we have $\left[ F_\mathcal{C}^{n-k} \right]_\hbar = 0$, and applying the contracting homotopy $\mathsf{h}_\hbar$ we obtain
\begin{splitequation}
\mathcal{C}_{n+1}\left( F^{\otimes n} \otimes \stq G \right) &= \mathsf{h}_\hbar \stq \mathcal{C}_{n+1}\left( F^{\otimes n} \otimes \stq G \right) = - \mathsf{h}_\hbar \stq \left[ F_\mathcal{C}^n, G \right]_\hbar \\
&= - \left[ F_\mathcal{C}^n, G \right]_\hbar + \stq \mathsf{h}_\hbar \left[ F_\mathcal{C}^n, G \right]_\hbar \eqend{,}
\end{splitequation}
such that with the redefinition
\begin{equation}
\mathcal{C}_{n+1}\left( F^{\otimes n} \otimes \stq G \right) \to \mathcal{C}_{n+1}\left( F^{\otimes n} \otimes \stq G \right) - \stq \mathsf{h}_\hbar \left[ F_\mathcal{C}^n, G \right]_\hbar
\end{equation}
also the second identity of~\eqref{proof_obs_contactdef} can be fulfilled.

Finally, we show that the interacting time-ordered products including the contact terms are $\st$-closed and, up to $\st$-exact terms, independent of the choice of representative $F$ of an element in $H^0(\stq)$. Using the form~\eqref{anomward_quantumbracket} of the interacting anomalous Ward identity, we calculate
\begin{splitequation}
&\st \mathcal{T}_L\left[ \exp_\otimes\left( \frac{\mathi}{\hbar} F - \frac{\mathi}{\hbar} \mathcal{C}\left( \mathe_\otimes^F \right) \right) \right] \\
&\quad= \frac{\mathi}{\hbar} \mathcal{T}_L\left[ \left[ \exp\left( F - \mathcal{C}\left( \mathe_\otimes^F \right) \right) \right]_\hbar \otimes \exp_\otimes\left( \frac{\mathi}{\hbar} F - \frac{\mathi}{\hbar} \mathcal{C}\left( \mathe_\otimes^F \right) \right) \right] \eqend{,}
\end{splitequation}
which vanishes by the identity~\eqref{thm_obs_contactdef}. Consider then a local functional $G \in \mathcal{F}^{-1}$ and the representative $F' = F + \alpha \stq G$ of the same element in $H^0(\stq)$. We then want to show that
\begin{splitequation}
&\frac{\partial}{\partial \alpha} \mathcal{T}_L\left[ \exp_\otimes\left[ \frac{\mathi}{\hbar} F' - \frac{\mathi}{\hbar} \mathcal{C}\left( \mathe_\otimes^{F'} \right) \right] \right] \\
&\quad= \frac{\mathi}{\hbar} \mathcal{T}_L\left[ \exp_\otimes\left[ \frac{\mathi}{\hbar} F' - \frac{\mathi}{\hbar} \mathcal{C}\left( \mathe_\otimes^{F'} \right) \right] \otimes \left[ \stq G - \mathcal{C}\left( \mathe_\otimes^{F'} \otimes \stq G \right) \right] \right]
\end{splitequation}
is $\st$-exact. Using the identity~\eqref{thm_obs_contactdef2}, it follows that
\begin{splitequation}
\mathcal{C}\left[ \mathe_\otimes^{F'} \otimes \stq G \right] = \left[ 1 - \exp\left[ F' - \mathcal{C}\left( \mathe_\otimes^{F'} \right) \right], G \right]_\hbar \eqend{,}
\end{splitequation}
and from this
\begin{splitequation}
&\frac{\partial}{\partial \alpha} \mathcal{T}_L\left[ \exp_\otimes\left[ \frac{\mathi}{\hbar} F' - \frac{\mathi}{\hbar} \mathcal{C}\left( \mathe_\otimes^{F'} \right) \right] \right] \\
&\quad= \frac{\mathi}{\hbar} \mathcal{T}_L\left[ \exp_\otimes\left[ \frac{\mathi}{\hbar} F' - \frac{\mathi}{\hbar} \mathcal{C}\left( \mathe_\otimes^{F'} \right) \right] \otimes \left[ \exp\left( F' - \mathcal{C}\left[ \mathe_\otimes^{F'} \right] \right), G \right]_\hbar \right] \eqend{.}
\end{splitequation}
Multiplying the anomalous interacting Ward identity in the form~\eqref{proof_linfty_anomward2} by $\left( \mathi/\hbar \right)^k/k!$ and summing over $k$, we obtain
\begin{splitequation}
\st \mathcal{T}_L\left[ G \otimes \exp_\otimes\left( \frac{\mathi}{\hbar} F \right) \right] &= \mathcal{T}_L\left[ \left[ \mathe^F, G \right]_\hbar \otimes \exp_\otimes\left( \frac{\mathi}{\hbar} F \right) \right] \\
&\quad+ \frac{\mathi}{\hbar} \mathcal{T}_L\left[ \left[ \mathe^F \right]_\hbar \otimes G \otimes \exp_\otimes\left( \frac{\mathi}{\hbar} F \right) \right] \eqend{.}
\end{splitequation}
Substituting in this equation $F \to F' - \mathcal{C}\left[ \exp_\otimes\left( F' \right) \right]$, the quantum bracket in the second term vanishes by the identity~\eqref{thm_obs_contactdef} for the contact terms, while the first term is exactly the right-hand side of the previous equation, such that
\begin{equation}
\frac{\partial}{\partial \alpha} \mathcal{T}_L\left[ \exp_\otimes\left[ \frac{\mathi}{\hbar} F' - \frac{\mathi}{\hbar} \mathcal{C}\left( \mathe_\otimes^{F'} \right) \right] \right] = \frac{\mathi}{\hbar} \st \mathcal{T}_L\left[ G \otimes \exp_\otimes\left[ \frac{\mathi}{\hbar} F' - \frac{\mathi}{\hbar} \mathcal{C}\left( \mathe_\otimes^{F'} \right) \right] \right] \eqend{.}
\end{equation}
\end{proof}
\begin{remark*}
Expectation values in a BRST-invariant state $\omega$, a state fulfilling
\begin{equation}
\omega\left( \st \mathcal{T}_L\left( F_1 \otimes \cdots \otimes F_k \right) \right) = 0
\end{equation}
for arbitrary $F_i$, are thus independent of the choice of representative for the observables in $H^0(\stq)$. In general, the construction of such a state is quite difficult, in particular if one wants to study global (infrared) issues. However, if one is only interested in local statements it is possible to do an explicit construction of $\omega$ for the free theory and then proceed to the interacting theory by a deformation argument. The details of such a construction are described in~\cite{duetschfredenhagen1999,junkerschrohe2002,fewsterpfenning2003,hollands2008} for spacetimes with compact Cauchy surfaces, but can be generalised to general spacetimes for the free theory (with a certain cohomological condition on the spacetime region under consideration)~\cite{hollands_rev}, and then to the interacting theory using the above definition of $\Delta Q$.
\end{remark*}

\section{Summary and Outlook}
\label{sec_summary}

We have shown that any derivation on the algebra of free quantum fields gives rise to an anomalous Ward identity when acting on time-ordered products. The classical, i.e., $\hbar$-independent part of these identities can be explicitly calculated, and we have given an explicit formula for case that the derivation is inner (its action is given by the graded commutator with an element of the algebra), and the classical limit of its generator is at most quadratic in fields. In this case, we have further shown that all except the first two classical terms vanish. This applies in particular to all symmetries of the free theory, whose Noether charge is quadratic in fields. Our main interest was the application to gauge theories, which can be treated in the BV--BRST formalism even if the algebra of gauge transformations only closes on-shell, as happens for supersymmetric theories without auxiliary fields. The introduction of antifields (sources for the BRST transformation of the fields) is strictly necessary in this case to obtain an off-shell nilpotent BRST differential. The cohomologies of this differential (and its restriction to the free theory) are highly relevant in the construction of the quantum theory; in particular if a certain cohomology is empty the anomaly in the anomalous Ward identities can be removed by a field redefinition (change of renormalisation prescription). Generalising previous works, we have explicitly determined the classical part of the anomalous Ward identities for the free BRST differential, and shown that also in this case all except the first two classical terms vanish. To obtain an anomalous Ward identity for the full interacting BRST differential, previous authors have constructed the BRST differential as the graded commutator with the interacting BRST Noether charge, which exists and fulfils the required properties on-shell provided another cohomology class is empty. However, the classical BRST symmetry is an off-shell symmetry (even for open gauge algebras in the BV--BRST formulation), and we have suceeded in constructing an element of the algebra whose commutator gives the action of the interacting BRST differential (i.e., the required interacting anomalous Ward identity) even off-shell, and without any further restriction on cohomology classes.

We have then analysed the structure of the interacting anomalous Ward identities further. Combining the classical terms with the anomalous ones, we obtained the quantum BRST differential $\stq = \brst + \bigo{\hbar}$ and the quantum antibracket $(\cdot,\cdot)_\hbar = (\cdot,\cdot) + \bigo{\hbar}$. While classical observables are (representatives of) elements of the cohomology of the classical BRST differential $\brst$ at zero ghost number $H^0(\brst)$, observables in the quantum theory are (representatives of) elements of $H^0(\stq)$. We have shown that if the classical cohomology at ghost number $1$ is empty, $H^1(\brst) = \{0\}$, each element of $H^0(\brst)$ can be extended to an element of $H^0(\stq)$, and thus there is a quantum observable corresponding to each classical one. From the construction, it follows that the extension involves higher powers of $\hbar$, and thus corresponds to genuine quantum corrections. The quantum BRST differential, the quantum antibracket and the higher anomalous terms form an $L_\infty$ algebra, which in particular ensures that the quantum BRST differential is nilpotent (such that one can examine its cohomology classes) and that the quantum antibracket is a map between cohomology classes fulfilling the Jacobi identity. Furthermore, the $L_\infty$ structure guarantees (again if $H^1(\brst) = \{0\}$) that one can find contact terms (local functionals supported on the total diagonal) such that the interacting time-ordered products are independent of the choice of representative for the quantum observable, elements of $H^0(\stq)$, up to a BRST-exact term. Since the expectation value of BRST-exact terms vanishes in a (physical) BRST-invariant state, this shows that physical results are independent of the choice of representatives, i.e., the anomalous Ward identities become proper Ward identities for expectation values in physical states, including all contact terms.

In a next step, one would also like to show that expectation values in physical states are independent of the choice of gauge fixing. While this follows formally because the gauge-fixing term (including the contribution from ghosts) is a BRST-exact term in the classical action, a fully rigorous proof is still missing. The issue is complicated by the fact that already the free quantum theory depends on the gauge-fixing, and one would thus have to show that perturbative agreement holds for BRST-exact terms such that one can freely shift the gauge-fixing terms from the free action to the interaction. However, the very definition of the BRST differential, and thus of BRST-exact terms depends on the action (i.e., acting on antifields one obtains the gauge-fixed equations of motion), and it would be necessary to restrict to gauge-fixing terms such that this is unproblematic, e.g., to antifield-independent ones for Yang--Mills theory. Independence of the expectation values of interacting time-ordered products in physical states (up to possible contact terms) then seems to follow quite straightforwardly from the proper Ward identities we have derived in this article, by the fact that a BRST-exact term is cohomologically equivalent to zero. However, also the definition of the quantum BRST differential and the BRST-invariant state depend on the gauge fixing, and one would have to show that all these can be consistently deformed.

A related issue is the existence of the limit $g \to 1$. On the algebraic level, one shows~\cite{brunettifredenhagen2000} that for interacting time-ordered products $T$ supported within a causally closed region $R$, a change of the cutoff function $g(x) \to g'(x)$ outside of $R$ (i.e., such that $\supp (g'-g) \cap R = \emptyset$) results in the conjugation $T \to T' = V \star_\hbar T \star_\hbar V^{-1}$ with a unitary operator $V(g,g')$, due to the factorisation property of the interacting time-ordered products. Therefore, the algebraic relations between the interacting time-ordered products are unaffected, and the limit $g \to 1$ can be realised as an inductive limit, called the algebraic adiabatic limit. To obtain the algebraic adiabatic limit of the constructions presented in this work, one would need to show their compatibility with this conjugation, i.e., show that the ($g$-dependent) interacting nilpotent BRST differential $\st$ and consequently the quantum brackets $[\cdot]_\hbar$ and contact terms $\mathcal{C}_n$ transform appropriately under a change of the cutoff function. However, what is needed for physical applications is the weak adiabatic limit, i.e., the limit $g \to 1$ on the level of expectation values of interacting time-ordered products in physical, BRST-invariant states. This is much more complicated problem, also because the construction of an interacting BRST-invariant state depends on the definition of the interacting BRST differential, and thus on the cutoff function. Due to infrared issues, it becomes even more complicated when massless particles are involved, which is the case for gauge theories (see, e.g., the work~\cite{duch2018} for a recent construction in Minkowski spacetime, and references therein).

Lastly, we would like to apply the general theory to the case of observables in perturbative quantum gravity, of the kind considered in~\cite{brunettietal2016,froeb2018,froeblima2018}. While it is well known that gravity is non-renormalisable as a quantum theory, one can still treat it in the sense of an effective field theory, and the perturbative expansion is well defined up to any fixed order. These observables are non-local, and it is not clear whether they can be renormalised, even at first order in perturbation theory. A possible way to accomplish this is to find a special (non-linear) gauge in which they become local, and where then the general theorems of perturbative AQFT apply. At linear order, this can be done~\cite{froeblima2018}, but the extension to the fully interacting theory is more difficult. This special gauge condition must be imposed strictly, i.e., inside time-ordered products, which is accomplished by obtaining it as equation of motion of the auxiliary field. In turn, this can be obtained as the BRST transformation of the corresponding antifield, and we thus require the corresponding anomalous terms to vanish. By the results of this article, this follows if perturbative agreement holds for the free equation of motion of the auxiliary field, i.e., the linear part of the special gauge condition.

\begin{acknowledgements}
It is a pleasure to thank Chris Fewster, Atsushi Higuchi, Stefan Hollands, Kasia Rejzner, Mojtaba Taslimi Tehrani and Jochen Zahn for discussions on (algebraic) quantum field theory, Igor Khavkine for comments on $L_\infty$ algebras and a critical reading of the manuscript, Pawe{\l} Duch for pointing out a mistake and a simplification in the proof of Theorem~\ref{thm_brstward}, and the anonymous referee for a careful reading of the manuscript and for pointing out various mistakes and typos. This work is part of a project that has received funding from the European Union's Horizon 2020 research and innovation programme under the Marie Sk{\l}odowska-Curie grant agreement No. 702750 ``QLO-QG''.
\end{acknowledgements}

\bibliography{literature}

\end{document}